\documentclass[A4,12pt]{article}

\usepackage{fullpage}
\usepackage{color}
\usepackage{url}
\usepackage{graphicx}
\usepackage{epstopdf}
\usepackage{xspace}
\usepackage{amsmath}
\usepackage{amssymb}
\usepackage{amsthm}
\usepackage{algorithmicx}
\usepackage[noend]{algpseudocode}

\usepackage{pgfplots}
\pgfplotsset{compat=newest}

\usepackage{fancyvrb}
\usepackage{ifthen}
\usepackage{tikz}

\newcommand{\FINDATE}{31.08.2014} 
\newcommand{\DNUM}{D2.1}
\newcommand{\DNAME}{Models for energy consumption of data structures and algorithms}

\newcommand{\ema}[1]{\ensuremath{#1}\xspace}
\newcommand{\ie}{\textit{i.e.}\xspace}

\newcommand{\cw}{\ema{\mathit{cw}}}
\newcommand{\pw}{\ema{\mathit{pw}}}
\newcommand{\thr}{\ema{\mathcal{T}}}
\newcommand{\nth}{\ema{n}}

\newcommand{\nfr}{\ema{F}}
\newcommand{\nalg}{\ema{A}}
\newcommand{\ntht}{\ema{N}}

\newcommand{\powi}{\ema{P}}
\newcommand{\pstati}{\ema{\powi_{\mathit{stat}}}}
\newcommand{\pdyni}{\ema{\powi_{\mathit{dyn}}}}
\newcommand{\pacti}{\ema{\powi_{\mathit{active}}}}

\newcommand{\expo}[2]{\ema{#1^{\ifthenelse{\equal{#2}{}}{}{(#2)}}}}

\newcommand{\pow}[1]{\expo{\powi}{#1}}
\newcommand{\pstat}[1]{\expo{\pstati}{#1}}
\newcommand{\pdyn}[1]{\expo{\pdyni}{#1}}
\newcommand{\pact}[1]{\expo{\pacti}{#1}}

\newcommand{\prog}{\ensuremath{prog}\xspace}
\newcommand{\pin}{\ensuremath{\mathit{pinning}}}

\newcommand{\cas}{\textit{Compare-and-Swap}\xspace}

\newcommand{\ghz}[1]{\ema{#1\;\text{GHz}}}
\newcommand{\ps}{parallel section\xspace}
\newcommand{\rl}{retry-loop\xspace}
\newcommand{\itemx}[1]{}

\newcommand{\tcs}{\ema{t_{\text{RL}}}}
\newcommand{\tps}{\ema{t_{\text{PS}}}}
\newcommand{\so}{\ema{\leadsto}}
\newcommand{\facf}{\ema{\lambda}}
\newcommand{\freq}{\ema{f}}

\graphicspath{{figs-chalmers/},{figs-chalmers/mandelbrot/},{figs-chalmers/results/}}

\newcommand{\remind}[1]{} 
\newcommand{\leaveout}[1]{}
\newcommand{\comment}[2]{}{}

\newcommand{\op}[1]{{\textsf{\textit{#1}}}}
\newcommand{\var}[1]{{\textsf{#1}}}
\newcommand{\code}[1]{{\textsf{\tt {#1}}}}

\algdef{SE}[DOWHILE]{Do}{doWhile}{\algorithmicdo}[1]{\algorithmicwhile\ #1}%

\algnewcommand\EMPTY{\textbf{EMPTY}}

\algblockdefx[NAME]{Start}{End}%
	[1]{\textbf{Struct} #1:}%
	{EndStruct}
\algnotext[NAME]{End}

\algblockdefx[Name]{START}{END}%
	[1]{#1}%
	{Ending}
\algnotext[Name]{END}

\newtheorem{theorem}{Theorem}[section]

\newtheorem{lemma}[theorem]{Lemma}

\begin{document}

\thispagestyle{empty}

\vspace{-3cm}
\begin{center}
\textbf{SEVENTH FRAMEWORK PROGRAMME}\\
\textbf{THEME ICT-2013.3.4}\\
Advanced Computing, Embedded and Control Systems
\end{center}
\bigskip

\begin{center}
\includegraphics[width=\textwidth]{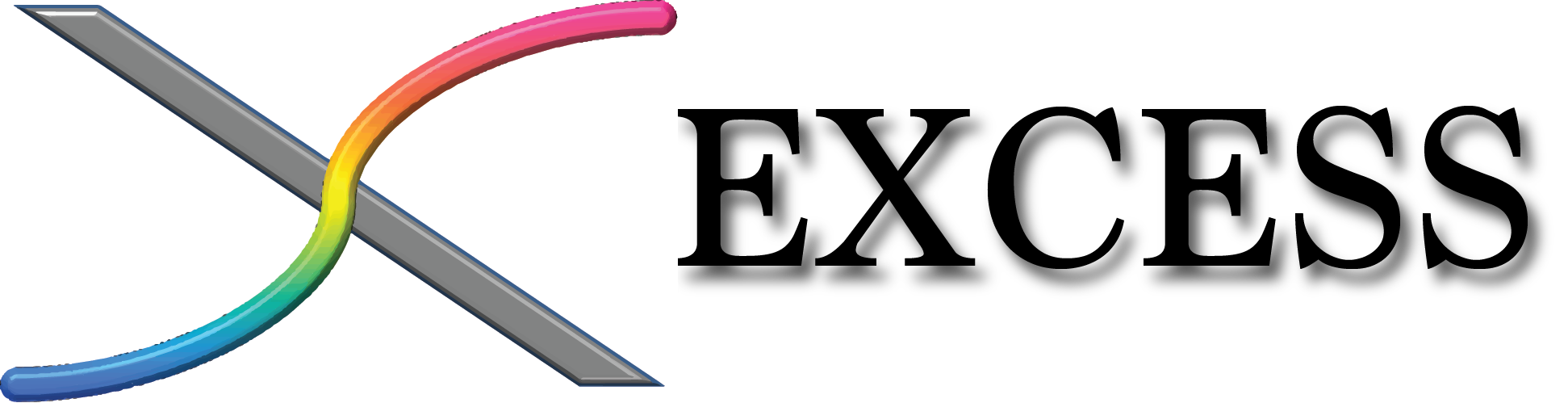}
\end{center}
\bigskip

\begin{center}
Execution Models for Energy-Efficient
Computing Systems\\
Project ID: 611183
\end{center}
\bigskip

\begin{center}
\Large
\textbf{\DNUM} \\
\textbf{\DNAME}
\end{center}
\bigskip

\begin{center}
\large
Phuong Ha, Vi Tran, Ibrahim Umar, \\
Philippas Tsigas, Anders Gidenstam, Paul Renaud-Goud,\\
Ivan Walulya, Aras Atalar 
\end{center}

\vfill

\begin{center}
\includegraphics[width=3cm]{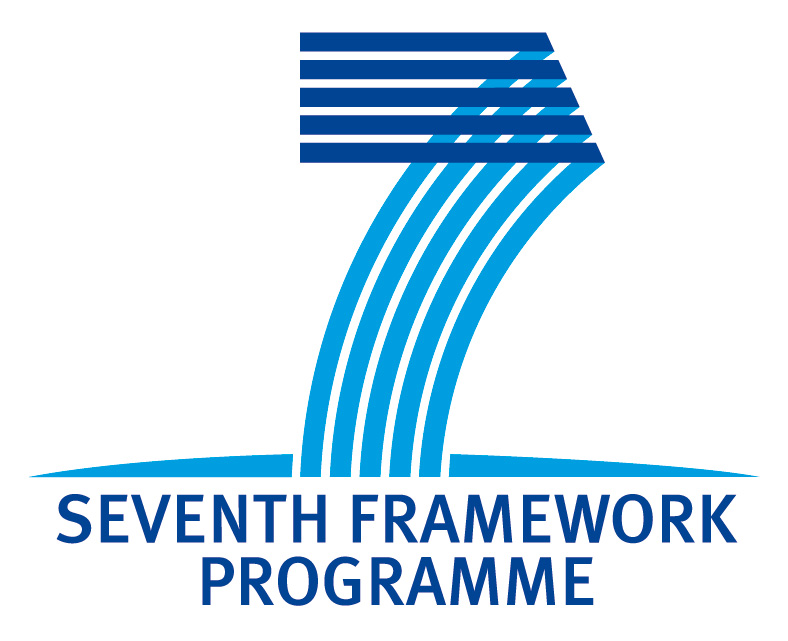}\\
Date of preparation (latest version): \FINDATE \\
Copyright\copyright\ 2013 -- 2016 The EXCESS Consortium \\
\hrulefill \\
The opinions of the authors expressed in this document do not
necessarily reflect the official opinion of EXCESS partners or of
the European Commission.
\end{center}

\newpage

\section*{DOCUMENT INFORMATION}

\vspace{1cm}

\begin{center}
\begin{tabular}{ll}
\textbf{Deliverable Number} & \DNUM \\
\textbf{Deliverable Name} & \DNAME \\
\textbf{Authors}

& Phuong Ha \\
& Vi Tran \\
& Ibrahim Umar \\
& Philippas Tsigas \\
& Anders Gidenstam \\
& Paul Renaud-Goud \\
& Ivan Walulya \\
& Aras Atalar \\

\textbf{Responsible Author} & Phuong Ha\\
& e-mail: \url{phuong.hoai.ha@uit.no} \\
& Phone: +47 776 44032 \\
\textbf{Keywords} & High Performance Computing; \\
& Energy Efficiency \\
\textbf{WP/Task} & WP2/Task 2.1 \\
\textbf{Nature} & R \\
\textbf{Dissemination Level} & PU \\
\textbf{Planned Date} & 01.09.2014 \\
\textbf{Final Version Date} & \FINDATE \\
\textbf{Reviewed by} & Christoph Kessler (LIU) \\
										& David Moloney (Movidius) \\
\textbf{MGT Board Approval} & YES\\
\end{tabular}
\end{center}

\newpage

\section*{DOCUMENT HISTORY}

\vspace{1cm}

\begin{center}
\begin{tabular}{llll}
\textbf{Partner} &
\textbf{Date} &
\textbf{Comment} &
\textbf{Version} \\
UiT (P.\ Ha) & 11.06.2014 & Deliverable skeleton & 0.1 \\
CTH (A. Gidenstam) & 22.07.2014 & Added outline for concurrent data structures & 0.2\\
UiT (V.\ Tran) & 02.08.2014 & Added sections 1.1, 1.3, 2.2, 3.1, 4.2   & 0.3 \\
UiT (I.\ Umar, P.\ Ha) & 05.08.2014 & Revised D2.1 and added sections 1.2.1, 3.4, 4.1.1 & 0.4\\
CTH(A. Gidenstam) & 05.08.2014 & Added sections 1.2, 1.5, 2,1, 2.3, 3.2, 3.3, 3.5 & 0.5 \\
UiT (P.\ Ha) & 11.08.2014 & Added Executive summary and sections 1.6 & 0.6 \\
 						& & Incorporated Chalmers and UiT inputs & \\
CTH, UiT & 29.08.2014 & Integration of internal review comments  & 0.7 \\
CTH, UiT & 30.08.2014 & Final version  & 1.0 \\

\end{tabular}
\end{center}

\newpage

\begin{abstract}


This deliverable reports our early energy models for data structures and algorithms based on both micro-benchmarks and concurrent algorithms. It reports the early results of Task 2.1 on investigating and modeling the trade-off between energy and performance in concurrent data structures and algorithms, which forms the basis for the whole work package 2 (WP2). The work has been conducted on the two main EXCESS platforms: (1) Intel platform with recent Intel multi-core CPUs and (2) Movidius embedded platform.

\end{abstract}

\newpage

\section*{Executive Summary}
Computing technology is currently at the beginning of the disruptive transition from petascale to exascale computing (2010 –- 2020), posing a great challenge on energy efficiency. High performance computing (HPC) in 2020 will be characterized by data-centric workloads that, unlike those in traditional sequential/parallel computing, are comprised of big, divergent, fast and complex data. In order to address energy challenges in HPC, the new data must be organized and accessed in an energy-efficient manner through novel fundamental data structures and algorithms that strive for the energy limit. Moreover, the general application- and technology-trend indicates finer-grained execution (i.e. smaller chunks of work per compute core) and more frequent communication and synchronization between cores and uncore components (e.g. memory) in HPC applications. Therefore, not only concurrent data structures and memory access algorithms but also synchronization is essential to optimize the energy consumption of HPC applications. However, previous concurrent data structures, memory access algorithms and synchronization algorithms were designed without energy consumption in mind. The design of energy-efficient fundamental concurrent data structures and algorithms for inter-process communication in HPC  remains a largely unexplored area and requires significant efforts to be successful.

Work package 2 (WP2) aims to develop interfaces and libraries for energy-efficient inter-process communication and data sharing on the new EXCESS platforms integrating Movidius embedded processors. In order to set the stage for these tasks, WP2 needs to investigate and model the trade-offs between energy consumption and performance of data structures and algorithms for inter-process communication, which is Task 2.1. The energy models are developed in close cooperation with WP1 to ensure that they will be compatible with the energy modeling method of WP1.

The early result of Task 2.1 (PM1 - PM36) on investigating and modeling the trade-off between energy and performance in concurrent data structures and algorithms, as available by project month 12, are summarized in this report. The main contributions are the following:
\begin{itemize} 
\item An improved and extended energy model for the CPU-based platform based on the model presented in EXCESS D1.1~\cite{D1.1}. This model decomposes the power into static, active and dynamic power, while classifying CPU-based platform components into three groups: CPU, main memory and “uncore” (e.g. shared cache, IMC, PCU, HA, etc.) (cf. Sec. \ref{sec:cpu-micro}). The experiment results confirm that static power is constant while active power depends on the frequency, the number of socket and not on the operations. Dynamic power is decomposed into dynamic CPU, dynamic memory and dynamic uncore power. Dynamic CPU power depends on the frequency and the operation type. It also shows almost linear behaviors to the number of threads. Dynamic power of memory and uncore components relate to the locality and bandwidth requirement of the implementations.

\item A new power model for the Movidius Myriad platform that is able to predict power consumption of our micro-benchmarks with $\pm 4\%$ margin of measured power consumption on the real platform (cf. Sec. \ref{sec:movidius_energy}). The new power model confirms the experimental power analysis of concurrent data structures such as concurrent queues: the dynamic power consumption is proportional to the number of SHAVE (Streaming Hybrid Architecture Vector Engine) processors used.

\item A case study on how to choose the most suitable implementations for a multi-variant
  shared data structure in a certain application and context; and the prediction of the
  energy efficiency of different queue implementations through two metrics, namely
  throughput and power (cf. Sec. \ref{sec:concurrent-data-structures-CPU}). The case study
  shows that the energy-efficiency is mainly ruled by the contention on the queue, which
  impacts both throughput and memory power dissipation.


\item Implementation and evaluation of several different concurrent queue designs for the Myriad1 platform using three synchronizations primitives: mutex, message passing over shared variables and SHAVE FIFOs (a set of registers accessed in a FIFO pattern)(cf. Sec. \ref{sec:concurrent-data-structures-Movidius}). The valuations are performed on three metrics: execution time, power consumption and energy per operation. In terms of execution time, the implementation using mutex with two locks  is the fastest and most scalable since it provides maximum concurrency. In terms of power, SHAVE FIFOs communication method is the most energy efficient. In terms of energy per operation, SHAVE FIFO implementation also consume the least energy. 

\item Investigation of the energy consumption and performance of concurrent data structures such as concurrent search trees (cf. Sec. \ref{sec:DeltaTree_description}). Based on our investigation, we have developed new locality-aware concurrent search trees called $\Delta$Trees that are up to 140\% faster and 80\% more energy efficient than the state-of-the-art (cf. Sec. \ref{sec:perval_main} and Sec.  \ref{sec:energyeval}).

\end{itemize}






\newpage

\tableofcontents

\newpage


\section{Introduction}
\subsection{Purpose}
In order to address energy challenges in HPC and embedded computing, data must be organized and accessed in an energy-efficient manner through novel fundamental data structures and algorithms that strive for the energy limit. Due to more frequent communication and synchronization between cores and memory components in HPC and embedded computing, not only concurrent data structures and memory access algorithms but also synchronization is essential to optimize the energy consumption. However, previous concurrent data structures, memory access algorithms and synchronization algorithms were designed without considering energy consumption. Although there are existing studies on the energy utilization of concurrent data structures demonstrating non-intuitive results on energy consumption, the design of energy-efficient fundamental concurrent data structures and algorithms for inter-process communication in HPC and embedded computing is not yet widely explored and becomes an challenging and interesting research direction.

EXCESS aims to investigate the trade-offs between energy consumption and performance of concurrent data structures and algorithms as well as inter-process communication in HPC  and embedded computing. By analyzing the non-intuitive results, EXCESS devises a comprehensive model for energy consumption of concurrent data structures and algorithms for inter-process communication, especially in the presence of component composition. The new energy-efficient technology will be delivered through novel execution models for the energy-efficient computing paradigm, which consist of complete energy-aware software stacks (including energy-aware component models, programming models, libraries/algorithms and runtimes) and configurable energy-aware simulation systems for future energy-efficient architectures.

The goal of Work package 2 (WP2) is to develop interfaces and libraries for inter-process communication and data sharing on EXCESS new platforms integrating Movidius embedded processors, along with investigating and modeling the trade-offs between energy consumption and performance of data structures and algorithms for inter-process communication. WP2 also concerns supporting energy-efficient massive parallelism through scalable concurrent data structures and algorithms that strive for the energy limit, and minimizing inter-component communication through locality- and heterogeneity-aware data structures and algorithms.

The first objective of WP2 (Task 2.1) is to investigate and model the trade-off between energy and performance in concurrent data structures and algorithms. In order to model energy and performance, the analysis is conducted for non-intuitive results and their trade-offs to devise comprehensive models for energy consumption of concurrent data structures and algorithms of inter-process communication. The energy models are developed in close cooperation with WP1, ensuring that they will be compatible with the modeling method of WP1.  

This report summarizes the early results of Task 2.1 on investigating and modeling the consumed energy of concurrent data structures and algorithms. The work of Task 2.1 forms the theoretical basis for the whole work package.

\subsection{Concurrent Data Structures and Algorithms for Inter-process Communication}\label{sec:int-ds}

\remind{[START REWRITE Slightly]}
Concurrent data structures are the data sharing side of parallel programming.
Data structures give the means to the program to
store data but also provide operations to the program to access and manipulate these data. These operations are
implemented through algorithms that have to be efficient. In the sequential setting, data structures are crucially important for the performance of the respective computation. \index{parallel programming}
In the parallel programming setting, their importance becomes even more crucial because of the increased use of data and resource sharing for utilizing parallelism. In parallel programming, computations are split into subtasks in order to introduce parallelization at the control/computation level. To utilize this opportunity of concurrency, subtasks share data and
various resources (dictionaries, buffers, and so forth).
This makes it possible for logically independent programs to share various resources and data structures. A subtask that wants to update a data structure, say add an element into a dictionary, that operation may be logically independent of other subtasks that use  the same dictionary.

\index{mutual exclusion}
Concurrent data structure designers are striving to maintain consistency of data structures while keeping the use of mutual exclusion and expensive synchronization to a minimum, in order to prevent the data structure from becoming a sequential bottleneck. Maintaining consistency
in the presence of many simultaneous updates is a complex task.
Standard implementations of data structures are based on locks
in order to avoid inconsistency of the shared data due to concurrent modifications. In simple terms, a single lock around the whole data structure may create a bottleneck in the program where all of the tasks serialize, resulting in a loss of parallelism because too few data locations are concurrently in use.
\index{deadlock} \index{priority inversion} \index{convoying} \index{lock}
Deadlocks, priority inversion, and convoying are also side-effects of locking.
The risk for deadlocks
makes it hard to compose different blocking data structures since it is not always
possible to know how closed source libraries do their locking.
It is worth noting that in graphics processors (GPUs) locks are not recommended for designing concurrent data structures.
GPUs prior to the NVIDIA Fermi architecture do not have writable caches, so for
those GPUs, repeated checks to see if a lock is available or not require expensive repeated
accesses to the GPU's main memory. While Fermi GPUs do support writable caches, there is no guarantee that the thread scheduler
will be fair, which can make it difficult to write deadlock-free locking code. OpenCL \index{OpenCL}
explicitly disallows locks for these and other reasons.

\index{lock-free}
Lock-free implementations of data structures support concurrent access.
They do not involve mutual exclusion and make sure that all steps
of the supported operations can be executed concurrently.
\index{optimistic synchronization}
Lock-free implementations employ an optimistic conflict control
approach, allowing several processes to access the shared data object at the
same time. They suffer delays only when
there is an actual
conflict between operations that causes
some operations to retry.
This feature allows lock-free algorithms to scale much better when
the number of processes increases.

An implementation of a data structure is called {\em lock-free} if it allows multiple processes/threads to access the data structure concurrently and also guarantees that
at least one operation among those
finishes in a finite number of its own steps regardless of the state of the other operations.
\index{linearizability}
A consistency (safety) requirement for
lock-free data structures is {\em linearizability} \cite{HerW90},
which ensures that each operation
on the data appears to take effect instantaneously during its actual duration and
the effect of all operations are consistent with the
object's
 sequential specification.
Lock-free data structures offer several advantages over their
blocking counterparts, such as being immune to deadlocks,
priority inversion, and convoying, and have been shown to
work well in practice in many different settings \cite{Tsigas02,Sundell02}. They have been included \index{Threading Building Blocks}
in Intel's Threading Building Blocks Framework \cite{TBB09}, \index{NOBLE}
the NOBLE library \cite{Sundell02,Sundell08}, the Java concurrency package  \cite{JavaConc09} and the Microsoft .NET Framework \cite{msconc}. \index{.NET} \index{C++} \index{Java}
They have also been of interest to designers of languages such as C++ \cite{DecPS06}.
\remind{[END REWRITE]}

\remind{Merge with the text above.}

\leaveout{ 
While a parallel computation would ideally consist only of independent
parts, communication of some form between concurrent subcomputations is
often needed. Concurrent shared data structures offer ways to allow
such communication.
} 

The focus here is on concurrent implementations of common abstract
data types, such as queues, stacks, other producer-consumer
collections, dictionaries and priority queues, which can act as a
communication ``glue'' in parallel applications. Moreover, practical
lock-free protocols/algorithms are preferred due to their desirable
qualities in terms of performance and
fault-tolerance~\cite{Tsigas02,Sundell02}. For each of these abstract data
types there exist a considerable number of proposed
protocols/algorithms in the literature, see, e.g., the surveys
in~\cite{HerlihyShavit08,Cederman11-data-structures}.  As each
implementation of an abstract data type has different qualities
depending on how it is actually used, e.g. its contents; the level of
concurrent accesses to it; the mix of read or update operations issued
on it etc., it makes good sense to view them as multi-variant ``components''.
However, a concurrent shared data structure does not match the notion
of a component in the EXCESS component model defined in
EXCESS D1.2~\cite{EXCESS:D1.2} since an EXCESS component is a
computational entity while a concurrent shared data structure is a
data storage and communication entity. Hence, multi-variant concurrent
shared data structures need to enter the framework as something
different from components.
There are a number of ways concurrent shared data structures can be used in the
EXCESS programming model and runtime system:
\begin{itemize}
\item As internal communication medium, ``glue'', inside component
      implementations. The variant selection for the concurrent shared
      data structure can then either be performed inside the component
      implementation or the component implementation itself treated as
      a template and expanded into one actual component implementation
      for each variant of the data structure available. In the latter
      case variant selection for the data structure would reduce to
      component variant selection.
\item As parameters to components. In the EXCESS component
      model parameters to components are used to pass data in and out
      of components. Concurrent shared data structures can be used in
      this role and would support concurrent updates of the data structure
      from the inside and/or outside of the component during its execution.
      This use could be integrated in a similar way to the smart containers
      discussed in EXCESS D1.2~\cite{EXCESS:D1.2}.
\item As part of the implementation of the runtime system itself.
\end{itemize}

To select the most suitable data structure implementation for a given
situation is not an easy problem as many aspects of the subsequent use
of it impacts the time and energy costs for operations on the data
structure.
For example, superior operation throughput at high contention (a
common selling point of new algorithms for concurrent shared data
structures) from which usually (due to system static power) also
follows superior energy efficiency in that state does not necessarily
translate to superior energy efficiency at lower levels of contention
as the empirical case study in Section~\ref{sec:mandelbrot} below
indicates.

In EXCESS D1.1~\cite{D1.1} we determined a number of
energy-affecting factors for System~A, an Intel server system.
Below follows a discussion of
each of these factors in the context of a concurrent shared data structure
used as communication ``glue'' in a parallel computation.

\remind{(All) Filter the list below for sanity}

\vspace*{3ex}
\textbf{General}
\begin{itemize}
\item \emph{Execution time}.
      The time spent executing operations on the data structure
      depends on the number of calls to execute and the duration of
      each operation/call.  The latter is often difficult to predict
      as it may depend on many aspects of the state of the system and
      the data structure, such as the algorithm for the operation
      combined with the current state(/contents) of the data structure
      (cf. time complexity of sequential data structure operations)
      and the interference from other concurrent operations on the
      data structure.

\item \emph{Number of sockets used}.
      Determined by the scheduling of the tasks using the data structure.

\item \emph{Number of active cores}.
      Determined by the scheduling of the tasks using the data structure.

\end{itemize}

\textbf{Functional units}
\begin{itemize}
\item \emph{Instruction types}.
      Depends among other things on the algorithm for the operation
      combined with the current state(/contents) of the data structure
      and the interference from other concurrent operations on the
      data structure as these may activate different code-paths.

\item \emph{Dependency between operations}.
      Depends among other things on the algorithm for the operation
      combined with the current state(/contents) of the data structure
      and the interference from other concurrent operations on the
      data structure as these may activate different code-paths.

\item \emph{Branch prediction}.
      Failed predictions depend among other things on the algorithm
      for the operation combined with the current state(/contents) of
      the data structure and the interference from other concurrent
      operations on the data structure. The last may be particularily
      difficult since values are changed outside the current
      instruction sequence.

\item \emph{Clock frequency}.
      Determined by the system configuration.
\end{itemize}

\textbf{Memory}
\begin{itemize}
\item \emph{Resource contention}.
      Resource contention in the memory hierarchy can be either
      accidental, such as cache line eviction due to the limited size
      of the cache or due to placement policy restrictions forcing
      otherwise independent cache lines to content for a particular
      slot, or deliberate as is often the case in shared data structure
      implementations where code running on different cores tries to
      touch the same cache lines at the same time.

\item \emph{Number of memory requests}.
      Depends among other things on the algorithm for the operation
      combined with the current state(/contents) of the data structure
      and the interference from other concurrent operations on the
      data structure as these may activate different code-paths and force
      retries.

\item \emph{Level of memory request completion}.
      Interference from concurrent operations introduce additional cache line
      invalidations and, hence, cache misses.

\item \emph{Locality of memory references}.
      Interference from concurrent operations introduce additional cache line
      invalidations and coherence traffic.
\end{itemize}

As can be seen above, for concurrent shared data structures several of
these factors are affected by the concurrent operations on/dynamic
state of the data structure. This means that to estimate the time and
energy cost for one operation this state must be known or estimated.
As a first approach we will consider concurrent shared data structures
in a ``steady state'', that is, exposed to an unchanging mix of
operations issued by an unchanging set of tasks at an unchanging
rate. In this case we can then assign average time and energy costs to
operations based on the total throughput and power use. Empirical data
for a selection of such ``steady states'' on a particular system can be
collected with micro-benchmarks.

For the initial work on energy efficiency prediction for concurrent
shared data structures we have picked some commonly used collection
data types as case studies. These are concurrent producer/consumer
collections, such as queues, and concurrent search trees. Most of the
data structure implementations that we use are part of the NOBLE
library~\cite{Sundell02} which is described in Section~\ref{sec:NOBLE}
below.

\comment{AG}{
Energy efficiency for tasks/components v.s. for shared data structures.\\
\\
Energy and performance dependencies for:\\
Call to task/component(=function when viewed from the outside)\\
- Implementation\\
- Parameters in call\\
- Machine state due to other tasks/components\\
\\
Call to operation of concurrent shared data structure\\
- Implementation\\
- Which operation\\
- Parameters in call to operation\\
- State of the data structure\\
--- Internal (in-memory) state (size, contents, etc.)\\
--- Dynamic state (concurrent operations)\\
- Machine state due to other tasks/components\\
\\
Can a call to an operation on a data structure be viewed as a call to a
component?\\
My answer would be no - the information about the state of the data structure
needed is very likely to be more specific than the general machine state.\\
\\
Consequently prediction of energy efficiency becomes harder as the whole state
of the data structure during the duration of the operation needs to be
considered.\\
\\
More observations:\\
- Once an instance of a data structure is in use the variant choice cannot easily be undone (short of stopping all users and moving the contents to the new instance).
}

\leaveout{ 
\subsubsection{Concurrent Queues}
\remind{Brief intro here. See also shm-queues-description.tex and above.}
} 
\leaveout{
\subsubsection{Concurrent Search Trees}
Most of the existing highly-concurrent search trees are not considering the fine-grained
data locality. This includes some of the state-of-the-art non-blocking concurrent search trees 
\cite{Brown:2011:NKS:2183536.2183551, EllenFRB10} and Software Transactional
Memory (STM) search trees
\cite{Afek:2012:CPC:2427873.2427875, BronsonCCO10,
Crain:2012:SBS:2145816.2145837, DiceSS2006}. 
Prominent concurrent search trees which are often included in several benchmark 
distributions such as the concurrent red-black
tree \cite{DiceSS2006} developed by Oracle Labs and the concurrent AVL tree
developed by Stanford \cite{BronsonCCO10} are not optimising data locality as well. 

Concurrent B-trees \cite{BraginskyP12, Comer79, Graefe:2010:SBL:1806907.1806908,
Graefe:2011:MBT:2185841.2185842}
are optimised for a known memory block size $B$ (e.g. page size) to minimise the
number of memory blocks accessed during a search, thereby improving data
locality. 
In reality there are different block sizes at different levels of the memory hierarchy
that can be used in the design of locality-aware data layout for search trees. 
In \cite{KimCSSNKLBD10}, Intel shows that a very fast search trees are possible after optimising  
a data layout based on the register size, SIMD width, cache line size, and page size. 
Existing concurrent B-trees limit its spatial locality optimisation to the memory level
within block size $B$, leaving access to other memory levels unoptimised.

For a concurrent B-tree that is optimised for accessing disks (i.e.
using page size $B$), the cost of searching a key in a block of size $B$ in
memory is $\Theta (\log (B/L))$ cache line transfers, where $L$ is the cache line
size \cite{BrodalFJ02}. This is because each cache line transfers of size $L$ contains only one node
that is usable in a top down traversal path in the search tree sized $B$, 
except for the topmost $\lfloor \log(L+1) \rfloor$ levels. There
is an optimal number cache line transfers of $O(\log_L B)$ for the same scenario that is achievable
by using the van Emde Boas layout.

A van Emde Boas (vEB) tree is an ordered dictionary data type which implements
the idea of a recursive structure of priority
queues \cite{vanEmdeBoas:1975:POF:1382429.1382477}.
The efficient structure of the vEB tree arranges data
in a recursive manner so that related values are placed in contiguous memory
locations (cf. Figure \ref{fig:vEB}). This work has inspired cache oblivious (CO) 
data structures \cite{Prokop99} 
such as CO B-trees \cite{BenderDF05, BenderFFFKN07, BenderFGK05} and CO
binary trees \cite{BrodalFJ02}. These previous research has demonstrated that 
vEB layout is suitable for cache oblivious algorithms as it lowers the upper bound on memory transfer complexity.

\begin{figure}[t]
\centering \scalebox{0.4}{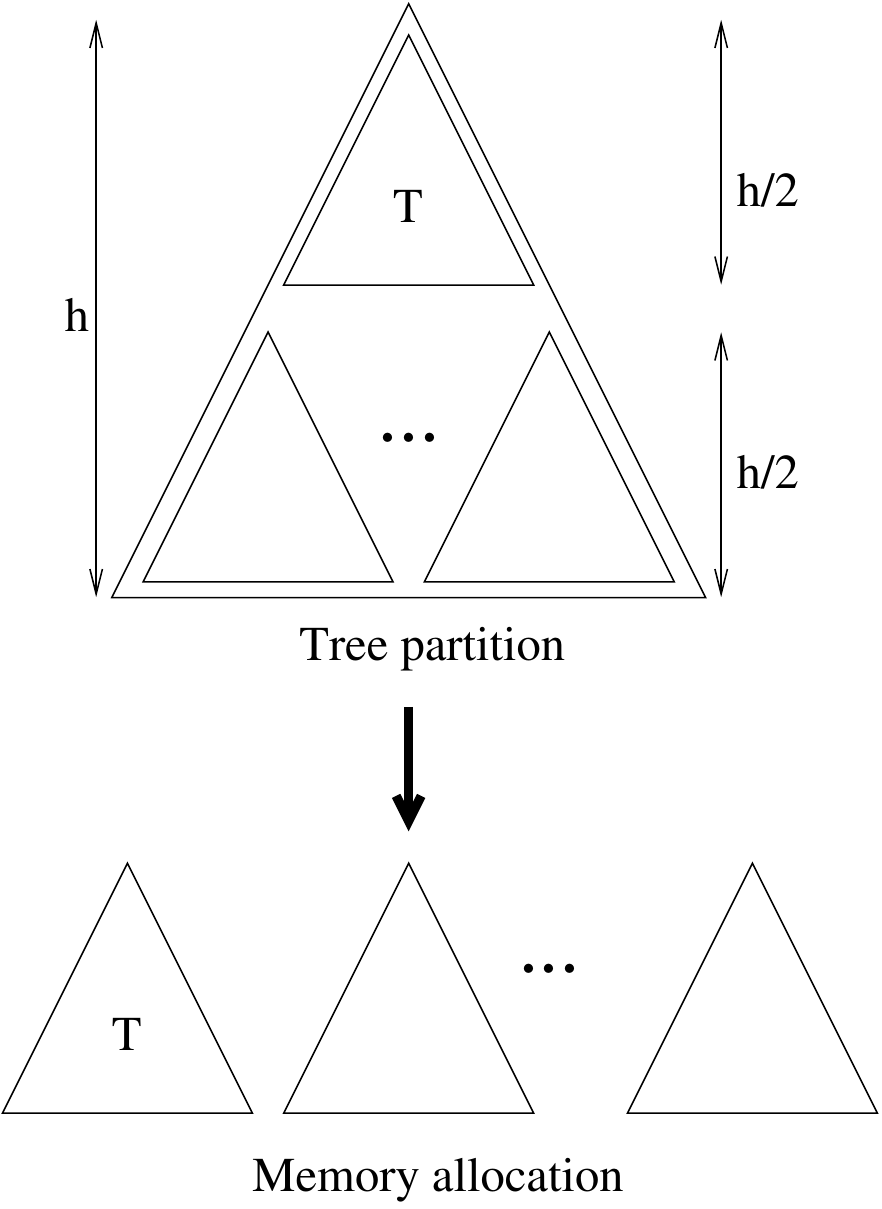}
\caption{Van Emde Boas layout: a tree of height $h$ is recursively split at height $h/2$. 
The top subtree $T$ of height $h/2$ and $m=2^{h/2}$ bottom subtrees $B_1;B_2; \ldots ;B_m$ 
of height $h/2$ are located in contiguous memory locations in the order of 
$T;B_1;B_2;\ldots;B_m$.}\label{fig:vEB}
\end{figure}

While proven to perform well on search, vEB-based trees poorly support concurrent update operations.
Inserting or deleting a node in a tree may result in relocating a large part of
the tree in order to maintain the vEB layout. The work in  \cite{BenderFGK05} has discussed this problem and
a feasible implementation \cite{BraginskyP12} has yet to appear.

We introduce DeltaTree ($\Delta$Tree), a novel locality-aware concurrent search
tree that combines both locality-optimisation techniques from vEB-based trees
and concurrency-optimisation techniques from non-blocking highly-concurrent
search trees.

}

\subsection{Micro-benchmarking}

In general, micro-benchmarking is conducted to discover  targeted properties of a specific system. A \emph{micro-benchmark} in this work context is a small piece of code designed to measure the performance or energy of basic operations of hardware or software systems.

The micro-benchmarks are developed by a loop over a low-level operation (e.g. xor, mul) with N iterations. This description of micro-benchmarks can be represented as EBNF-inspired formal notation (\textit{op}$)^N$ which means the body $op$ of a loop is executed with $N$ iterations. In order to analyze experimental results, micro-benchmarks must work with a fixed size data-set, perform a constant amount of work per iteration and run for a reasonable amount of time. 
In this work package, micro-benchmarks are used to find out the key features and properties of the components that affect energy efficiency of the systems. From the measurement results of micro-benchmarks, the parameters in a proposed analytic energy model are derived.  


Micro-benchmarking is a common method used for performance and energy modeling of a given computer system described at an abstract level.  It can work directly on the a specific component and does not require detailed simulation. Micro-benchmarking can be applied to different systems since its code is portable. By using the micro-benchmarking method, offline predictors of the energy model are built to predict energy and performance and to support energy optimization.

\leaveout{ 
\subsection{Models for Energy Consumption of Concurrent Data Structures and Algorithms}
\remind {Chalmers: please fill  this section}
} 

\subsection{Metrics}\label{sec:metrics}
\label{sec:int-met}
We rely on two original metrics:
\begin{itemize}
\item Throughput: it is a natural metric, used extensively in the
  performance analysis of data structures, and which measures the
  number of operations that has been done on the data structure
  per second.
\item Power.
\end{itemize}

The experiments consist of running a benchmark on the data structure
during a given time. Then we count the number of successful
operations, which gives the throughput. As we work under a constant
execution time, power and energy are equal within a multiplicative
factor that is the execution time.

However, another metric is studied in this deliverable, so that we are
able to evaluate the energy efficiency of different implementations of
the same data structure: the energy per operation. This metric can be
useful in the following case: we are given a workload that needs to be
executed on a given platform, and there is no time requirement. Then,
in terms of energy savings, the implementation that uses the minimum energy
per operation is the best implementation.

In our case, the energy per operation is obtained by simply dividing
the power by the throughput.


Finally, if we are interested in the bi-criteria problem that mixes energy
and performance (\ie where we aim at optimizing both energy and
performance), we can plot the energy per operation according to the
throughput. By doing this we can also trace the Pareto-optimal frontier for
this bi-criteria problem (we eliminate every point such that there exists
another point that is better both in terms of energy per operation and
throughput).

In this deliverable, we model the two original metrics, namely throughput
and power, and derive the other ones from those two.

\subsection{Overview and Contributions}
As a first step to investigate the energy and performance trade-offs in concurrent data
structures, we have developed a new power model for the Movidius Myriad platform that is
able to predict power consumption of our micro-benchmarks with $\pm 4\%$ margin of
measured power consumption on the real platform (cf. Sec. \ref{sec:movidius_energy}). The
new power model confirms the experimental power analysis of concurrent data structures
such as concurrent queues: the dynamic power consumption is proportional to the number of
SHAVE processors used. 

We have performed micro-benchmarks on CPU as well, where we have decomposed power (CPU
power, memory power and uncore power) into several parts, namely static, active and
dynamic part. In this decomposition, we are able to split the dependencies of power
dissipation according to the kind of operation, locality of operands, number of active
cores and sockets.

This micro-benchmark study is a preliminary step towards the modeling of performance and
power dissipation of several concurrent queue implementations. We define parameters that rule the
behavior of the queues, and show how to extrapolate both throughput and power values,
by relying on only a few measurements.

Moreover, we have analyzed the feasibility of porting concurrent data structures onto low
energy embedded platforms (cf. Sec. \ref{sec:queue_modeling}). We have selected a few
synchronization mechanisms from the HPC domain that could easily be replicated on a
Movidius Myriad MPSoC and analyzed the performance when used to implement concurrent FIFO
queues. This work has been continued to investigate the energy consumption of these
synchronization mechanisms and we have then proceeded to determine energy-performance
trade-offs.

\remind{Chalmers: please add the other key contributions here}

On Intel platform, we have investigated the energy consumption and performance of
concurrent data structures such as concurrent search trees
(cf. Sec. \ref{sec:DeltaTree_description}). Based on our investigation, we have developed
new concurrent search trees called $\Delta$Trees that are up to 140\% faster and 80\% more
energy efficient than the state-of-the-art concurrent search trees.

The remaining of the report is organized as follows. Section \ref{sec:descr} describes the
methodology to measure the consumed energy for two EXCESS platforms, namely CPU-based and
Movidius. Section \ref{sec:EnergyModel} presents energy models that can be used to
predict the power consumed on each platform. In section \ref{sec:queue_modeling}, the
performance and energy-efficiency of concurrent queue data structures are investigated on
CPU-based and Movidius platforms. The investigation results are supported by experimental
evaluations. Section \ref{sec:DeltaTree_description} analyzes the performance and
energy-efficiency of concurrent tree data structures and introduces a novel locality-aware
concurrent search tree. The conclusions and future work are provided in
section \ref{sec:Conclusion}.

\section{EXCESS Platforms and Energy Measurement Methodology}
\label{sec:descr}

\comment{AG}{Rehash of D1.1 content? Keep short?}

In this section we summarize the EXCESS systems and the energy
measurement methodology used for
the experiments in this report. A summary is given here as
these systems have previously been described in EXCESS D1.1~\cite{D1.1} and
D5.1~\cite{EXCESS:D5.1}.

\remind{(All) More systems? Have the EXCESS cluster at HLRS been used?}
\begin{itemize}
\item \textbf{System A}: An Intel multicore CPU server (located at Chalmers);
\item \textbf{System B}: Movidius Myriad1 MV153 development board and simulator
       (evaluated at Movidius and UiT).
\end{itemize}

\subsection {System~A: CPU-based platform}
\label{sec:chalmers-system}

\subsubsection{System description}
\label{sec:sysDscrA}

\begin{itemize}
  \item CPU: Intel(R) Xeon(R) CPU E5-2687W v2
    \begin{itemize}
      \item 2 sockets, 8 cores each
      \item Max frequency: 3.4GHz, Min frequency: 1.2GHz, {frequency speedstep by DVFS: 0.1-0.2GHz.} Turbo mode: 4.0GHz.
      \item Hyperthreading (disabled)
      \item L3 cache: 25M, internal write-back unified, L2 cache: 256K, internal write-back unified. L1 cache (data): 32K internal write-back
    \end{itemize}
  \item DRAM: 16GB in 4 4GB DDR3 REG ECC PC3-12800 modules run at
    1600MTransfers/s. Each socket has 4 DDR3 channels, each supporting 2
    modules. In this case 1 channel per socket is used.
  \item Motherboard: {Intel} Workstation W2600CR, BIOS version: 2.000.1201 08/22/2013
  \item Hard drive: Seagate ST10000DM003-9YN162 1TB SATA
\end{itemize}

\subsubsection{Measurement methodology for energy consumption}
\label{sec:SystemA:overview}
\begin{figure}
\begin{center}
\includegraphics[width=0.5\textwidth]{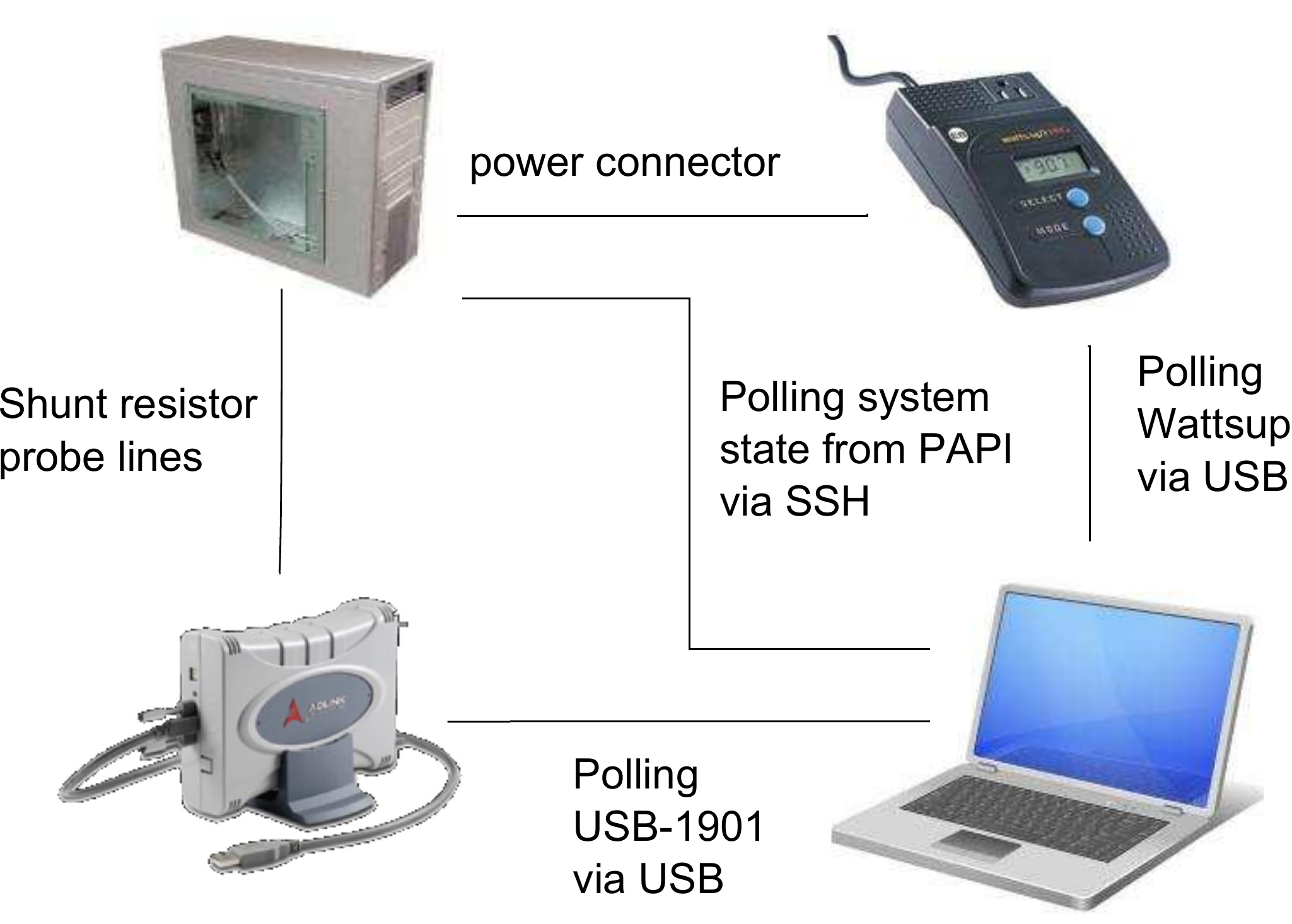}
\end{center}
\caption{Deployment of energy measurement devices for System~A.}
\label{fig:devicesSystemA}
\end{figure}

The energy measurement equipment for System~A at CTH, described in
Section~\ref{sec:sysDscrA}, is shown in Figure~\ref{fig:devicesSystemA}
and outlined below. It has previously been
described in detail in EXCESS D1.1~\cite{D1.1} and
D5.1~\cite{EXCESS:D5.1}.

The system is equipped with external hardware sensors for two levels
of energy monitoring as well as built in energy sensors:
\begin{itemize}
  \item At the system level using an external Watts Up
    .Net~\cite{EED:2003:WattsUp} power meter, which is connected
    between the wall socket and the system.
  \item At the component level using shunt resistors inserted between the
    power supply unit and the various components, such as CPU, DRAM and
    motherboard. The signals from the shunt resistors are captured with an
    Adlink USB-1901~\cite{Adlink:2011:USB1900} data acquisition unit (DAQ)
    using a custom utility.
  \item Intel's RAPL energy counters are also available for the CPU
    and DRAM components. A custom utility based on the PAPI
    library~\cite{BrDoGaHoMu:2000:PAPI,Weaver:2012:MEP:2410139.2410475}
    is used to record these counters and other system state parameters
    of interest.
\end{itemize}

For the work presented in this report the component level hardware
sensors and the RAPL energy counters have mainly been used.

\subsection{System~B: Movidius Embedded Platform (Myriad1)}
\label{sec:myriad1-system}
\subsubsection{Myriad1 Platform Description} 

The Myriad1 platform developed by Movidius contains a total of 8 separate SHAVE (Streaming Hybrid Architecture Vector Engine) processors (see Figure ~\ref{fig:ShaveInstructionUnits}), each existing on solitary power islands. 

The SHAVE processor contains a mix of RISC, DSP, VLIW and GPU features 
and supports the following data types: 
(float) f16/32, (unsigned) u8/16/32, and (int) i8/16/32. The SHAVE architecture uses Very Long Instruction Words (VLIWs) as input.  
The processor is designed to provide a platform that excels in multimedia and video processing.
Each SHAVE has its own Texture Management Unit (TMU).

 
SHAVE also contains wide and deep register files coupled with a 
Variable-Length Long Instruction-Word (VLLIW) for code-size efficiency.  
As shown in Figure~\ref{fig:ShaveInstructionUnits} VLLIW packets control multiple functional units which have SIMD capability for high parallelism and throughput at a functional unit and processor level. 

\textbf{Functional Units of SHAVE}

\begin{itemize}
\item Integer Arithmetic Unit (IAU) – Performs all arithmetic instructions that operate on integer numbers, accesses the IRF.
\item Integer Arithmetic Unit (IAU) – Performs all arithmetic instructions that operate on integer numbers, accesses the IRF.
\item Scalar Arithmetic Unit (SRF) – Performs all Scalar integer/floating point arithmetic and interacts with the SRF or IRF depending on what values are used.
\item Vector Arithmetic Unit (VAU) – Performs all Vector integer/floating point arithmetic and interacts with the VRF.
\item Load Store Unit (LSU) – There are 2 of these (LSU0 \& LSU1) and they perform any memory IO instructions. This means that it interacts with the 128kB CMX memory tile located in the SHAVE.
\item Control Move Unit (CMU) – This unit interacts with all register files, and allows for comparing and moving between the register files.
\item Predicated Execution Unit (PEU) – Performs operations based on condition code registers.
\item Branch Repeat Unit (BRU) – Manages instructions involving any loops, as well as branches.
\item Instruction Decoding Unit (IDC) – This unit takes a SHAVE variable-length instruction as input and decodes it to determine which functional units are being utilised by the inputted instruction.
\item Debug Control Unit (DCU) – Used for monitoring the execution of the program, takes note of interrupts and exceptions.

\end{itemize}

\textbf{Register Files}

\begin{itemize}
\item Integer Register File (IRF) – Register file for storing integers from either the IAU or the SAU. Can hold up to 32 words which are each 32-bits wide.
\item Scalar Register File (SRF) – Register file for storing integers from either the SAU. Can hold up to 32 words which are each 32-bits wide.
\item Vector Register File (VRF) – Register file for storing integers from either the VAU. Can hold up to 32 words which are each 128-bits wide. 
\end{itemize} 


The additional blocks in the diagram are the Instruction DeCode (IDC) and Debug Control Unit (DCU).  
An instruction fetch width of 128-bits and 5-entry instruction pre-fetch buffer  guarantee that at least one instruction is ready taking account of branches.  
Data and instructions reside in a shared Connection MatriX (CMX) memory block which can be configured in 8kB increments to accommodate different  instruction/data mixes depending on the workload.  The CMX also includes address-translation logic to allow VLLIW code to be easily relocated to any core in the system.

In the 65nm System-on-Chip (SoC), eight SHAVE processors are combined  with a software-controlled memory subsystem and caches which can be  configured to allow a large range of workloads to be handled,  providing exceptionally high sustainable on-chip bandwidth to support  data and instruction supply to the 8 processors. 
Data is moved between peripherals, processors and memory via a bank of  software-controlled DMA engines.    
The device supports 8, 16, 32 and some 64-bit integer operations 
as well as fp16 and fp32 arithmetic and is capable of aggregate 1 TOPS/W  maximum 8-bit equivalent operations in a low-cost plastic BGA package with integrated 128Mbit or 512Mbit Mobile DDR2 SDRAM.

As power efficiency is paramount in mobile applications, 
in addition to extensive clock and functional unit gating and support for dynamic clock and voltage scaling for dynamic power reduction, the device contains a total of 11 power-islands: one for each SHAVE, one for the CMX RAM, one for the RISC and peripherals and one always-on domain. 
This allows very fine-grained power control in software with minimal latency to return to normal operating mode, including maintenance of SRAM-state eliminating the need to reboot from external storage.  

\begin{figure}
\begin{center}
\includegraphics[width=0.9\textwidth]{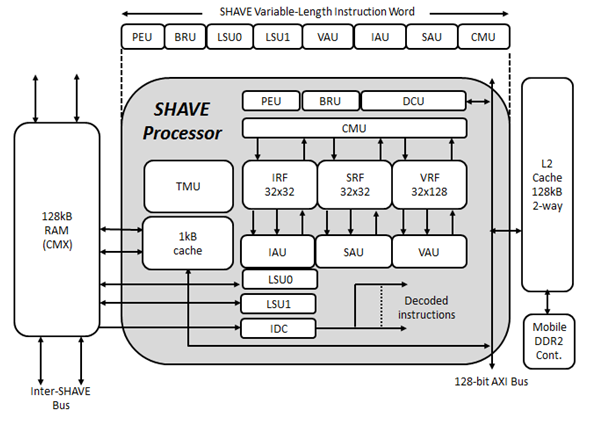}
\end{center}
\caption{SHAVE Instruction Units}
\label{fig:ShaveInstructionUnits}
\end{figure}

\subsubsection{Measurement methodology}

An extensive effort has been made to measure Myriad1 performance in a semi-automated way in order to produce better power estimates.  

The Power Measurements tests were introduced in order to have an insight into the power consumed by Myriad1 in several basic cases in order to be able both to decompose the Myriad1 power consumption into power components and characterize power consumption in such basic operations of Myriad1. The tests set were devised to characterize the power consumed when running SHAVE code without DDR data accesses.

\begin{figure}
\center
\includegraphics[width=0.9\textwidth]{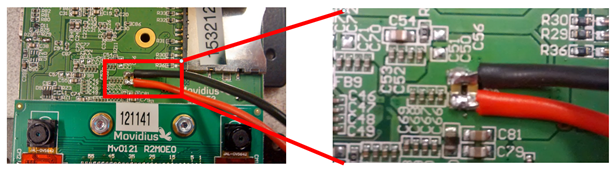}
\caption{Power Supply Modification}
\label{fig:PowerSupplyModification}
\end{figure}

\begin{figure}
\center
\includegraphics[width=0.9\textwidth]{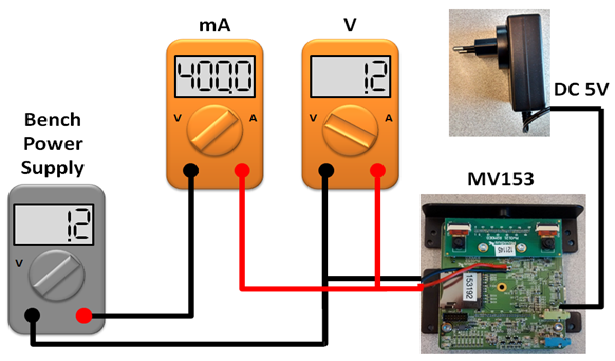}
\caption{Bench setup for MV153 Power Measurement Schematic}
\label{fig:MeasurementSchematic}
\end{figure}

The modifications were made to the MV153 to bypass the on-board voltage regulator which down-regulates the 5V wall PSU to the 1.2V core voltage required by Myriad1 allowing an external bench power-supply to be used in its place as shown in Figure \ref{fig:PowerSupplyModification}.

The schematic for the connection of the Power Supply Unit (PSU), multimeters and MV153 for power measurement are shown in Figure \ref{fig:MeasurementSchematic}. Note the standard DC wall supply is required in addition to the bench PSU in order to supply the other elements of the system.

The bench setup consists of a modified MV153 board, a DC step down converter down-regulating the 5V wall PSU to the 1.2V core voltage  and one HAMEG multimeter measuring all the voltage, current and consumed power values as shown in Figure \ref{fig:MeasurementSetup}.  

\begin{figure}
\begin{center}
\includegraphics[width=0.9\textwidth]{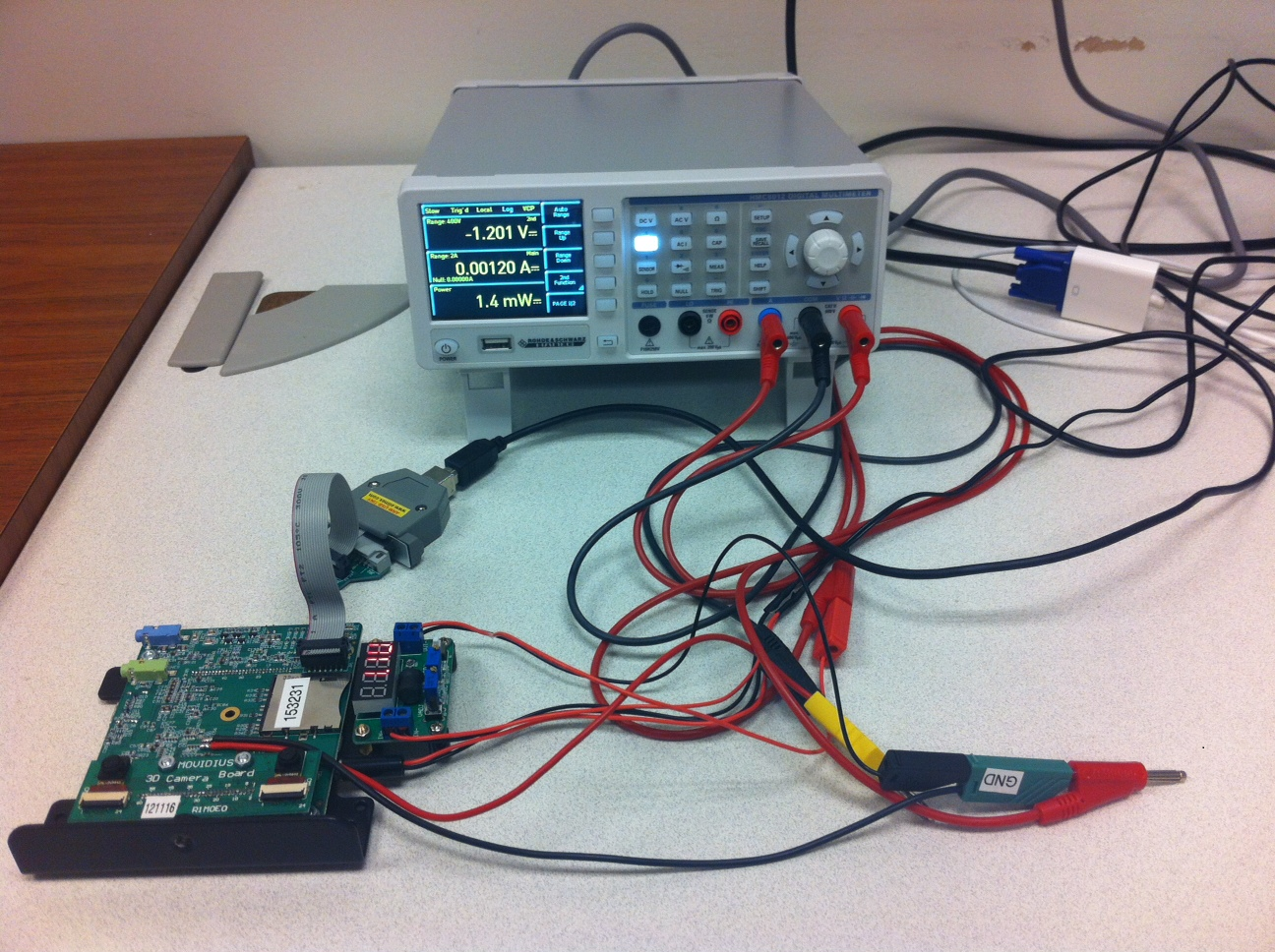}
\end{center}
\caption{Bench setup for MV153 Power Measurement}
\label{fig:MeasurementSetup}
\end{figure}

\subsection {Significant Synchronization Hardware Differences Between the Two Systems}
\label{sec:myriad1-sync}

To synchronize processes efficiently, multi-/many-core systems usually support certain synchronization primitives.
This section discusses the fundamental synchronization primitives, which typically read the value of a {\em single} memory word, modify the value and write the new value back to the word {\em atomically}. Different architectures support different synchronization primitives in hardware.

\subsubsection{Fundamental synchronization primitives}

The definitions of the primitives are described in Figure~\ref{fig:intro_primitives}, where
$x$ is a memory word,
$v, old, new$ are values and $op$ can be operators $add$, $sub$, $or$, $and$ and $xor$. Operations between angle brackets $\langle\rangle$ are executed atomically.
\begin{figure}[tbh]
\hrule
\centering
\small
\begin{tabular}{p{2in}p{1in}}
\begin{tabbing}
{\bf TAS}\=$(x)$ {\em /* test-and-set, init: $x \leftarrow 0$ */}\\
\> $\langle oldx \leftarrow x$; $x \leftarrow 1$; {\bf return} $oldx; \rangle$ \\

\\

{\bf FAO}\=$(x,v)$ {\em /* fetch-and-op */} \\
\> $\langle oldx \leftarrow x$; $x \leftarrow op(x,v)$; {\bf return} $oldx; \rangle$\\

\\

{\bf CAS}\=$( x, old, new)$ {\em /* compare-and-swap */} \\
\> $\langle$ \= {\bf if}$(x=old)$ $\{ x \leftarrow new$; {\bf return}$(true);\}$\\
\>\>{\bf else return}$(false)$; $\rangle$

\end{tabbing}
&
\begin{tabbing}
{\bf LL}\=$(x)$ {\em /* load-linked */}\\
\> $\langle return~the~value~of~x~so~that$\\
\> $it~may~be~subsequently~used$\\
\> $with$ \bf{SC} $\rangle$\\

{\bf SC}$(x,v)$ {\em /* store-conditional */}\\
\> $\langle$ \= {\bf if} $(no~process~has~written~to~x$ \\
\> $since~the~last$ {\bf LL}$(x))$ $\{ x \leftarrow v$; \\
\>\> {\bf return}$(true) \}$;\\
\>\> {\bf else return}$(false)$; $\rangle$
\end{tabbing}
\end{tabular}
\hrule
\caption{Synchronization primitives}\label{fig:intro_primitives}
\end{figure}

\paragraph{Synchronization power}
The primitives are classified according to their synchronization power or \index{consensus number} {\em consensus number}~\cite{Her91}, which is, roughly speaking, the maximum number of processes for which the primitives can be used to solve a {\em consensus problem} in a fault tolerant manner. In the consensus problem, a set of $n$ asynchronous processes, each with a given input, communicate to achieve an agreement on one of the inputs. A primitive with a consensus number $n$ can achieve consensus among $n$ processes even if up to $n-1$ processes stop~\cite{Turek92}.

\index{synchronization primitives!compare-and-swap (CAS)}
 \index{synchronization primitives!load-linked/store-conditional (LL/SC)}
 \index{synchronization primitives!fetch-and-op (FAO)}
  \index{synchronization primitives!test-and-set (TAS)}
According to the consensus classification, read/write registers have consensus number 1, i.e. they cannot tolerate any faulty processes in the consensus setting. \index{compare-and-swap (CAS)} \index{load-linked/store-conditional (LL/SC)} \index{fetch-and-op (FAO)} \index{test-and-set (TAS)}
There are some primitives with consensus number 2 (e.g. {\em test-and-set (TAS)} and {\em fetch-and-op (FAO)}) and some with infinite consensus number (e.g. {\em compare-and-swap (CAS)} and {\em load-linked/store-conditional (LL/SC)}).
It has been proven that a primitive with consensus number $n$ cannot implement a primitive with a higher consensus number in a system of more than $n$ processes \cite{Her91}. For example, the {\em test-and-set} primitive, whose consensus number is two, cannot implement the {\em compare-and-swap} primitive, whose consensus number is unbounded, in a system of more than two processes. Most modern general purpose multiprocessor architectures support  {\em compare-and-swap (CAS)} in hardware. {\em compare-and-swap (CAS)} is also the most popular synchronization primitive for implementing both lock-based and nonblocking concurrent data structures. For many non-blocking data structures a primitive with a consensus number $n$ is needed.

The Myriad platform, as many other embedded platforms, avails {\em test-and-set (TAS)} registers which have consensus number 2 and not {\em compare-and-swap (CAS)}. These  Test-and-Set registers
can be used to create spin locks, which are commonly referred as "mutexes". Spin-locks are used to create busy-waiting synchronization techniques: a thread spins to acquire the lock so as to have access to a shared resource.

The Myriad platform also avails a set of registers that can be used for fast SHAVE arbitration. Each SHAVE has its own copy of these registers and its size is 4x64 bit words. An important characteristic is that they are accessed in a FIFO pattern, so each one of them is called a ``SHAVE's FIFO''. Each SHAVE can push data to the FIFO of any other SHAVE, but can read data only from its own FIFO. A SHAVE writes to the tail of another FIFO and the owner of the FIFO can read from any entry. If a SHAVE attempts to write to a full FIFO, it stalls. Finally, the LEON processor cannot access the FIFOs. SHAVE FIFOs can be utilized to achieve efficient synchronization between the SHAVEs. Also, they provide an easy and fast way for exchanging data directly between the SHAVEs (up to 64 bits per message), without the need to use shared memory buffers.

Analysis of experiments on the Myriad platform shows that the mutex implementation is a fair lock with round-robin arbitration. But most scalable designs for concurrent data structures require \index{compare-and-swap (CAS)}, a hardware primitive that has unbounded consensus number. 
Because of the lack of support of strong synchronization primitives, from the embedded hardware side,
we had to come with new algorithmic designs for the data structures under consideration fitting the capabilities of the embedded systems area.

\section{Energy Models for EXCESS Platforms} \label{sec:EnergyModel}

\subsection{Energy Models for CPU-based Platforms}\label{sec:cpu-micro}
\newcommand{\pcaA}{\ema{A}}       
\newcommand{\pcbA}{\ema{B}}       
\newcommand{\pceA}{\ema{\alpha}}  
\newcommand{\pceoA}{\ema{\alpha_0}}  

\newcommand{\freqA}{\ema{freq}}
\newcommand{\freqB}{\ema{\mathit{freq}}}

\renewcommand{\pin}{\ema{\mathit{pin}}}
\newcommand{\sock}{\ema{\mathit{soc}}}
\newcommand{\numth}{\ema{\mathit{thr}}}

\newcommand{\loc}{\ema{\mathit{loc}}}
\newcommand{\ope}{\ema{\mathit{op}}}
\newcommand{\opreg}{\ema{\mathit{op_{reg}}}}


\subsubsection{General Power Model}
\label{sec.gen-mod-cpu}

The power model that is presented in EXCESS D1.1~\cite{D1.1} decomposes the total power into
static, socket activation and dynamic power, as recalled in
Equation~\ref{eq.power-model-micro}. In this equation, \freq is the clock frequency, \sock
the number of activated sockets on the chip, \ope is the considered operation and \numth is
the number of active cores; the active power is proportional to the number of active
sockets, while the dynamic power is proportional to the number of active cores.

For modeling power consumption of data structures, we need to estimate the dynamic
component which depends on the frequency, number of active cores, locality and amount of
memory requests together with the instruction type as is also mentioned in D1.1~\cite{D1.1}.
Therefore, we improve the model with an additional parameter \loc to represent the locality
of operands for instructions that can transfer data between memory and registers, 
such a move from L1, L2, last level cache, main memory or remote memory. This parameter was
not included in D1.1 because we consider only the total power in which \loc parameter plays
a negligible role. For fine-grained analysis, we include this parameter in our model.

\begin{equation}\left\{\begin{array}{l}
\pow{}(\freq,\ope,\sock,\loc,\numth) = \pstat{}  + \pact{} (\freq,\sock)  + \pdyn{} (\freq,\ope,\loc,\numth)\\
\pact{}(\freq,\sock) = \sock \times \pact{}(\freq)\\
\end{array}\right.
\label{eq.power-model-micro}
\end{equation}

\begin{table} [h!]
\begin{center}
\begin{tabular}{|c|c|c|c|}
\cline{2-4}
\multicolumn{1}{c|}{}&        Static & Active & Dynamic\\\hline
CPU    & \pstat{C} & \pact{C} & \pdyn{C} \\\hline
Memory & \pstat{M} & \pact{M} & \pdyn{M} \\\hline
Uncore & \pstat{U} & \pact{U} & \pdyn{U} \\\hline
\end{tabular}
\end{center}
\caption{Power views\label{fig.pow-dec-micro}}
\end{table}

As another improvement to the previous power model, we
decompose the power into two orthogonal bases, each base having three
dimensions. On the one hand, we define the model basis by separating the power into
static, active and dynamic power, such that the total power is computed by:
\[ \pow{} = \pstat{} + \pact{} + \pdyn{}. \]
On the other hand, the measurement basis corresponds to the components
that actually dissipates the power, \ie CPU, memory and uncore. 
The power dissipation measurement is done through Intel's RAPL energy counters
read via the PAPI
library~\cite{BrDoGaHoMu:2000:PAPI,Weaver:2012:MEP:2410139.2410475}. These counters reflect this discrimination
by outputting the power
consumption along three dimensions:
\begin{itemize}
\item power consumed by CPU, which includes the consumption of the computational cores, and
  the consumption of the first two level of caches;
\item power consumed by the main memory;
\item remaining power, called ``uncore'', which includes the ring interconnect, shared cache, 
integrated memory controller, home agent, power control unit, integrated I/O module, config agent, caching agent
and Intel QPI link interface.
\end{itemize}
Also, total power is obtained by the sum:
\[ \pow{} = \pow{C} + \pow{M} + \pow{U}. \]
This latter additional orthogonal dimension will provide a better
perspective for modeling power consumption of data structures, especially for the dynamic
component. Table~\ref{fig.pow-dec-micro} sums up both dimensions.

In this section, we study each dimension, in each base, so that we are able to express the
power dissipation from any perspective:
\[ \pow{} = \sum_{X \in \{C,M,U\}} \left( \pstat{X} + \pact{X} + \pdyn{X}  \right). \]




\subsubsection{Power Components Derivation}

By definition, only the dynamic component of power is dependent on the type of instruction
or more generally the executing program. In order to obtain dynamic component \pdyn{}, we
first have to determine static \pstat{} and socket activation \pact{} costs.  This was done
in D1.1 but the derivation process was not described in detail. Therefore, we will explain
it briefly in this subsection.

In D1.1, a large variety of instructions were examined with respect to their power and
energy consumption. We have observed a linear relation between the number of threads and
power for instructions that do not lead to data transfer between the memory hierarchy
and registers. For instance, addition operates on two registers and do not lead to data
transfer. Data transfer is done via move instructions before or after addition instruction
if required. So, the locality parameter \loc is only valid for instructions that is dependent
on the locality of data, like variants of the move instruction. These operations are also prone
to variability due to cache and memory states which can also change with the interaction
between threads. Briefly, \pdyn{M} and \pdyn{U} is significant only for the instructions
that lead to data transfer in the memory hierarchy. Also the \loc parameter is only
meaningful for \pdyn{M} and \pdyn{U}. For derivation of \pstat{} and \pact{}, we just use the
instructions that operate on the registers because the \pact{M} and \pact{U} parts can be neglected
for these instructions. We refer to these instructions as \opreg
and utilize them to obtain static and socket activation
costs for each component (CPU, memory, uncore) of the orthogonal decomposition.
A bunch of instructions belonging to \opreg is executed repeatedly for some time
interval with varying number of threads for each frequency.
We formulate the derivation process as, for all $X\in \{C,M,U\}$:

\begin{align*}
\pdyn{}(\freq,\ope,\loc,\numth) =& \,\pdyn{M,U}(\freq,\ope,\loc,\numth) + \pdyn{C}(\freq,\ope,\numth)\\
\pdyn{M,U}(\freq,\opreg,\loc,\numth) =&\, 0\\
\pdyn{C}(\freq,\opreg,\numth) =&\, \numth \times \pdyn{C}(\freq,\opreg)\\
\pdyn{X}(\freq,\opreg) =&\, 
\frac{1}{2} \left( \pow{X}(\freq,\opreg,\sock,\loc,16) - \pow{X}(\freq,\opreg,\sock,\loc,14) \right)\\
\pact{X}(\freq) =&\, \pow{X}(\freq,\opreg,2,\loc,10) -
\pow{X}(\freq,\opreg,1,\loc,8) - \pdyn{X}(\freq,\opreg) \times 2\\
\pstat{X}() =&\, \pow{X}(\freq,\opreg,\sock,\loc,\numth) - 
\sock \times \pact{X}(\freq) - \numth \times \pdyn{X}(\freq,\opreg)
\end{align*}

Using above equations, we verified that $\pstat{X\in \{C,M,U\}}$ is approximately constant according to
instruction type, pinning, number of threads and frequency thus we take the mean of the values of \pstat{} over
the whole space to find \pstat{}. We apply the same approach to find $\pact{X\in \{C,M,U\}}$ which
only depends on frequency, and not on the operation.
Having obtained \pstat{} and \pact{X}, we extract \pdyn{X\in \{C,M,U\}} for ``all''
types of instructions, thread, pinning and frequency setting, by removing the static and
active part from the total power.

\subsubsection{Dynamic CPU Power}
\begin{figure} [h!]
\includegraphics[width=\textwidth]{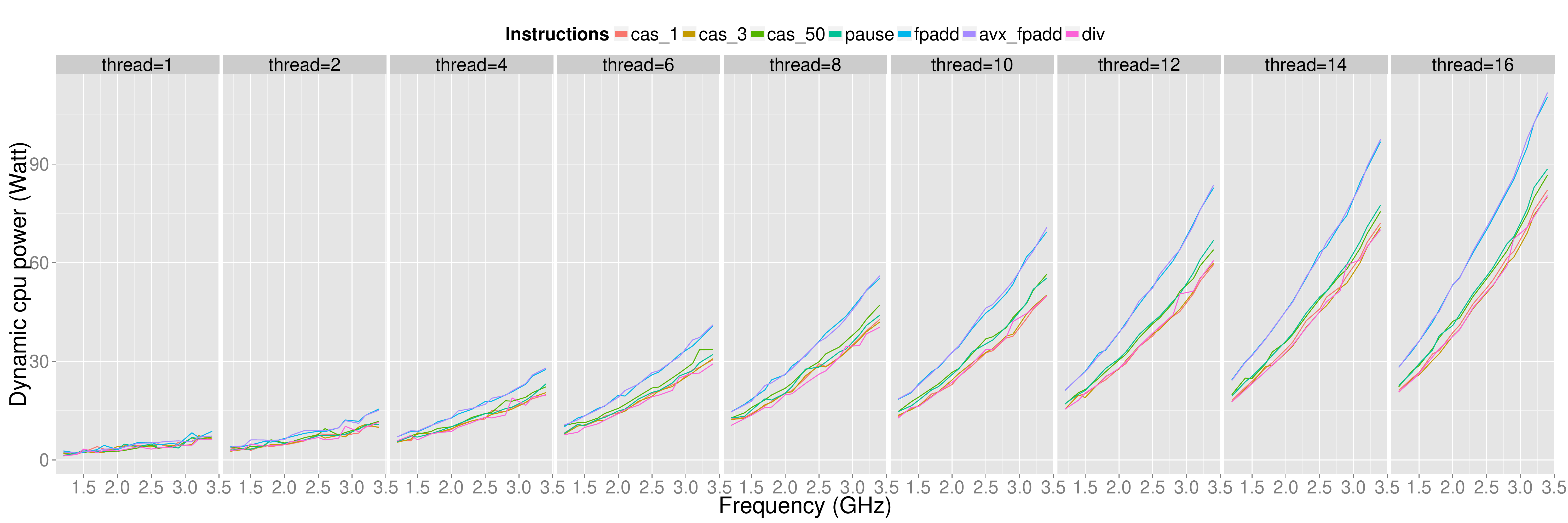}
\caption{Dynamic CPU power for micro-benchmarks\label{fig.cpu_dyn_micro}}
\end{figure}

\newcommand{\alf}{\ema{\alpha}}  
\newcommand{\A}{\ema{A}}       
\newcommand{\B}{\ema{B}}       

Having determined and excluded static and socket components, we obtain the dynamic power
component for each instruction, thread count, pinning and frequency setting. In the
micro-benchmarks of D1.1, a large variety of instructions are surveyed. Among them, we
pick a small set of instructions that can be representative for data structure
implementations, namely \cas, pause, floating point division, addition together with
vector addition.  \cas can be representative for the retry loops and divisions/additions
can be used to represent the parallel work which determines the contention on the data
structures. The decomposition of dynamic power in terms of CPU, memory and uncore
components for these instructions are illustrated in
Figures~\ref{fig.cpu_dyn_micro},~\ref{fig.mem_dyn_micro} and~\ref{fig.uncore_dyn_micro}.

Based on the observation that \pdyn{C} shows almost linear behavior with respect to number
of threads, we model the convex \pdyn{C} as:

\[ \pdyn{C}(\freq,\ope) =  \left( \A \times \freq^{\alf} + \B \right) \]

Each instruction might provide different power behavior as illustrated in Figure~\ref{fig.cpu_dyn_micro},
therefore we find \A, \B, \alf for each instruction separately. \B could be different for each instruction because
of the activation of different functional units, this is also why we included this constant in \pdyn{C}. 
\begin{table} [!bh]
\begin{center}
\begin{tabular}{|c|c|c|c|}
\cline{2-4}
\multicolumn{1}{c|}{}&        \A & \alf & \B \\\hline
cas   & 0.001392 & 1.6415 & 0.0510 \\\hline
fpdiv & 0.001038 & 1.7226 & 0.0585 \\\hline
add & 0.001004 & 1.8148 & 0.0912 \\\hline
avx-add & 0.001130 & 1.7828 & 0.0894 \\\hline
pause & 0.000854 & 1.7920 & 0.0736 \\\hline
\end{tabular}
\end{center}
\caption{Instruction power coefficients\label{fig.pow-coeff-micro}}
\end{table}

\newcommand{\vfreq}[1]{\ema{v^{(\mathit{freq})}_{#1}}}
\newcommand{\vmeas}[1]{\ema{v^{(\mathit{meas})}_{#1}}}
\newcommand{\vest}[1]{\ema{v^{(\mathit{est})}_{#1}}}

To obtain \A, \B and \alf, we proceed in the following way.  We are given an operation
\ope, and we consider the executions of this operation with $16$ threads on $2$
sockets.  Let \vfreq{} be the vector of frequencies where we want to estimate the dynamic
power (we dispose \nfr different frequencies, expressed in $\ghz{10^{-1}}$, such that
$\vfreq{1} = 12$ and $\vfreq{\nfr} = 34$). We note \vmeas{} the vector of dynamic powers
that have been computed from the measurements through the process described above, and
$\vest{}(A,B,\alf)$ the vector of estimated dynamic powers. More especially, for all $i
\in \{ 1,\dots,\nfr \}$:
\begin{align*}
 \vmeas{i} &= \pdyn{C}(\vfreq{i},\ope)\\
 \vest{i}(A,B,\alf) &= \left( A\times \left( \vfreq{i} \right)^\alf + B \right) 
\end{align*}
The Euclidean norm of a vector $v$ is denoted $\lVert v \rVert$.

We solve the following minimization problem, with the help of the Matlab ``fminsearch''
function:


\[ \min_{A,B,\alf} \left\lVert \vmeas{} - \vest{}(A,B,\alf) \right\rVert \]

Table~\ref{fig.pow-coeff-micro} provides the values for power constants and exponent for selected instructions.

\subsubsection{Dynamic Memory and Uncore Power}
\label{sec.mem-pow-cpu}

\begin{figure}[!b]
\includegraphics[width=\textwidth]{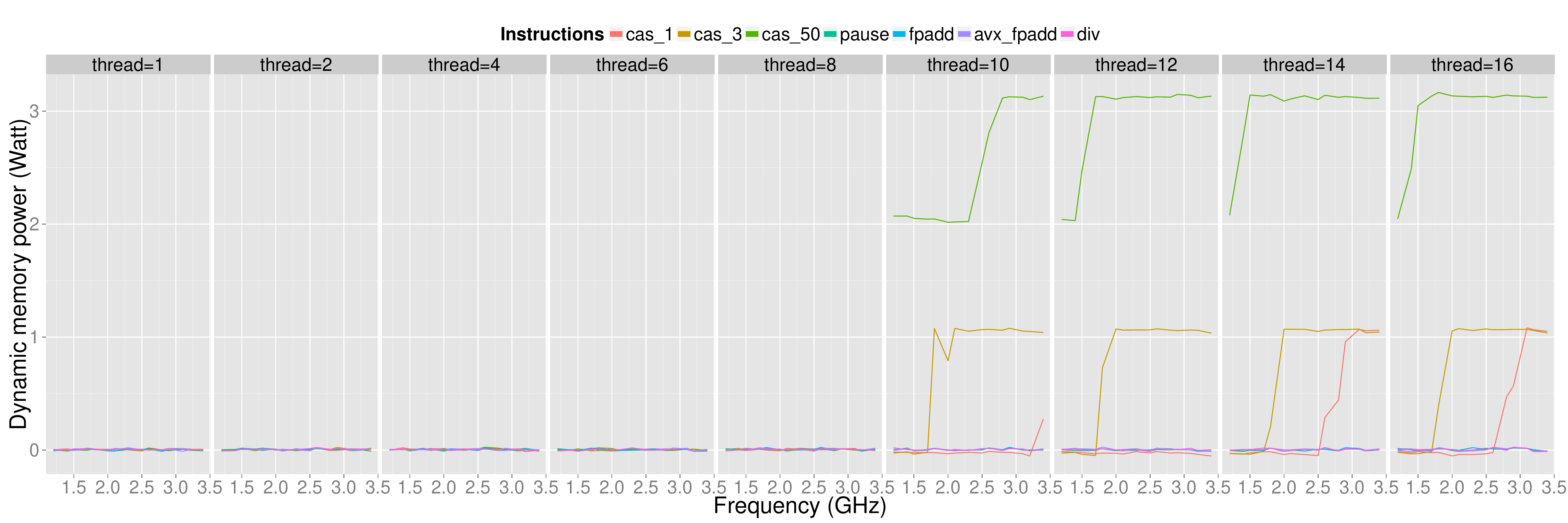}
\caption{Dynamic memory power for micro-benchmarks\label{fig.mem_dyn_micro}}
\end{figure}

\begin{figure}[t!]
\includegraphics[width=\textwidth]{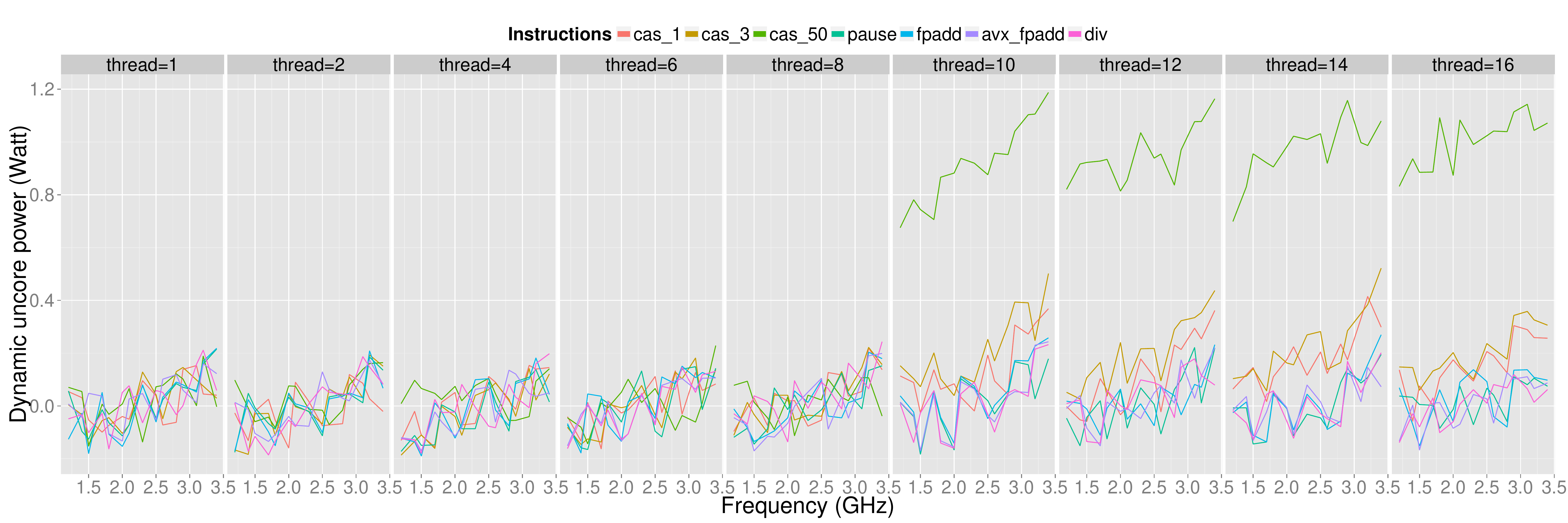}
\caption{Dynamic uncore power for micro-benchmarks\label{fig.uncore_dyn_micro}}
\end{figure}
\begin{figure} [t!]
\includegraphics[width=\textwidth]{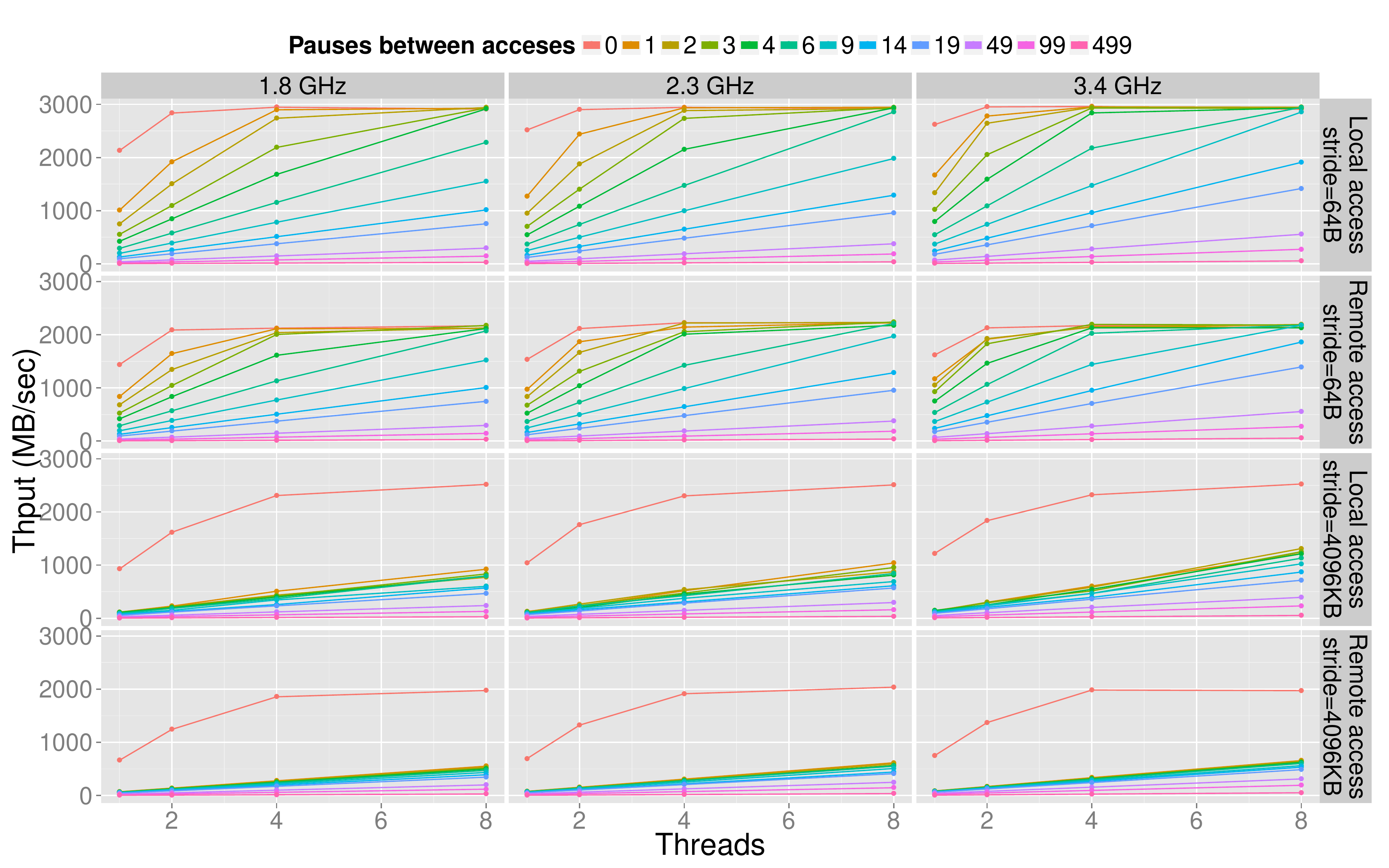}
\caption{Throughput for array traversal benchmark\label{fig.access-thput}}
\end{figure}

In the micro-benchmarks, we observe that many instructions do not lead to an increase in
dynamic memory and uncore power because the operands of the instructions,
except \cas, presumably reside in the core. On the other hand, \cas lead to an
increase in memory and uncore power only when the threads are pinned to different sockets. The
resulting {\it ping-pong} of the updated cache line between sockets is responsible for this effect. As
mentioned in D1.1, \cas micro-benchmarks are indeed prone to unfairness among threads. When \cas
is executed repeatedly by all threads on the same cache line without any work in between
\cas attempts, the thread which gets the ownership of cache line succeeds repeatedly while
others starve. This fact decreases the transfer rate of the cache line between local
caches. Due to this, we introduce 3 different \cas micro-benchmarks looping
on 1, 3 and 50 shared variables that are aligned to different cache lines. By doing so, we
aim at increasing the traffic between cores and sockets together with the amount of memory
accesses. Figures~\ref{fig.mem_dyn_micro} and~\ref{fig.uncore_dyn_micro} provide the
\pdyn{U} and \pdyn{M} values. It can be observed that all parameters including number of
threads, frequency, pinning play a role for \cas. As a remark for the provided figures,
threads are pinned using a dense mapping strategy that leads to inter-socket communication
only after 8 threads.

\begin{figure} [t!]
\includegraphics[width=\textwidth]{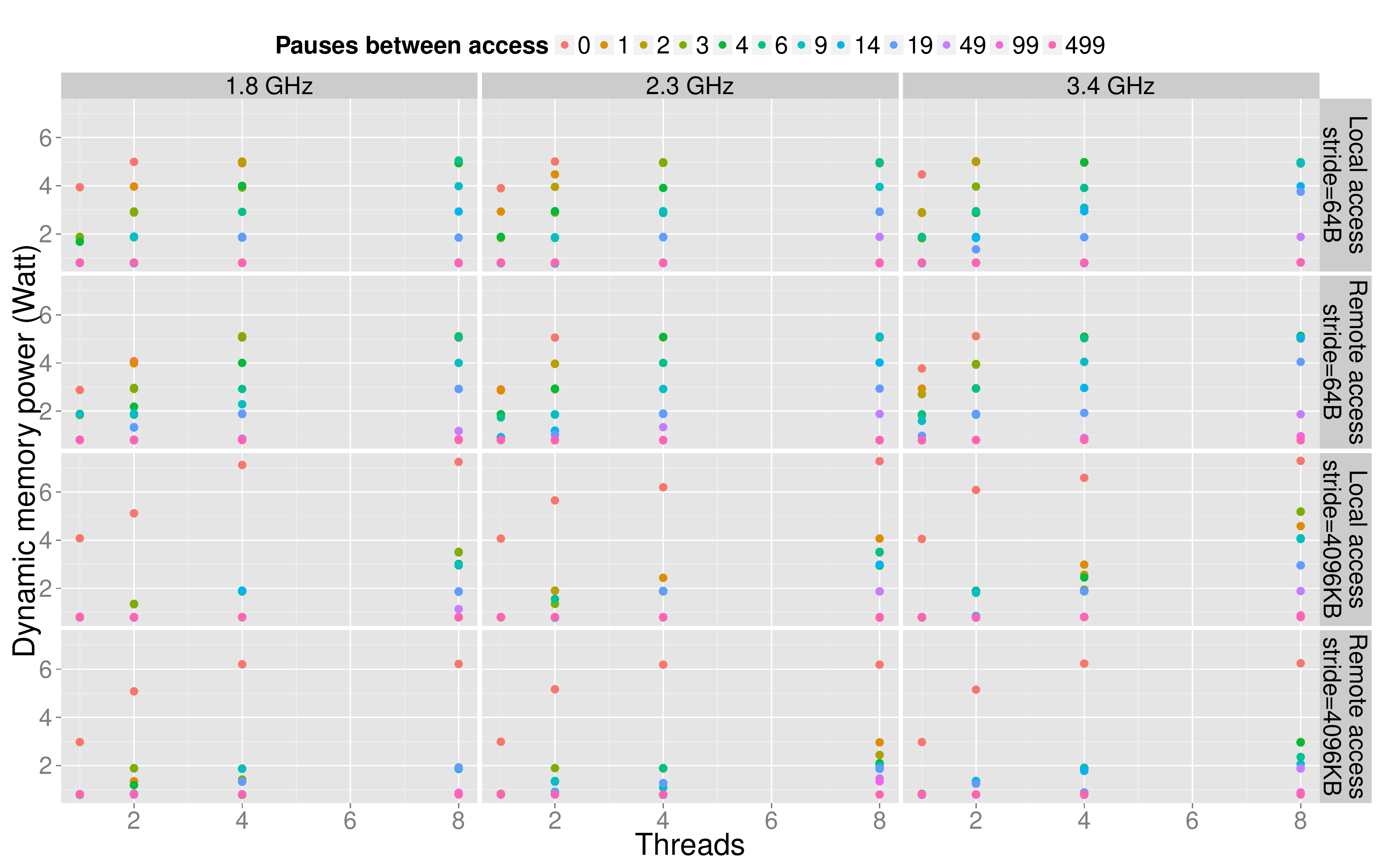}
\caption{Step-like power for array traversal benchmark\label{fig.access-work-mem}}
\end{figure}

\pdyn{U} and \pdyn{M} do not increase when threads are pinned to same socket. In this
case, the intra-socket communication between threads takes place via the ring interconnect
without introducing a memory access. Thus, absence of increase in memory power is
reasonable. However, one might expect an increase in uncore power for these cases because
RAPL uncore power presumably includes LLC and ring interconnect power. 
We do not observe this probably because the main components that attached to
the ring are not used. The increase of uncore and memory power can be observed when threads
are pinned to different sockets, due to remote memory accesses which uses important uncore
components such as the QPI link interface and home agent. An interesting observation
regarding memory power is that it shows a step function behavior. We think that this is
because of the RAPL power capping algorithm which determines a power budget based on
memory bandwidth, as presented in the work of David et al.~\cite{rapl-power-cap}. The RAPL
algorithm specifies a power cap for a time window depending on the memory bandwidth
requirements of previous time intervals and sets the memory in a power state that is
expected to maximize energy efficiency. Based on the amount of memory accesses, it jumps
between states finding a trade-off between bandwidth and power. The finite number of states
leads to the step-like power curves in Figure~\ref{fig.mem_dyn_micro}. Thus, the memory
power seems to be determined by the amount of memory accesses per unit of time which is dependent on
frequency, number of threads and the amount of shared variables for our \cas experiments.

\begin{figure} [t!]
\includegraphics[width=\textwidth]{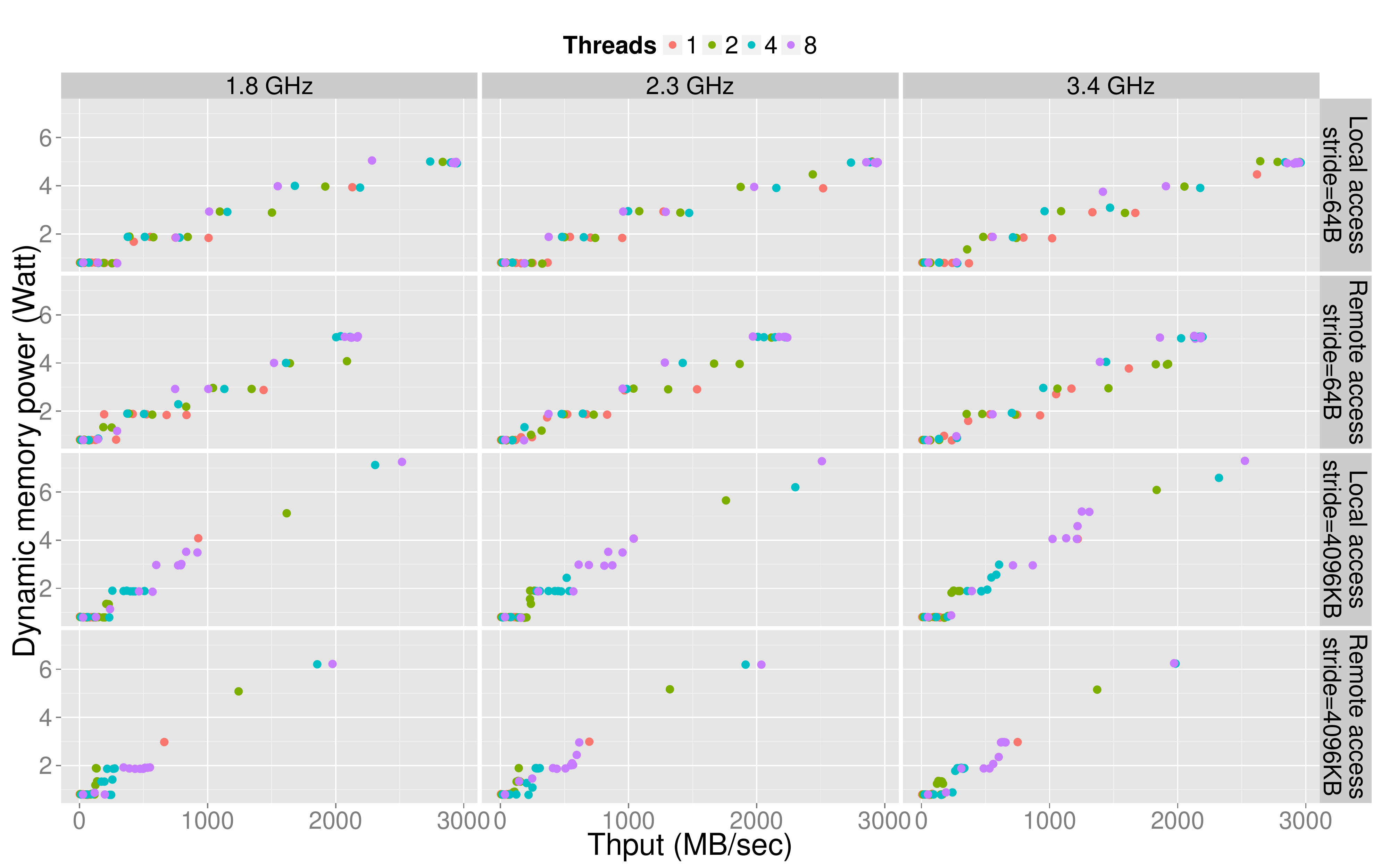}
\caption{Memory power for array traversal benchmark\label{fig.access-mempow}}
\end{figure}

\begin{figure} [h!]
\includegraphics[width=\textwidth]{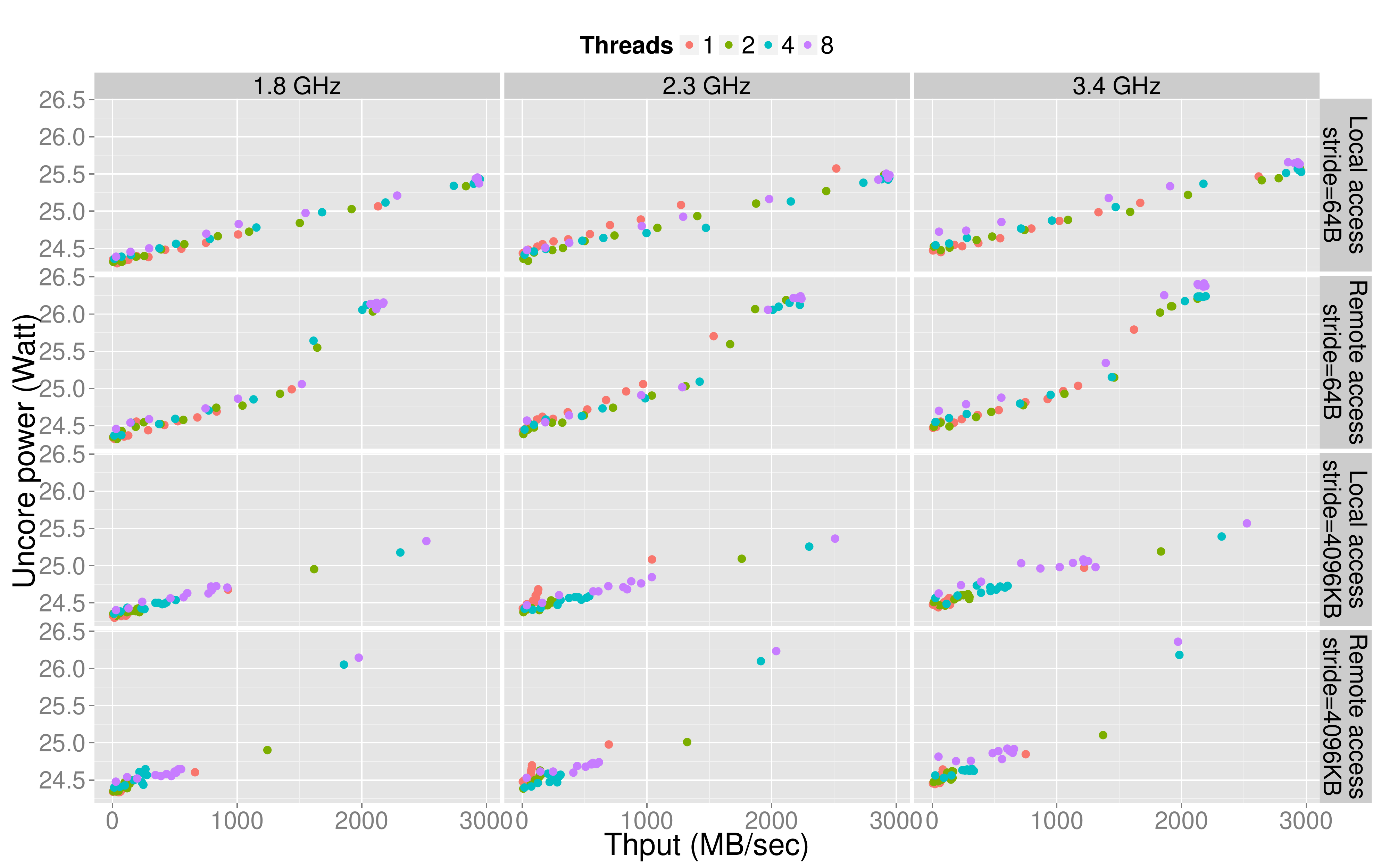}
\caption{Uncore power for array traversal benchmark\label{fig.access-uncorepow}}
\end{figure}

\newcommand{\byt}[1]{\ema{#1\;\text{Bytes}}}
\newcommand{\kbyt}[1]{\ema{#1\;\text{kBytes}}}

To justify this observation, we use a benchmark which stresses the main memory. We
allocate a huge contiguous array and align each element of the array to a separate
cache line. In addition, we force the array to be allocated in the memory module residing
in the first socket. Thus, we can regulate remote and local memory accesses by pinning
strategies. We pin all threads either to first or second socket. We also change the number of
threads, frequency and interleave varying amount of pause operations between array accesses to
change the bandwidth requirements of the benchmark. Threads access independent portions of the
array with a stride. The hardware prefetcher increases the performance remarkably when adjacent
cache lines are accessed while traversing the array and a stride of page size can be used
to disable the hardware prefetcher. We run the same experiment both with a stride of \byt{64},
which is the size of a cache line, and \kbyt{4096} which is the page size, to reveal effect of
prefetching. As provided in Figure~\ref{fig.access-thput}, the system reaches its peak
bandwidth more rapidly when the prefetcher is activated and attains better bandwidth. 
Moreover, the bandwidth difference between completely remote and
local accesses is noticeable. Another point is that frequency does not influence the maximum
achievable bandwidth. This fact means that there is 
opportunity for energy savings with DVFS for memory-bounded regions of applications.
We also observe the step-like power behavior, due to the RAPL algorithm, with this benchmark
in Figure~\ref{fig.access-work-mem}.

In Figure~\ref{fig.access-mempow}, dynamic memory power consumption is shown for the array traversal 
benchmarks. From the analysis of the results, it can be deduced that the memory power is strongly correlated with
the number of bytes accessed per second. There is no clear impact of the number of threads and frequency
to the memory power except their indirect effect on bandwidth. In contrast,
access stride has a direct, though limited, impact on the memory power together with its indirect impact as it increases
the bandwidth. By accessing data with a stride of a cache line, we possibly make use
of the open page mode of DRAM which could be influential in terms of energy efficiency due to
avoidance of bit-line precharge and row access cost. But, we still observe a linear relation between 
throughput and memory power for both strides.
In addition, remote or local accesses do not provide a noticeable difference for memory power. 
On the other hand, it is observed that uncore power depends on the frequency, presumably due to
the traffic on the ring interconnect and components attached to it such as the Home Agent and Integrated
Memory Controller. Furthermore, remote memory accesses increase the uncore power consumption
because they use the QPI link interface which adds an additional cost compared to local memory
accesses as shown in Figure~\ref{fig.access-uncorepow}.
All these observations regarding memory and uncore power will shed light to the analysis
and modeling of data structures in Section~\ref{sec:queue_modeling}. One major source of differences in power
consumption between different implementations is the memory and uncore consumption, which
is related to locality and bandwidth requirements of the implementations.

\subsubsection{Summary of Micro-Benchmarking for Power Modeling on CPU}
\newcommand{\da}{\ema{d}}        
\newcommand{\pca}{\ema{A}}       
\newcommand{\pcb}{\ema{B}}       
\newcommand{\pce}{\ema{\alpha}}  

Figure~\ref{fig.dep-rem} recalls the main achievements of the micro-benchmark study on CPU, 
where \da is the amount of memory accessed per unit of time in the main memory or through
QPI link.

\begin{figure}[h!]
\begin{tikzpicture}
\node (caca2) at (0,0)
{$\begin{array}{|c|c|c|c|}
\cline{2-4}
\multicolumn{1}{c|}{}&  \text{Static} & \text{Active} & \text{Dynamic}\\\hline
\text{CPU   } & \pstat{C}() & \sock \times \pact{C}(\freq) & \nth \times \left( \pca(\ope) \times \freq^{\pce(\ope)} + \pcb(\ope) \right)  \\\hline
\text{Memory} & \pstat{M}() &                          & \da \times \pdyn{M}(\ope,\loc) \\\cline{1-2}\cline{4-4}
\text{Uncore} & \pstat{U}() &                          & \da \times \pdyn{U}(\ope,\loc) \\\cline{1-2}\cline{4-4}
\end{array}$};
\node (caca1) at (0,5) 
{$\begin{array}{|c|c|c|c|}
\cline{2-4}
\multicolumn{1}{c|}{}&  \text{Static} & \text{Active} & \text{Dynamic}\\\hline
\text{CPU}    & \pstat{C}(\freq,\ope,\sock,\loc,\nth) & \pact{C}(\freq,\ope,\sock,\loc,\nth) & \pdyn{C}(\freq,\ope,\sock,\loc,\nth) \\\hline
\text{Memory} & \pstat{M}(\freq,\ope,\sock,\loc,\nth) & \pact{M}(\freq,\ope,\sock,\loc,\nth) & \pdyn{M}(\freq,\ope,\sock,\loc,\nth) \\\hline
\text{Uncore} & \pstat{U}(\freq,\ope,\sock,\loc,\nth) & \pact{U}(\freq,\ope,\sock,\loc,\nth) & \pdyn{U}(\freq,\ope,\sock,\loc,\nth) \\\hline
\end{array}$};
\path (caca1) edge[thick,->,bend left=45] node[right,text width = 2cm] {Dependency removal} (caca2);
\end{tikzpicture}
\caption{Dependency shrinking\label{fig.dep-rem}}
\end{figure}

\subsection{Energy Models for Movidius Embedded Platforms}\label{sec:movidius_energy}

\subsubsection{Description of Microbenchmarks}
Regarding the assembly files used in the test execution process, a fixed number of instructions in the loop was established for all the tests, meaning that each assembly file contains six instructions in the loop that is infinitely repeated. This convention was made in order to keep a continuity and a consistency of tests, by giving an insight in measuring the consumption on different SHAVE units, enabling to make comparisons between SHAVE units.
 
The assembly files used in the testing process  contain code that test the instruction power decode and the instruction fetch. The majority of tests use pseudo-realistic data, by pseudo-realistic data we understand having as many non-zero values as possible and avoiding data value repetition at different offsets. 

Below are the used test cases as micro-benchmarks for Movidius platform.

\begin{itemize}

\item{SauMul, SauXor}:
 SAU operations - with instruction fetch and different data values (XOR-MUL)

\item{IauMul, IauXor}:
 IAU operations - with instruction fetch and different data values (XOR-MUL)

\item{VauMul, VauXor}:
 VAU operations - with instruction fetch and different data values (XOR-MUL)

\item{CmuCpss, CmuCpivr}:
 CMU operations --- with instruction fetch and different data values (CPSS - CPIVR)

\item{SauXorCmuCpss}:
 SAU \& CMU --- with instruction fetch and different data values (SAU.XOR $\|$ CMU.CPSS)

\item{SauXorCmuCpivr}:
 SAU \& CMU --- with instruction fetch and different data values (SAU.XOR $\|$ CMU.CPIVR)

\item{SauXorIauXor, IauXorCmuCpss}:
 SAU \& IAU --- with instruction fetch and different data values (SAU.XOR $\|$ IAU.XOR \& IAU.XOR $\|$ CMU.CPSS)

\item{SauXorVauXor, SauXorVauMul}:
 SAU \& VAU --- with instruction fetch and different data values (SAU.XOR $\|$ VAU.XOR \& SAU.XOR $\|$ VAU.MUL)

\item{SauXorCmuIauXor}:
 SAU \& IAU \& CMU --- with instruction fetch and different data values 
 (SAU.XOR $\|$ CMU.CPIS $\|$ IAU.XOR)
 
\end{itemize}

\subsubsection{Movidius Power Model and Its Sanity Check} 
The experiments are conducted for benchmarks with single unit and multiple units. Each benchmark is tested by running Myriad1 with  1, 2, 4, 6 and 8 SHAVE cores.

From the experiment results, we observe that the power consumption of Movidius Myriad1 platform is ruled by the following model:
\begin{equation} \label{eq:Pfirst}
	P=P^{stat}+\#\{\text{active SHAVE}\} \times (P^{act}+P^{dyn}_{SHAVE})
\end{equation}
The operands in the formula are explained as below.
The static power $P^{stat}$ is the needed power when the Myriad1 processor is on. The $P^{act}$ is the power consumed when a SHAVE core is on. Therefore, this active power is multiplied with the number of used SHAVE core when the benchmark is run with several cores. 

The dynamic power $P^{dyn}_{SHAVE}$ of each SHAVE is the power consumed by all working operation units working on SHAVE. As described in the previous section, each SHAVE core has several components, including LSU, PEU, BRU, IAU, SAU, VAU, CMU, etc. Different arithmetic operation units have different $P^{dyn}$ values. In this power model, we focus on the experiments with the benchmarks performing different arithmetic operations such as IAU, VAU, SAU and CMU.  

The dynamic power $P^{dyn}$ of each SHAVE is the power consumed by all working operation units. As described in the previous section, each SHAVE core has several components, including LSU, PEU, BRU, IAU, SAU, VAU, CMU, etc. Different arithmetic operation units have different $P^{dyn}$ values. In this power model, we focus on the experiments with the benchmarks performing different arithmetic operations such as IAU, VAU, SAU and CMU.  


When  adding one more SHAVE, we can identify the sum of SHAVE $P^{act}$ and $P^{dyn}$ which is the power level difference of the two runs (with one SHAVE core and with two SHAVE corer).  Given the sum of $P^{act}$ and $P^{dyn}$, $P^{stat}$ is derived from the formula. Then, the average value of $P^{stat}$ from all micro-benchmarks experimental results is computed:
\begin{equation} \label{eq:Pstat}
	P^{stat}=62.63 \text{ mW}
\end{equation}
Then for each operation unit, we obtain the two parameters $P^{dyn}_{op}$ and $P^{act}$ by using the actual power consumption of the benchmark for individual units and multiple units. 
\begin{equation} \label{eq:Pcomb}
	\begin{cases}
		P^{dyn}_{(IauXor)} =P^{stat} + \#\{\text{active SHAVE}\} \times (P^{act}+P^{dyn}_{(IauXor)}) \\
		P^{dyn}_{(SauXorIauXor)}= P^{stat} + \#\{\text{active SHAVE}\} \times (P^{act}+P^{dyn}_{(SauXorIauXor)})\\
		P^{dyn}_{(SauXor)}= P^{stat} + \#\{\text{active SHAVE}\} \times (P^{act}+P^{dyn}_{(SauXor)})
	\end{cases}
\end{equation}

Then, the average value of  $P^{act}$ among all operation units is calculated.
\begin{equation} \label{eq:Pact}
	P^{act}=51.4 \text{ mW}
\end{equation}

At this point, we also have the $P^{dyn}_{op}$ of every arithmetic unit $op$. Applying the model to calculate the consumed energy, the results showed the deviation from the measured data, especially the benchmarks running a single arithmetic unit. We attribute this difference to the inter-operation cost when more than one unit of the SHAVE core work in a combination. This inter-operation cost is also mentioned in the model suggested by Movidius in EXCESS deliverable D4.1. The $P^{dyn}$ of SHAVE running multiple units is then computed by the formula below:
\begin{equation} \label{eq:Pdynshave}
	P^{dyn}_{SHAVE}=\sum_i{P^{dyn}_i(op)} + \max_i\{O_i(op)\}
\end{equation}

The combined $P^{dyn}$ is the sum of $P^{dyn}$ of each unit, plus the highest inter-operational cost $O_i$ among the units in the combination. By using the highest inter-operation cost, $P^{dyn}$ after computed is more accurate than using the sum of  inter-operation cost from all units in the combination. E.g. 
$P^{dyn}_{(SauXorIauXor)}) = P^{dyn}_{(SauXor)} + P^{dyn}_{(IauXor)} + \max(O_{SauXor}, O_{IauXor})$.

Given the $P^{stat}$ and $P^{act}$, the $P^{dyn}_{op}$ of an operation unit is computed based on the actual power consumption of the benchmark using a single unit(e.g. SauXor, IauXor, etc.). Then, its inter-operational cost is computed based on the actual power consumption when this unit works in a combination with other units.

The Table~\ref{table:InteroperationalCost} lists the inter-operational cost of each unit when it works in a combination with other units.

\begin{table}[h]
\begin{center}
\begin{tabular}{|l|l|l|}
\hline
\textbf{Operation Unit} & \textbf{$P^{dyn}_{op}$ (mW)} & \textbf{$O_{op}$(mW)} \\ \hline
SauXor                  & 3.05          & 1.15       \\ \hline
SauMul                  & 6.97          & 1.83       \\ \hline
VauXor                  & 17.57         & 13.12      \\ \hline
VauMul                 & 32.78         & 11.62      \\ \hline
IauXor                  & 4.53          & 1.07       \\ \hline
IauMul                  & 3.98          & 4.42       \\ \hline
CmuCpss                 & 1.00          & 4.60       \\ \hline
CmuCpivr                & 6.41          & 5.69       \\ \hline
\end{tabular}
\caption{Inter-operational cost (in mW) of SHAVE operation units}
\label{table:InteroperationalCost}
\end{center}
\end{table}

From Equations~\ref{eq:Pstat},~\ref{eq:Pact} and~\ref{eq:Pdynshave}, the complete model for Movidius Myriad1 is derived as follows:
\begin{equation} \label{eq:Psecond}
	P=P^{stat}+ \#\{\text{active SHAVE}\} \times (P^{act}) + \#\{\text{active SHAVE}\} \times \left( \sum_i{P^{dyn}_i(op)} + \max_i\{O_i(op)\} \right)
\end{equation}

Applying this formula to different combinations of operation units in the SHAVE core, we plot the relative error of this model in the Figure~\ref{fig:ErrorRelative}. The relative error is the difference between the actual power consumption measure by device and the predicted power consumption computed through Equation \ref{eq:Psecond}, then divided by the actual power consumption. 
Under this model, the relative error varies within $\pm 4\%$. This model is not only applicable for a single unit but also the combination of two or three units.

\begin{figure}
\begin{center}
\includegraphics[width=0.9\textwidth]{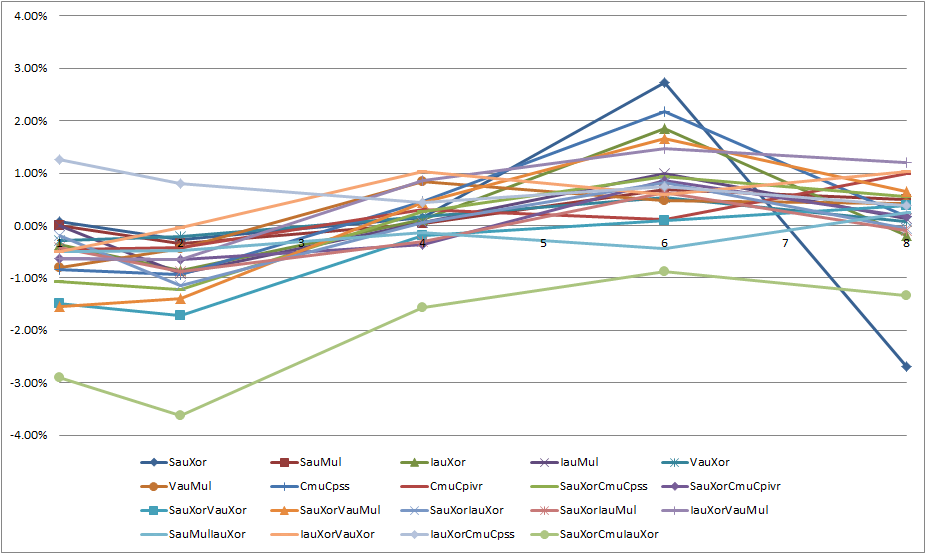}
\end{center}
\caption{Relative error of Myriad1 power prediction}
\label{fig:ErrorRelative}
\end{figure}

\section{Modeling Energy Consumption of Concurrent Queue Implementations} \label{sec:queue_modeling} \label{sec:concurrent-data-structures}

\subsection{Concurrent Queues on CPU-based Platform}
\label{sec:concurrent-data-structures-CPU}
\def\ofs{^{\text{(off)}}} 
\def\ons{^{\text{(on)}}}

\newcommand{\pwi}{\ema{\pw_\infty}}
\newcommand{\tpsi}{\ema{t_{\text{PS},\infty}}}

\newcommand{\thri}{\ema{\thr_\infty}}
\newcommand{\pwlc}{\ema{\pw_{lc}}}
\newcommand{\cwlc}{\ema{\cw_{lc}}}
\newcommand{\thrlc}{\ema{\thr_{lc}}}
\newcommand{\cwlcons}{\ema{\cwlc\ons}}
\newcommand{\thrlcons}{\ema{\thrlc\ons}}
\newcommand{\tcslc}{\ema{t_{\text{RL,LC}}}}
\newcommand{\tcslcons}{\ema{\tcslc\ons}}
\newcommand{\tcslcofs}{\ema{\tcslc\ofs}}
\newcommand{\cwlcofs}{\ema{\cwlc\ofs}}
\newcommand{\thrlcofs}{\ema{\thrlc\ofs}}

\newcommand{\frep}{\ema{\freq_0}}
\newcommand{\frepp}{\ema{\freq_1}}

\newcommand{\tpscr}{\ema{\tps^{(crit)}}}

\newcommand{\nthn}{\ema{\nth_0}}
\newcommand{\nthf}{\ema{\nth_1}}

\remind{(PT) [Re-write a little]
A common approach to parallelizing applications is to divide the problem into separate
threads that act as either producers or consumers. The problem of synchronizing these
threads and streaming of data items between them, can be alleviated by utilizing a shared
collection data structure.}

Concurrent FIFO queues and other producer/consumer
collections are fundamental data structures that are key components in
applications, algorithms, run-time and operating systems.
\index{concurrent data structures!queue}
\index{queue}
The Queue abstract data type is a collection of items in which only the earliest added
item may be accessed. Basic operations are \op{Enqueue} (add to the tail) and \op{Dequeue}
(remove from the head). Dequeue returns the item removed. The data structure is also known
as a ``first-in, first-out'' or FIFO buffer.

\subsubsection{Objective and Process}


As explained in Section~\ref{sec:int-ds}, we aim at predicting the energy efficiency of different
queue implementations through several metrics. We have seen in D1.1~\cite{D1.1} that the
energy efficiency is strongly related to both performance of the considered algorithm and
power dissipation of the architecture. Hence we naturally decompose this blurry notion of
energy-awareness, or energy efficiency, into these two dimensions.

We study the problem by modeling the behavior of the implementations, from both
the performance and the power point of view. In both cases, we run the implementations on some
problem inputs, and measure performance and power dissipation so that we can instantiate
the parameters of the model. Then, once the model is instantiated we are able to predict
the queue implementation's energy efficiency on any problem instance.

On the one hand, the performance-related metric that we use in this deliverable is
throughput, which represents the number of successful operations per unit of time (here,
per second). We consider the dequeuing of an element from the queue as the successful
operation, and the measurement is a simple counter of successful operations. On the other
hand, the power dissipation measurement is done through Intel's RAPL energy counters
read via the PAPI
library~\cite{BrDoGaHoMu:2000:PAPI,Weaver:2012:MEP:2410139.2410475}, 
as explained in Section~\ref{sec.gen-mod-cpu}.

Finally, we combine those two metrics into an energy-efficiency-related one: energy
consumed per successful operation, which is the ratio between power dissipation and
throughput.

\begin{figure}
\framebox{\includegraphics[width=\textwidth]{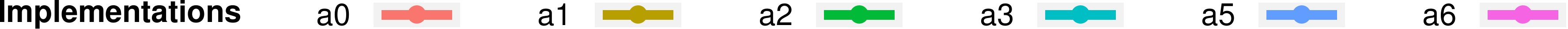}}
\caption{Key legend of the graphs\label{fig.key}}
\end{figure}

We dispose of a framework in which we have implemented the most well-known queue algorithm of
the literature. We consider the following implementations of queues,
which are described in some detail in Section~\ref{sec:shm-queue-algs} below:%
\footnote{The~{\bf a0}, {\bf a1}, etc., designations refer to the
micro-benchmark's naming of the algorithms.
Algorithm~{\bf a4} is a single-producer/single-consumer queue that is
unsuitable for our purposes and has therefore been left out.}
\begin{itemize}
\item {\bf a0}. Lock-free and linearizable queue by Michael and
  Scott~\cite{Michael96},
\item {\bf a1}. Lock-free and linearizable queue by Valois~\cite{Val94},
\item {\bf a2}. Lock-free and linearizable queue by Tsigas and
  Zhang~\cite{TsiZ01b},
\item {\bf a3}. Lock-free and linearizable queue by Gidenstam et
  al.~\cite{Gidenstam10:OPODIS},
\item {\bf a5}. Lock-free and linearizable queue by Hoffman et
  al.~\cite{DBLP:conf/opodis/HoffmanSS07},
\item {\bf a6}. Lock-free and linearizable queue by Moir et
  al.~\cite{MoirNSS:2005:elim-queue}.
\end{itemize}
We use the same legend for all the graphs in this section, except
from Figure~\ref{fig.cas}. It is depicted in Figure~\ref{fig.key}.
The idea here is to use as little knowledge as possible about the algorithms to predict
their behaviors, so that if a new algorithm is implemented its behavior can be
predicted as well, without changing the model that we present in the following sections.

We run a simple benchmark composed of the two functions described in
Figure~\ref{alg.ben}. Half of the threads are assigned to be enqueuers while the remaining
ones are dequeuers. We disable multi-threading and map separate threads into separate
cores, also the number of threads never exceeds the number of cores. In addition, the
mapping is done in the following way: when adding an enqueuer/dequeuer pair, they are
both mapped on the most filled but non-full socket.

The parallel work shall be seen as a processing activity, pre-processing for the enqueuers
before it enqueues an element, and post-processing on an element of the queue for the
dequeuers. This activity is presumed to be computation-intensive, that is why we coded it
in the benchmark as a sequence of floating point divisions.

As explained previously in Section~\ref{sec:int-ds}, the benchmark represents an application that
uses the queue in a steady-state manner; however, the behavior of the queue is likely to
differ from one application to the other, according to the amount of work in
the \ps; also, in the experiments, this amount of work will be part of the variables.

Two more clarifications are necessary. On the one hand, when we speak about
implementations of the queues, we actually refer to the different implementations of
enqueuing and dequeuing functions, along with their memory management schemes.  On the
other hand, the slowest of the two function calls (enqueue and dequeue) is the bottleneck
of performance and hence determines the throughput of the queue. Also when we reason about
the ``\rl'' in the following, we imply ``\rl of the slowest function call''. Notice that
only half of the threads are competing for this operation.


We conclude this introduction by defining some parameters that we use extensively in the
next subsections. We denote by \nth the number of running threads that call the same
operation, and by \freq the clock frequency of the cores (we only consider the case where
all cores share the same clock frequency, even if an asymmetric setting of the frequencies
between the sockets could be of interest). We note \pw the amount of work in the \ps, and
\cw the amount of work in one try of the \rl in the considered implementation.

More generally, an exponent ``(off)'' refers to an inter-socket execution, while ``(on)''
refers to an intra-socket one; a subscript ``lc'' denotes an execution in low-contention
state.

In the following subsections all experiments and their underlying predictions are done on Chalmers'
platform, which is described in Section~\ref{sec:sysDscrA}. We run the implementations
within a set of three frequencies $\left\{ \ghz{1.2}, \ghz{2.3}, \ghz{3.4} \right\}$, for
all possible even total numbers of threads, from 2 to 16, \ie for $\nth \in \{1,\dots,8\}$.

\begin{figure}[h!]
\begin{minipage}{.45\textwidth}
\begin{algorithmic}[1]
\Procedure{Enqueuer}{}
\State Initialization()\;
\While{execution time $<$ t}
	\State Parallel\_Work()\;
	\State Enqueue()\;
\EndWhile
\EndProcedure
\end{algorithmic}
\end{minipage}\hfill%
\begin{minipage}{.45\textwidth}
\begin{algorithmic}[1]
\Procedure{Dequeuer}{}
\State Initialization()\;
\While{execution time $<$ t}
	\State $\mathit{res} \leftarrow \text{Dequeue()}$\;
        \If{$\mathit{res} \neq \textsc{Null}$}
        	\State Parallel\_Work()\;
        \EndIf
\EndWhile
\EndProcedure
\end{algorithmic}
\end{minipage}%
\caption{Queue benchmark\label{alg.ben}}
\end{figure}

\subsubsection{Throughput}

We start this section by underlining some interactions between what we called ``work''
previously, and the actual execution time of those pieces of code according to different
parameters. Then we describe the throughput model under two distinct states that the queue
can experience, and finally, we exhibit our results.

\paragraph{Preliminaries}

First of all, we have seen that the \ps is full of computations, thus the amount of work
in it is actually the number of bunches of $10$ floating point divisions that we operate;
those operations are perfectly scalable, meaning that the time \tps spent in a given \ps
is proportional to $\frac{\pw}{\freq}$.

In order to obtain a stable execution time, which is linear with the number of bunches of
divisions despite compiler and runtime optimizations, we have used CPUID and RDTSCP
assembly primitives. As a consequence, the multiplicative factor that correlates \tps and
$\pw/\freq$ depends on the number of threads that are running on the platform. Finally,
the time spent in a \ps can be computed thanks to:
\begin{equation}
\tps = \frac{\pw}{\facf \times \freq}, \label{eq.tps}
\end{equation}
where \facf depends on the number of threads.

The execution time of dequeue and enqueue operations is more problematic, for three main
reasons. Primo, because of the lock-free nature of the implementations. From a high-level
perspective, those two functions calls are both retry-loops: the thread reads a data, then
works with this version of the data, modifies it, and finally tries to operate a \cas on
it. If the \cas fails, then it reads it again, and re-iterates the process. It exits the
function when the \cas is a success. As the number of retries is unknown, the time spent
in the function call is not straightforwardly computable. This behavior leads us to
distinguish two cases: low-contention case, where we are ensured that no \rl will fail,
and high-contention case, where the threads will generally fail one or several times
before succeeding.  Secundo, in the high-contention case, the threads compete for
accessing a shared data, and they wait for some time before actually being able to access
data. We name this as the \textit{expansion}, as it leads to an increase in the execution
time of one try of the \rl. Tertio, the time before obtaining the data changes, depending
whether the data is located in the same socket. This last pathology is however benign, if
we look at the following experiment.

\newcommand{\accns}{\ema{a}}
\newcommand{\accfs}{\ema{a'}}
\newcommand{\bccfs}{\ema{b'}}

\begin{figure}[h!]
\includegraphics[width=\textwidth]{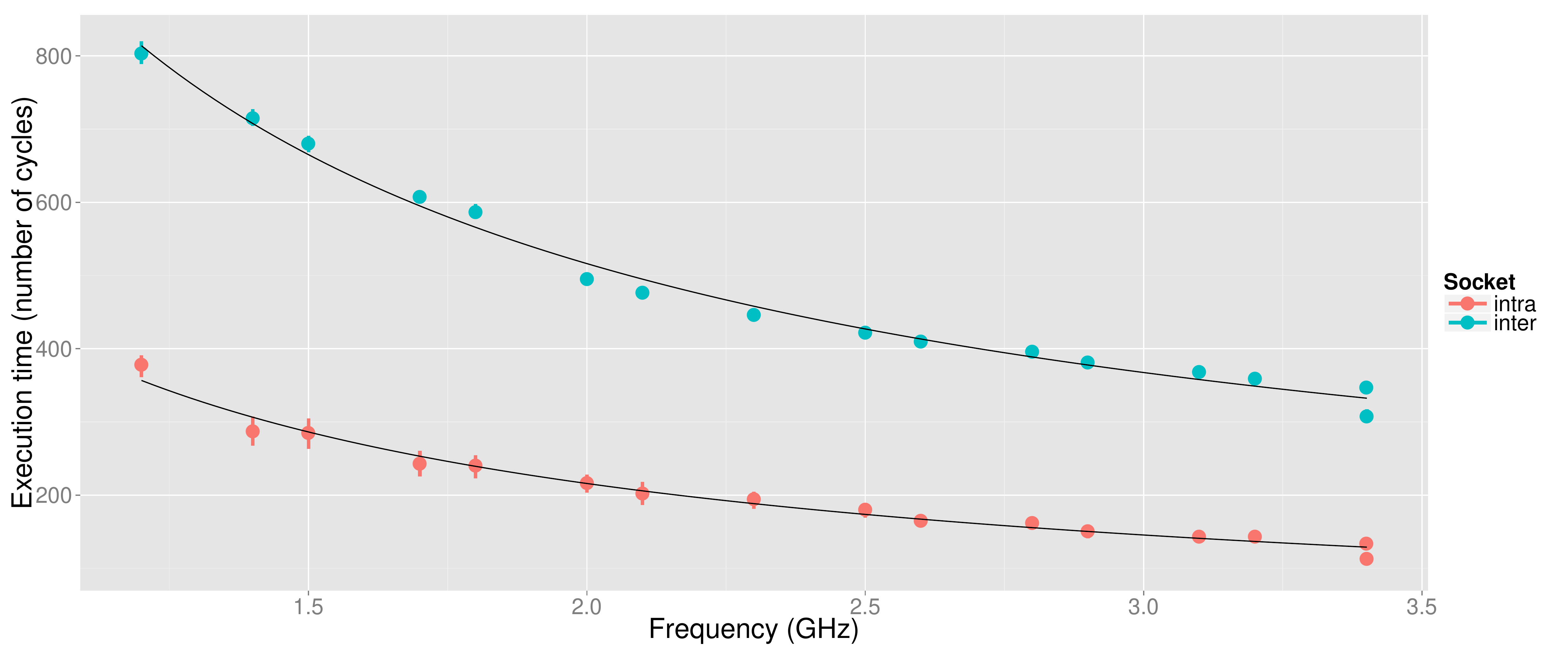}
\caption{Execution time of \cas\label{fig.cas}}
\end{figure}

We envision the approximation where the \rl is a mixed sequence of \cas and other shared
memory accesses, which is supposedly not too far from the reality. We have measured the
execution time of a \cas operation, on the one hand when the data is initially located in
the same socket as the requester, and on the other hand when the data is in the other
socket. In Figure~\ref{fig.cas}, we plot the execution time according to the clock
frequency. On-socket, the cost can be fitted with a function $\freq \mapsto \accns/\freq$,
while the cost of an off-socket access is fitted by $\freq \mapsto \accfs/\freq
+ \bccfs$. In other words, the off-socket access includes a non-scalable component that
the QPI link is responsible for.


As a consequence, in the low-contention case, \ie when we know that the function call
contains only one single try of the \rl and that there is no expansion, if we assume
that \cw is the equivalent of the number of \cas inside the \rl, we have:
\begin{equation}
\left\{ \begin{array}{l}
 \tcs = \cw \times \frac{\accns}{\freq}, \text{   if there are not more than 8 cores, and}\\
 \tcs = \cw \times \left(\bccfs + \frac{\accfs}{\freq} \right), \text{   otherwise.}
\end{array} \right.
\label{eq.tcslc}
\end{equation}



\paragraph{Low Contention}

We study in this section the low-contention case, \ie when (i) the threads does not suffer
from expansion and (ii) a success is obtained with a single try of the \rl. As it appears
on the scheme in Figure~\ref{fig.sch-lc}, we have a cyclic execution, and the length of
the shortest cycle is $\tps + \tcs$. Within each cycle, every thread performs exactly one
successful operation, thus the throughput is easy to compute thanks to:
\begin{equation}
\thr = \frac{\nth}{\tps + \tcs}. \label{eq.thr-lc}
\end{equation}
This model includes two parameters: \facf, which depends only on the number of threads
\nth, and \cw (through Equation~\ref{eq.tcslc}), which depends only on the implementation.
The constants of Equation~\ref{eq.tcslc} can indeed be determined beforehand.

\medskip

We determine \facf by running anyone of the implementations with a very large parallel
work \pwi, at a given frequency \frep, for every number of threads. We note \tpsi the
execution time of the \ps of size \pwi at the frequency \frep, we measure the
throughput \thri and approximate the Equation~\ref{eq.thr-lc} with
$\thri= \frac{\nth}{\tpsi}$, since the execution time of the \ps \tpsi is such that
$\tpsi \gg \tcs$.  Therefore, for each \numth, we can obtain \facf from
Equation~\ref{eq.tps}, with the following equation:
\[ \facf = \frac{\thri}{\nth} \times \frac{\pwi}{\frep}. \]

\medskip

Concerning the amount of work \cw in the \rl, we have observed that the model is very
sensitive to \cw, which is why we consider that the amount of work in the \rl differs from
an intra-socket to an inter-socket execution. We note \cwlcons the former one and \cwlcofs
the latter one.
Those two values are obtained by running each implementation in low-contention state.  In
other words, we pick an amount of parallel work \pwlc, which is big enough so that the
queue is lowly congested. At frequency \frep, we run the implementation once with \nthn
threads such that $2\nthn>8$ (leading to the throughput measurement \thrlcofs), and once
with \nthf threads such that $2\nthf\leq 8$ (leading to the throughput
measurement \thrlcons). Then the system in Equation~\ref{eq.tcslc} implies that:
\[\arraycolsep=1.4pt\def\arraystretch{2.2} \left\{\begin{array}{l}
 \thrlcons =   \dfrac{\nthn}{ \dfrac{\pwlc}{\facf\times\frep} + \cwlcons \times\dfrac{\accns}{\frep} }  \\
 \thrlcofs =  \dfrac{\nthf}{\dfrac{\pwlc}{\facf\times\frep} + \cwlcofs \times \left( \bccfs+\dfrac{\accfs}{\frep} \right)}
\end{array}\right. \quad,\quad\text{hence} \]
\[\arraycolsep=1.4pt\def\arraystretch{2.2} \left\{\begin{array}{l}
 \cwlcons = \dfrac{\nthn \times \frep}{\accns \times \thrlcons} - \dfrac{\pwlc}{\accns \times \facf}\\
 \cwlcofs = \dfrac{1}{\bccfs+\accfs/\frep} \times 
\left( \dfrac{\nthf}{\thrlcofs} - \dfrac{\pwlc}{\facf \times \frep} \right)
\end{array}\right. \]

\bigskip

Finally, given a frequency \freq, a parallel work \pw, and number of threads \nth, the
evaluation of the throughput in low-contention state is done thanks to:
\begin{equation}
\left\{
\arraycolsep=1.4pt\def\arraystretch{2.2} \begin{array}{ll}
\thr = \dfrac{\nth}{\dfrac{\pw}{\facf \freq} + \dfrac{\accns \times \cwlcons}{\freq}} &\quad\text{if~}2\nth<9\\
\thr = \dfrac{\nth}{\dfrac{\pw}{\facf \freq} 
+ \cwlcofs \times \left( \bccfs + \accfs/\freq \right)}&\quad\text{if~}2\nth>8\\
\end{array}
\right. \label{eq.thr-hc}
\end{equation}

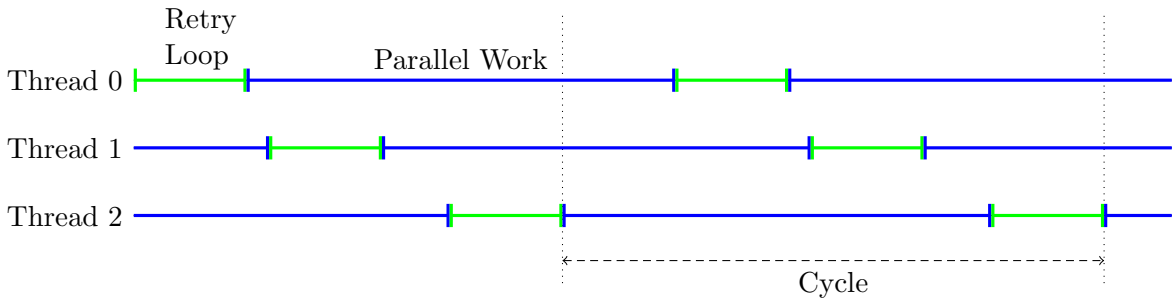
\begin{figure}[h!]
\begin{center}
\begin{tikzpicture} [scale=0.6, font=\small]
\draw[dotted] (9.5,1.5) -- ++ (0,6);
\draw[dotted] (21.5,1.5) -- ++ (0,6);
\node [above, font=\small, text width=2, align=center] at (0.75,6) {Retry Loop};
\node [above, font=\small] at (7.25,6) {Parallel Work};
\draw [<->,densely dashed] (9.5,2)--(21.5,2) node[midway,below,font=\small] {Cycle};

\node [left] at (0,6) {Thread 0 };
\draw[dashed] (0,6) -- ++(23,0);
\draw[very thick,green,|-|] (0,6) -- ++(2.5,0);
\draw[very thick,blue,|-|] (2.5,6) -- ++(9.5,0);
\draw[very thick,green,|-|] (12,6) -- ++(2.5,0);
\draw[very thick,blue,|-] (14.5,6) -- (23,6);

\node [left] at (0,4.5) {Thread 1 };
\draw[dashed] (0,4.5) -- ++(23,0);
\draw[very thick,blue,-|] (0,4.5) -- (3,4.5);
\draw[very thick,green,|-|] (3,4.5) -- ++(2.5,0);
\draw[very thick,blue,|-|] (5.5,4.5) -- ++(9.5,0);
\draw[very thick,green,|-|] (15,4.5) -- ++(2.5,0);
\draw[very thick,blue,|-] (17.5,4.5) -- (23,4.5);

\node [left] at (0,3) {Thread 2 };
\draw[dashed] (0,3) -- ++(23,0);

\draw[very thick,blue,-|] (0,3) -- (7,3);
\draw[very thick,green,|-|] (7,3) -- ++(2.5,0);
\draw[very thick,blue,|-|] (9.5,3) -- ++(9.5,0);
\draw[very thick,green,|-|] (19,3) -- ++(2.5,0);
\draw[very thick,blue,|-] (21.5,3) -- (23,3);

\end{tikzpicture}
\end{center}
\caption{Cyclic execution under low contention\label{fig.sch-lc}}
\end{figure}

\paragraph{High Contention}

In the preliminaries, we have explained why the evaluation of the throughput is complex
when contention is high: because of the expansion that changes the execution time of one
try of the \rl, and because the number of tries before a success in the \rl is variable.

However, in previous studies, we have seen that the throughput is approximately linear with
the expected number of threads that are in the \rl at a given time. In addition, this
expected number is almost proportional to the amount of work in the \ps. As a result, a
good approximation of the throughput, in high-congestion cases, is a function that is
linear with the amount of work in the \ps.

There remains that the way the threads interfere in the chip, hence the relation between
the slope of this straight line and the different parameters, is very dependent on the
architecture. That is why, for each frequency, each number of threads and each
implementation, we interpolate this line by measuring the throughput for two small amounts
of work in the \ps.


\paragraph{Frontier}

We now have to estimate when the queue is highly congested and when it is not. We recall
that, generally speaking, long parallel sections lead to a low-congested queue since
threads are most of the time processing some computations and do not try to access to the
shared data. Reversely, when the \ps is short, the ratio of time that threads spend in the
retry-loop is higher, and gets even higher because of both expansion and retries.

That being said, there exists a simple lower bound on the amount of work in the \ps, such
that there exists an execution where the threads are never failing in their \rl. Let us
note \tcslc the execution time of the \rl in low-contention case (we recall that we are
able to compute this value as we know the amount of work in the \rl), and its relation with
the clock frequency.
We plot in Figure~\ref{fig.sch-crit} an ideal execution with $\nth = 3$ threads and
$\tps=(\nth-1)\times\tcslc$. In this execution, all threads always succeed at their first
try in the \rl. Nevertheless, if we make the \ps shorter, then there is not enough
parallel potential any more, and the threads will start to fail: the queue enters the
high-congested state.

\begin{figure}[h!]
\begin{center}
\begin{tikzpicture} [scale=0.6, font=\small]
\draw[densely dotted] (0,1.5) -- ++ (0,6);
\draw[densely dotted] (2.5,1.5) -- ++ (0,6);
\draw[densely dotted] (5,1.5) -- ++ (0,6);
\draw[densely dotted] (7.5,1.5) -- ++ (0,6);
\draw[densely dotted] (10,1.5) -- ++ (0,6);
\draw[densely dotted] (12.5,1.5) -- ++ (0,6);
\draw[densely dotted] (15,1.5) -- ++ (0,6);

\node [left] at (0,6) {Thread 0 };
\draw[very thick,green,|-|] (0,6) -- ++(2.5,0);
\draw[very thick,blue,|-|] (2.5,6) -- ++(5,0);
\draw[very thick,green,|-|] (7.5,6) -- ++(2.5,0);
\draw[very thick,blue,|-|] (10,6) -- (15,6);
\draw[very thick,green,|-] (15,6) -- (16,6);

\node [left] at (0,4.5) {Thread 1 };
\draw[very thick,blue,-|] (0,4.5) -- (2.5,4.5);
\draw[very thick,green,|-|] (2.5,4.5) -- ++(2.5,0);
\draw[very thick,blue,|-|] (5,4.5) -- ++(5,0);
\draw[very thick,green,|-|] (10,4.5) -- ++(2.5,0);
\draw[very thick,blue,|-] (12.5,4.5) -- (16,4.5);

\node [left] at (0,3) {Thread 2 };
\draw[very thick,blue,|-|] (0,3) -- (5,3);
\draw[very thick,green,|-|] (5,3) -- ++(2.5,0);
\draw[very thick,blue,|-|] (7.5,3) -- ++(5,0);
\draw[very thick,green,|-|] (12.5,3) -- ++(2.5,0);
\draw[very thick,blue,|-] (15,3) -- (16,3);
\end{tikzpicture}
\end{center}
\caption{Critical contention\label{fig.sch-crit}}
\end{figure}
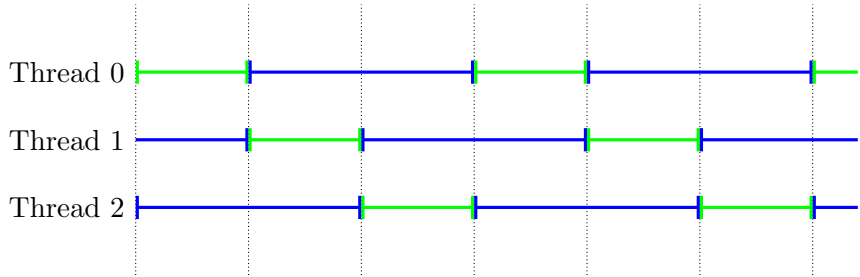

This lower bound ($\tps=(\nth-1)\times\tcslc$) is actually a good approximation for the
critical point where the queue switches its state. Altogether, we evaluate the throughput
in the following way:
\begin{itemize}
\item If the execution is intra-socket, \ie if $2\nth\leq 8$, then
\begin{itemize}
\item if $\tps\geq(\nth-1)\times\tcslcons$, use Equation~\ref{eq.thr-lc}
\item if $\tps<(\nth-1)\times\tcslcons$, use Equation~\ref{eq.thr-hc};
\end{itemize}
\item If the execution is inter-socket, \ie if $2\nth> 8$, then
\begin{itemize}
\item if $\tps\geq(\nth-1)\times\tcslcofs$, use Equation~\ref{eq.thr-lc}
\item if $\tps<(\nth-1)\times\tcslcofs$, use Equation~\ref{eq.thr-hc}.
\end{itemize}
\end{itemize}

\itemx{
\begin{itemize}
\item For queues, two main competitions: one for the head pointer, and the other one for
  the tail pointer. The most strangled one will determine the performance of the
  implementation. If dedicated, $\nth/2$ threads compete for the tail, and the remaining
  $\nth/2$ threads compete for the head.
\item Scheme: the minimum value of the time spent in \ps, for \nth threads
  competing for a single resource without conflicts, is $\tpscr = (\nth-1) \times \tcslcons$
\item We choose this value to move from one state to the other
\end{itemize}
}

\paragraph{Results}
The throughput prediction is plotted in Figure~\ref{fig.thr} (we recall that the key can
be found in Figure~\ref{fig.key}, page \pageref{fig.key}). Points are measurements, while
lines are predictions. We will follow this rule for all comparisons between prediction and
measurement. In the actual execution, the queue goes through a transient state when the
amount of work in \ps is near the critical point, but the prediction is not so far from
it.

\begin{figure}[h!]
\includegraphics[width=\textwidth]{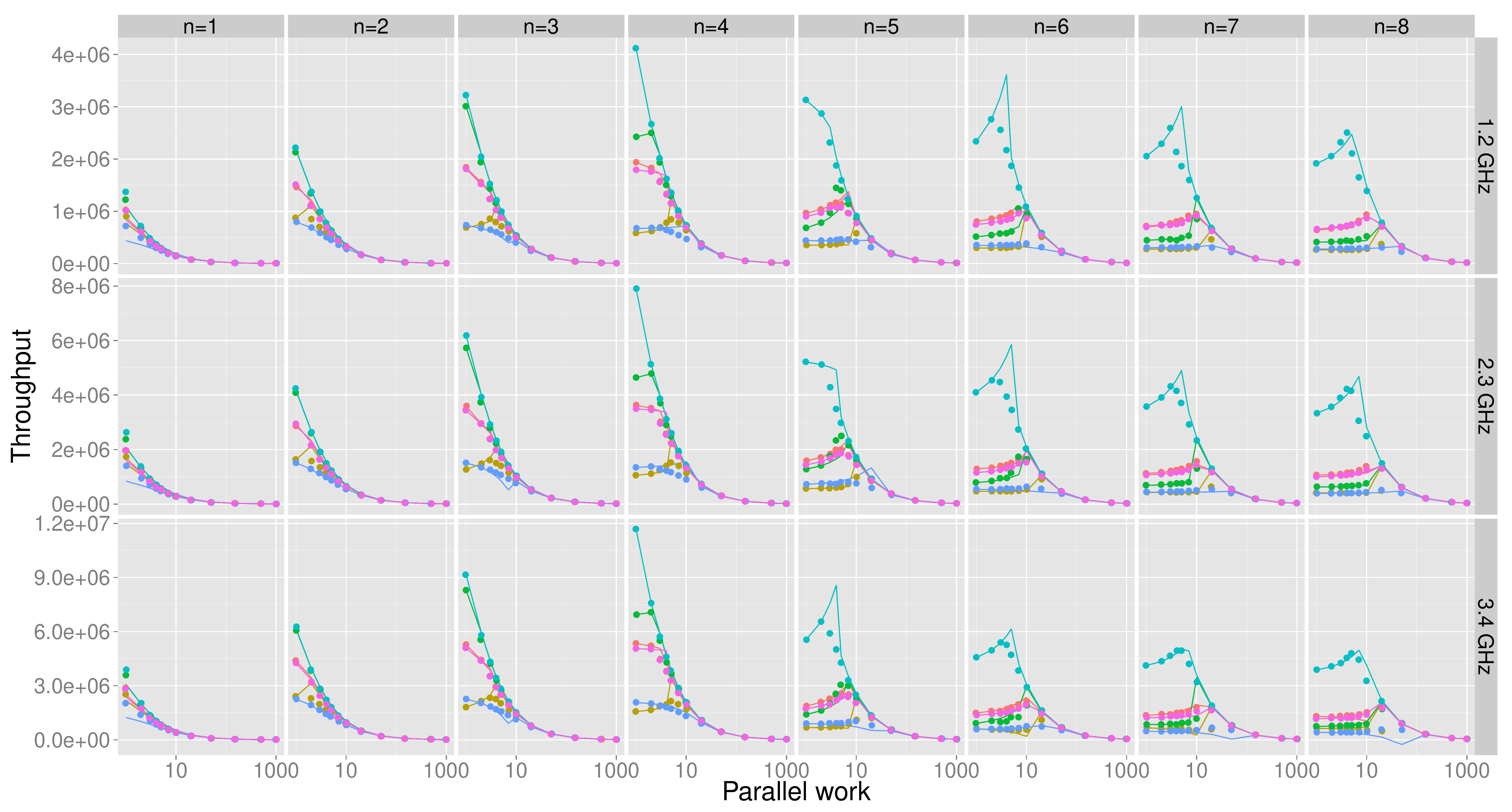}
\caption{Throughput\label{fig.thr}}
\end{figure}

\newcommand{\nm}[1]{\ema{M_{#1}}}

Let us summarize the measurements that are needed to instantiate the model, and note \nm{}
the number of measurements, decomposed into groups \nm{i}. We denote by \ntht, \nfr
and \nalg the cardinality of the sets of all possible numbers of respectively threads,
frequencies and implementations where we want to predict the throughput.
\begin{itemize}
\item \facf is found by running only one implementation, for all number of threads, hence
  $\nm{1} = \ntht$ measurements.
\item For low-contention case, for every implementation, for two different number of threads
we run the benchmark to obtain \cwlcons and \cwlcofs, thus $\nm{2} = 2\nalg$ measurements.
\item For high-contention case, for every implementation, every frequency and every number
  of threads, we launch two runs, leading to $\nm{3} = 2\nalg\nfr\ntht$ measurements.
\end{itemize}

In total, we need $\nm{} = \sum_i \nm{i} = \ntht + 2\nalg + 2\nalg\nfr\ntht$ runs.  In our
case, with three frequencies, six implementations, $8\times2$ threads, and one second
per run, it represents around $300$ seconds.

\itemx{
\begin{itemize}
\item Summary: measurement needed:
\begin{itemize}
\item Multiplicative factor between work over frequency and execution time.
  For one algorithm, one point per number of threads.
\item Low-contention: estimate of the amount of work inside \rl.
  For one frequency, for $8$ and $10$ threads, one point per algorithm.
\item High-contention: estimate of the slope of the throughput.
  Two points per frequency, algorithm, number of threads
\item Total: $\nfr + 2\nalg + 2\nfr\nalg\nth< (2\nth +1)\nalg\nfr$
\item $1$ second per measurement, \so in our case, $9\times6\times3/60 \approx 3\;\text{min}$
\end{itemize}
\item Graphs
\end{itemize}
}

\subsubsection{Power Given Throughput Measurement}

\paragraph{Power Model}


We use here the same power model that we have exposed in Section~\ref{sec:cpu-micro}. We
recall that we decompose the power into two orthogonal basis, each base having three
dimensions. On the one hand, we define the model base by separating the power into static,
active and dynamic power, such that the total power is computed by:
\[ \pow{} = \pstat{} + \pact{} + \pdyn{}. \]
On the other hand, the measurement base corresponds to the hardware that actually
dissipates the power, \ie CPU, memory and uncore. Also, power is obtained by the sum:
\[ \pow{} = \pow{C} + \pow{M} + \pow{U}. \]

In this section, we study each dimension, in each base, so that we are able to express the
power dissipation from any perspective:
\[ \pow{} = \sum_{X \in \{C,M,U\}} \left( \pstat{X} + \pact{X} + \pdyn{X}  \right). \]
Thanks to the micro-benchmarking of Section~\ref{sec:cpu-micro}, we know already all static and active
powers, therefore the whole point of this section will be to determine the dynamic power
of CPU, memory, and uncore.

On the other hand, the power model was only tailored for micro-benchmarking. However, in
this more involved case of power modeling of data structures, we take a single step further
towards a more realistic application: we can see both pairs \ps - enqueue, and dequeue - \ps as
two operations at a higher level, and we keep the steady-state property, which is
important in micro-benchmark philosophy.


\paragraph{CPU Power}
\begin{figure}[h!]
\includegraphics[width=\textwidth]{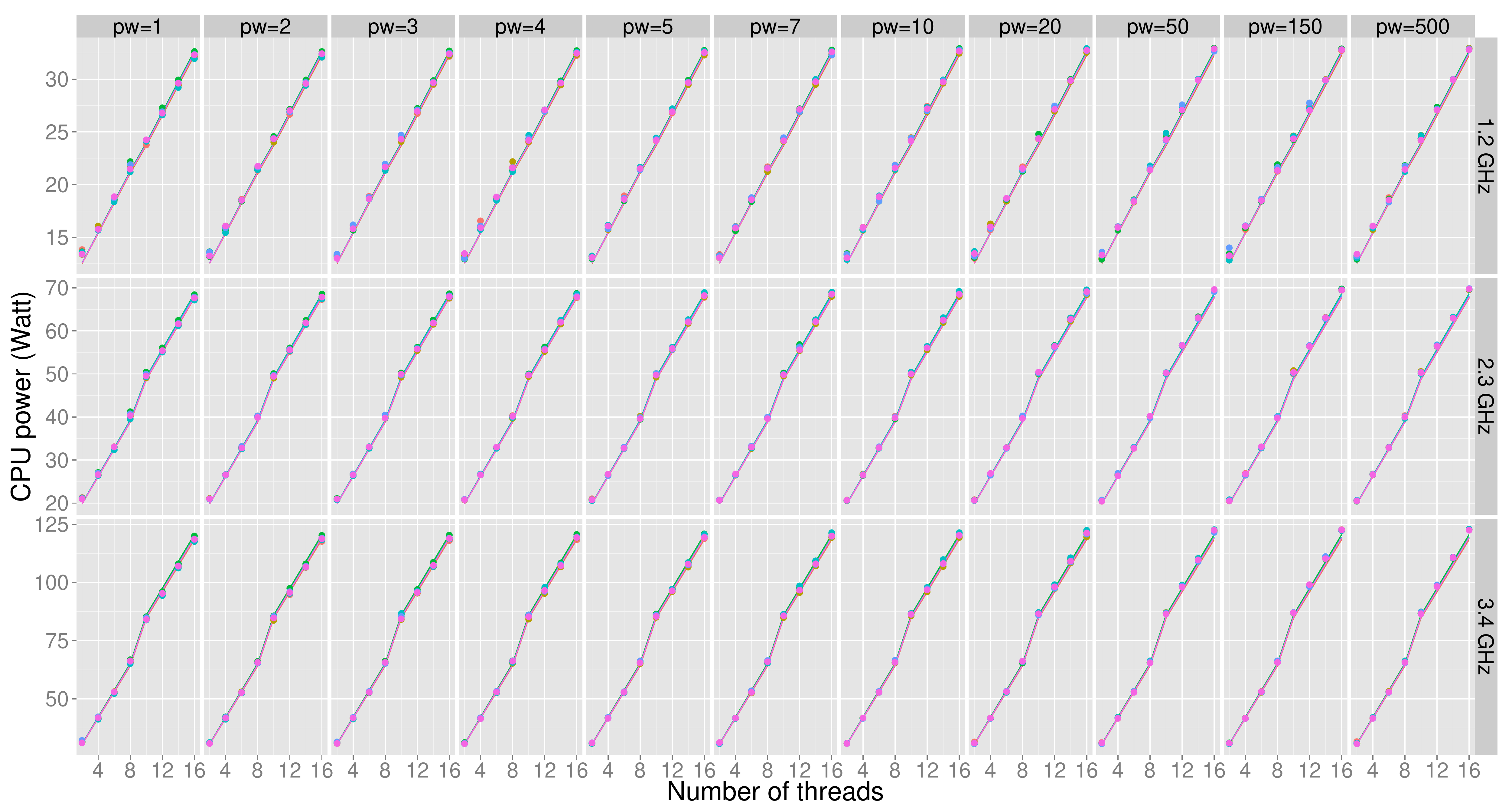}
\caption{CPU power\label{fig.cpu-pow}}
\end{figure}

\newcommand{\pceo}{\ema{\alpha_0}}  

We have seen in D1.1~\cite{D1.1} that in Chalmers' platform, most of the power is dissipated
in the CPU. Additionally, in Section~\ref{sec:cpu-micro}, we have successfully modeled the dynamic power of
CPU for several operations thanks to the generic formula:
\[ \pdyn{C} = \nth \times \left( \pca \times \freq^{\pce} + \pcb \right), \]
where \pca, \pcb and \pce are three numbers that depend only on the operation that is
executed on the CPU.

We rely on these observations to model the dynamic power of CPU for more complex
applications, especially in this deliverable for the queue implementations. We recall that
the \ps is filled with floating point divisions and our assumption such that the \rl can
be viewed approximately as a sequence of \cas has not been checked yet. On the other
side, we remark that both \cas and floating point divisions are modeled with a similar
\pce, which is around $1.7 = \pceo$. As a consequence, we consider now the queue at a higher level
and view it as a single complex operation that we can model through:
\[ \pdyn{C} = \nth \times \left( \pca \times \freq^{\pceo} + \pcb \right), \]
where \pca and \pcb have to be determined. One can notice that we have kept the linearity
according to the number of threads; this is because all threads in the queue
implementation have the same behavior, exactly in the same way as in the micro-benchmark
case.

In order to instantiate these parameters, at two frequencies \frep and \frepp, for a given
work in the \ps and a given number of threads, and for every implementation, we
measure the CPU power and extract the dynamic parts $p_0$ and $p_1$. Then we solve the system:
\[ \left\{ \begin{array}{rl}
p_0 = & \nthn \left( \pcb + \pca \times \frep^{\pceo} \right) \\
p_1 = & \nthn \left( \pcb + \pca \times \frepp^{\pceo} \right) \\
\end{array}\right., \quad\text{which leads to}\]
\[ \left\{ \arraycolsep=1.4pt\def\arraystretch{2.2}\begin{array}{rl}
\pca = & \dfrac{p_1 - p_0}{\frepp^{\pceo} - \frep^{\pceo}}\\
\pcb = & \dfrac{p_0 \times \frepp^{\pceo} - p_1 \times \frep^{\pceo}}{\frepp^{\pceo} - \frep^{\pceo}}
\end{array}\right. .\]


The prediction and measurements are plotted in Figure~\ref{fig.cpu-pow}.

\paragraph{Memory Power}
\begin{figure}[h!]
\includegraphics[width=\textwidth]{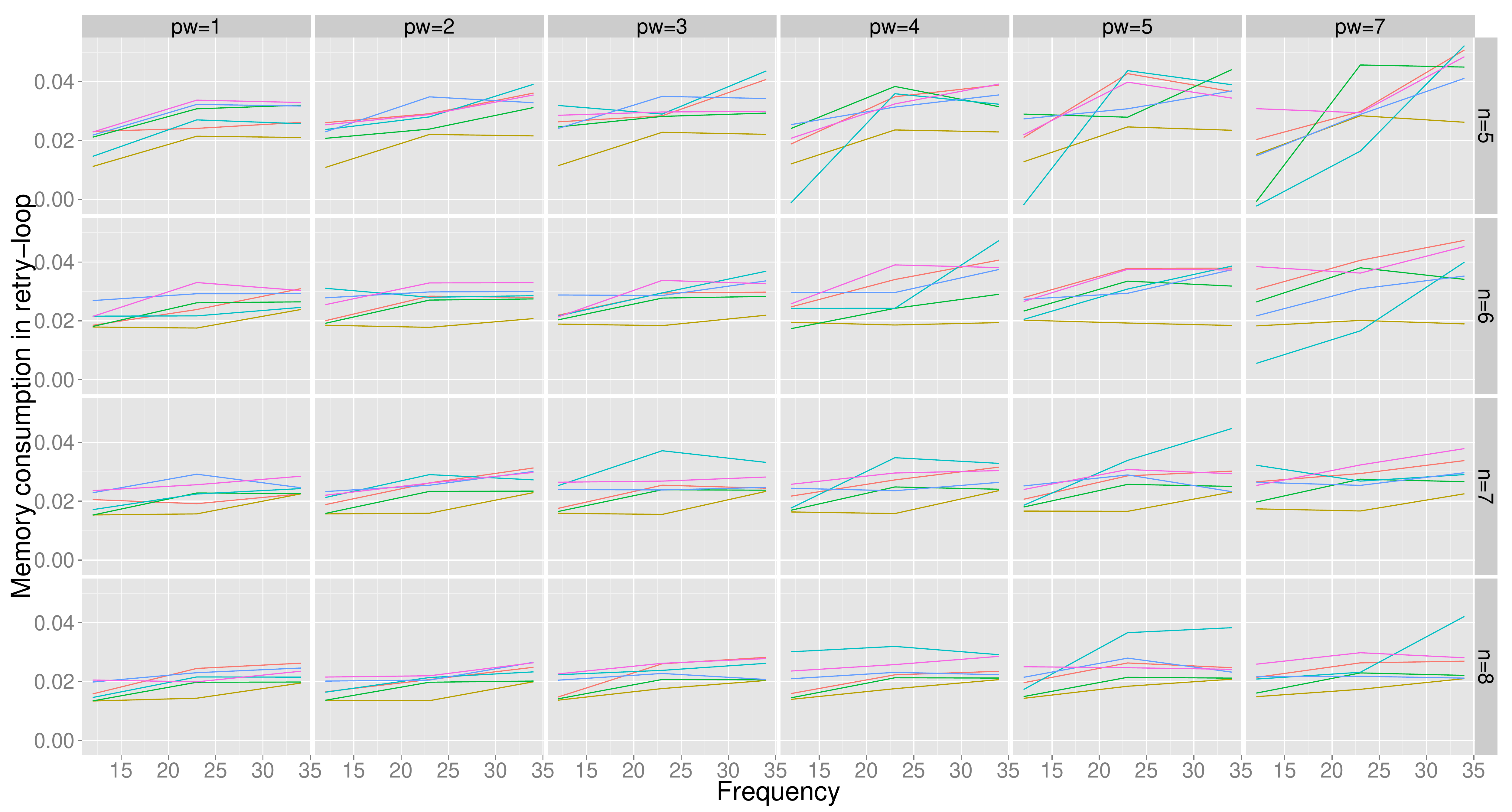}
\caption{Memory consumption\label{fig.mem-con}}
\end{figure}

\begin{figure}[h!]
\includegraphics[width=\textwidth]{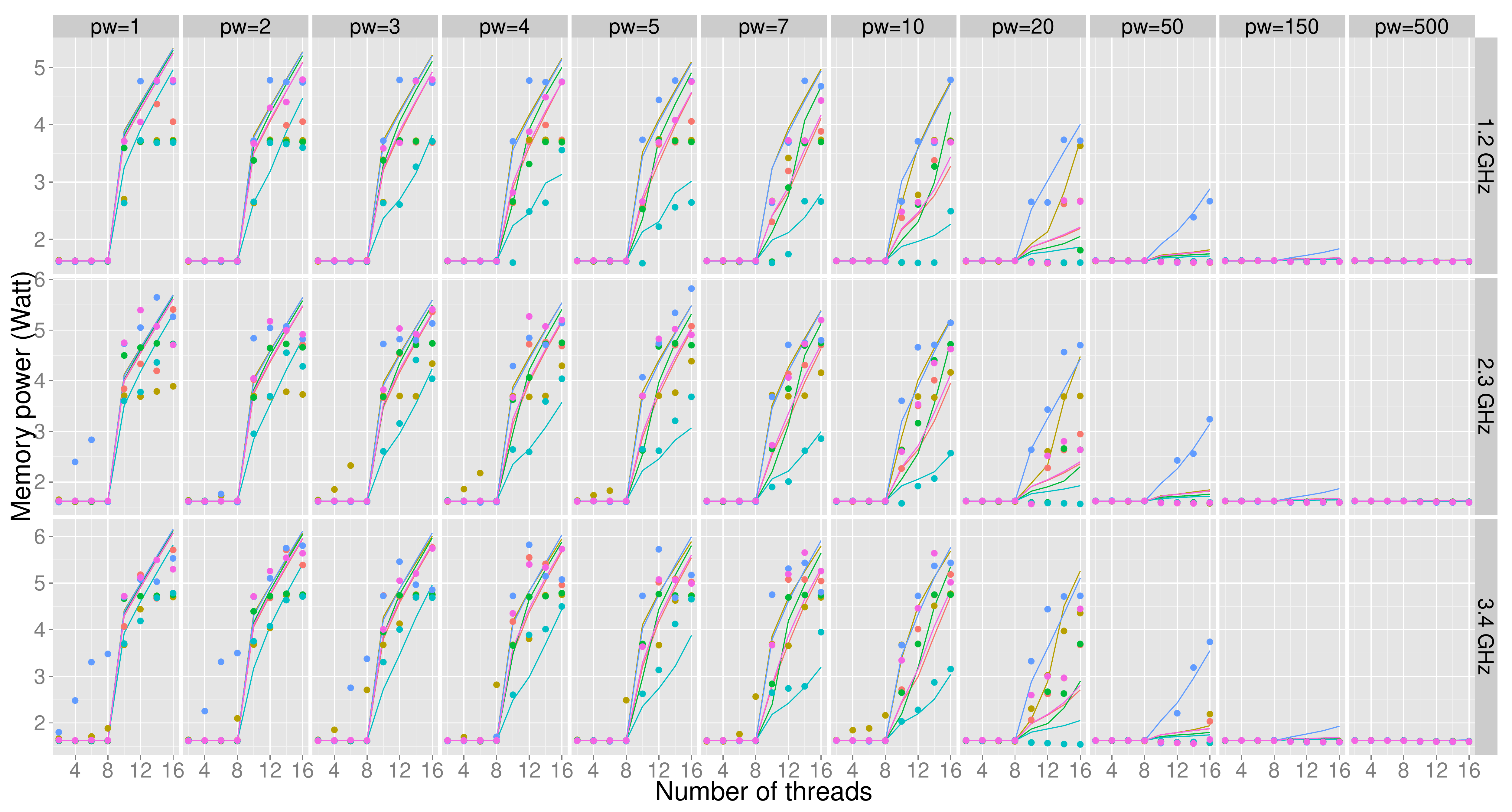}
\caption{Memory power based on measured throughput\label{fig.mem-pow}}
\end{figure}

\newcommand{\rat}{\ema{r}}

\newcommand{\coda}{\ema{\rho}}

As the \rl, which is particular to each implementation, is mainly composed of memory
operations, the main difference between the various implementations in terms of power
happens in the dynamic power of memory.

Generally speaking, we have shown show in Section~\ref{sec:cpu-micro} that the power
dissipation of the memory is due to both accesses to main memory and remote accesses to
memory. Those accesses are characterized by the amount of data \da that is accessed
remotely per second, and dynamic power dissipation is considered as proportional to this
amount. As in Section~\ref{sec.mem-pow-cpu}, we notice that the data structure does not
imply accesses to main memory, hence the power dissipation in memory is only due to remote
accesses, which only appear when the threads are spread across sockets.

As the \ps is full of pure computations, communications can only take place in the \rl. We
make one last assumption: the amount of data that are accessed per second in a \rl depends
on the implementation, but given an implementation, once a thread is in the \rl, it will
always try to access to the same amount of data per second. When the queue is highly
congested, if a thread fails then it will retry and will access the data in the same way
as the previous try; and if there is expansion, then the thread will still try to access
the data for the whole time it is in the \rl.  As a consequence, the amount of data that
are accessed remotely, hence the dynamic power of memory, is strongly related with the
ratio \rat of the time spent in the \rl over the time spent in the \ps. The dynamic power of
memory can be computed by:
\[ \pdyn{M} = \coda \times \nth \times \rat, \]
where \coda is a number depending only on the implementation, and represents the memory
access intensity in the \rl. We have again \nth as a multiplicative factor since all the
threads have the same behavior.

Now we rely on the structure of the benchmark and on the throughput to compute this
ratio. We do not know how many retries are necessary for a thread to successfully exit the
\rl; however, we know that the thread performs exactly one \ps per successful
operation. As the throughput \thr is the number of successful operations per second, the
ratio is found by:
\begin{equation}
1 - \rat = \frac{\thr \times \pw}{\nth \times \facf \times \freq}.
\label{eq.rat}
\end{equation}

We still have indeed that $\thr = \nth / (\tps + \tcs)$. This leads to:
\[ 1- \rat  = \frac{\tps}{\tps + \tcs} = \frac{\tps \times \thr}{\nth},\]
and Equation~\ref{eq.rat} is derived from the expression of the time
spent in the \ps of Equation~\ref{eq.tps}.

In Figure~\ref{fig.mem-con}, we plot:
\[ \frac{\pow{M} - \pstat{M}}{\rat \times \nth},\]
where \pow{M} is the measured power dissipated by the memory and \rat is computed through
Equation~\ref{eq.rat}. Firstly we remark that the ideas in the model are not contradicted
by the graph: everything is almost constant, and the power dissipated by the memory seems
to be ruled indeed by the considered ratio. A priori, this ratio should depend on the
implementation, but we observe that there is no clear trend, and implementations are very
close to each other. This means that all implementations behave in a similar way
concerning the amount of data accessed remotely per second.

That is why we only need, if the throughput is known, to run the benchmark for a given
implementation, a given size of \ps, with a given \nth, at a given frequency \freq, in
order to find the unique \coda, common to all implementations.

The comparison between the measured and the estimated power is plotted in
Figure~\ref{fig.mem-pow}. Two noticeable observations should be added: first, as in the
micro-benchmark experiments, we remark some steps in the measured power, but we prefer to
keep a continuous estimate. Second, we see that implementations {\bf a1} and {\bf a5}
sometimes consume memory power even for intra-socket execution. This could be due to the
fact that these versions implement reference counting in their memory management, which
could lead to the use of main memory due to overly long chains of unreclaimed nodes.


\paragraph{Uncore Power}

We predict the uncore power in the same way as the memory power, except that we have an
additive component which is linear with the number of threads. This linear component is
due to the RDTSCP utilization in the \ps, and the remaining part may be related to the
ring utilization when the threads access the shared data, both inter- and intra-socket.

Briefly, we take the uncore measurement, from which we subtract the static uncore power
and the linear component, then operate in the same way and find a new constant $\coda'$.

Results are pictured in Figure~\ref{fig.unc-pow}, where we notice that, even if the
behavior is similar, the amplitude of uncore power variations is relatively smaller than
the memory power, and almost negligible in front of CPU power.


\begin{figure}[h!]
\includegraphics[width=\textwidth]{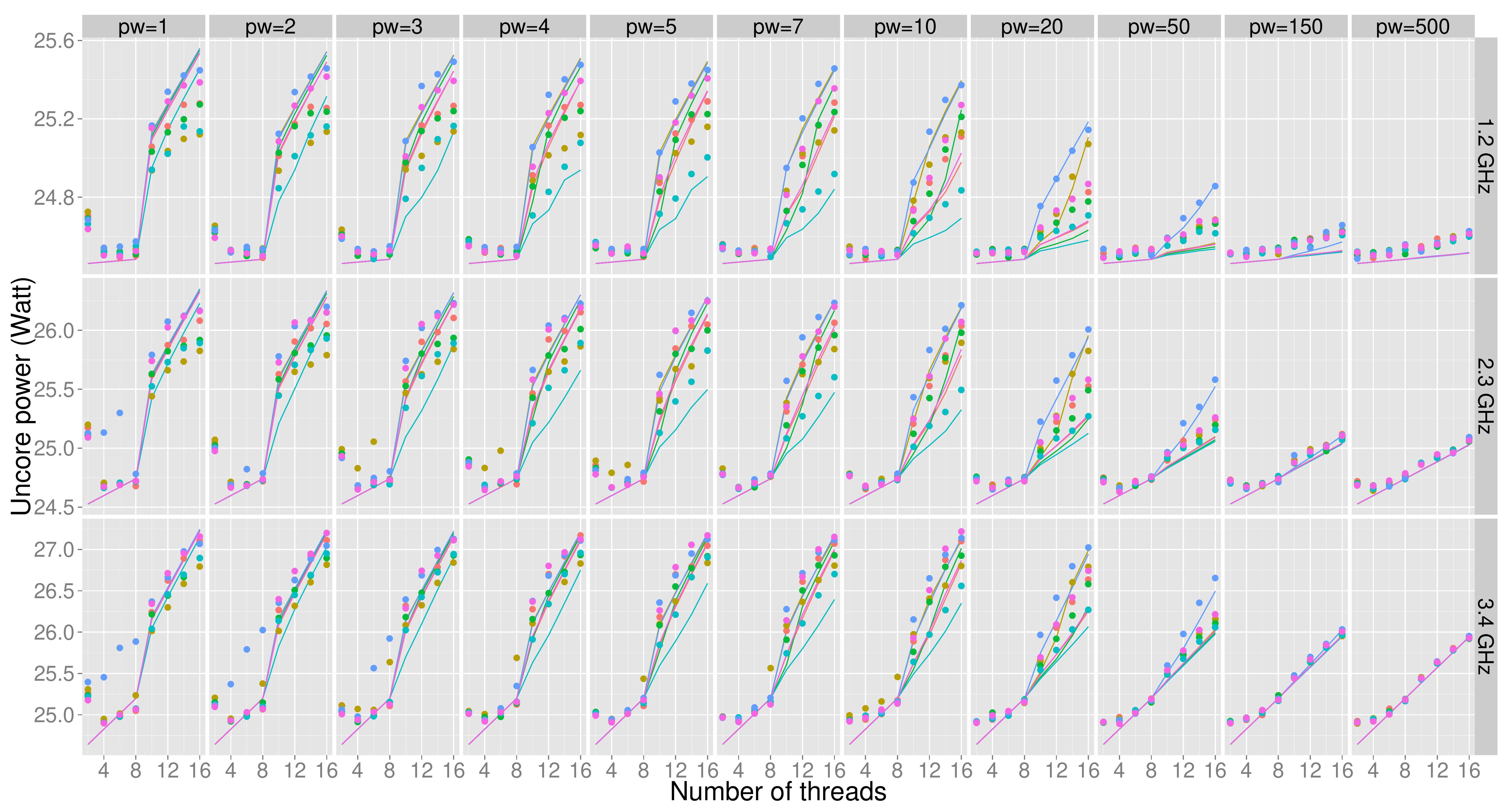}
\caption{Uncore power\label{fig.unc-pow}}
\end{figure}

\paragraph{Total Power}
\begin{figure}[h!]
\includegraphics[width=\textwidth]{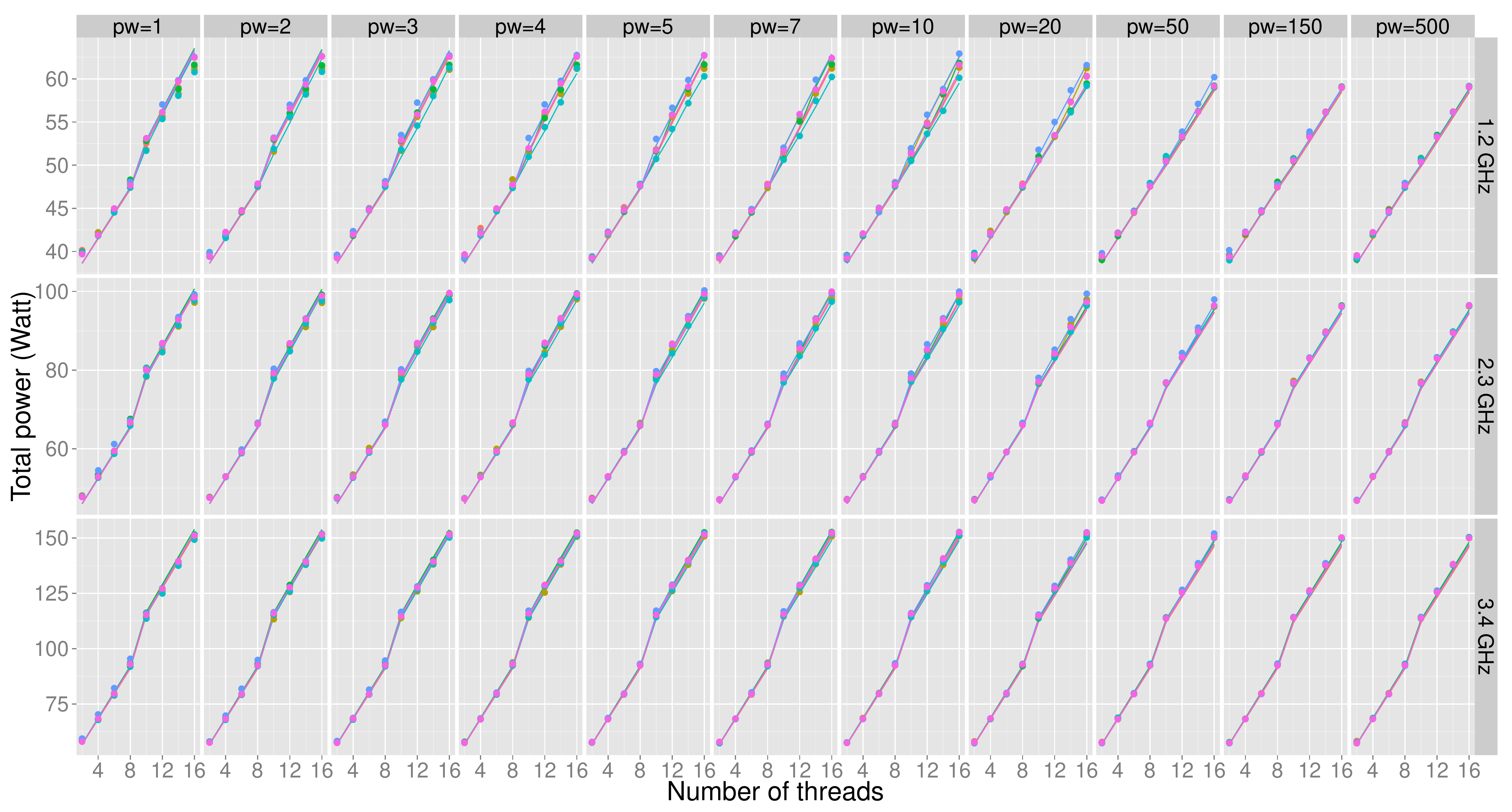}
\caption{Total power\label{fig.tot-pow}}
\end{figure}

In terms of number of measurements, estimating the power dissipation is not costly: we can
choose a \ps and \nth not less than 4, then run each implementation at two different
frequencies. This enables the prediction of CPU power, and we can use one of those
measurements to predict the memory and uncore powers. Altogether, we only need $\nm{} =
2\nalg$ measurements. We plot the comparison of total power in Figure~\ref{fig.tot-pow} to
appreciate the quality of the estimation.

\subsubsection{Complete Prediction}
\begin{figure}[h!]
\includegraphics[width=\textwidth]{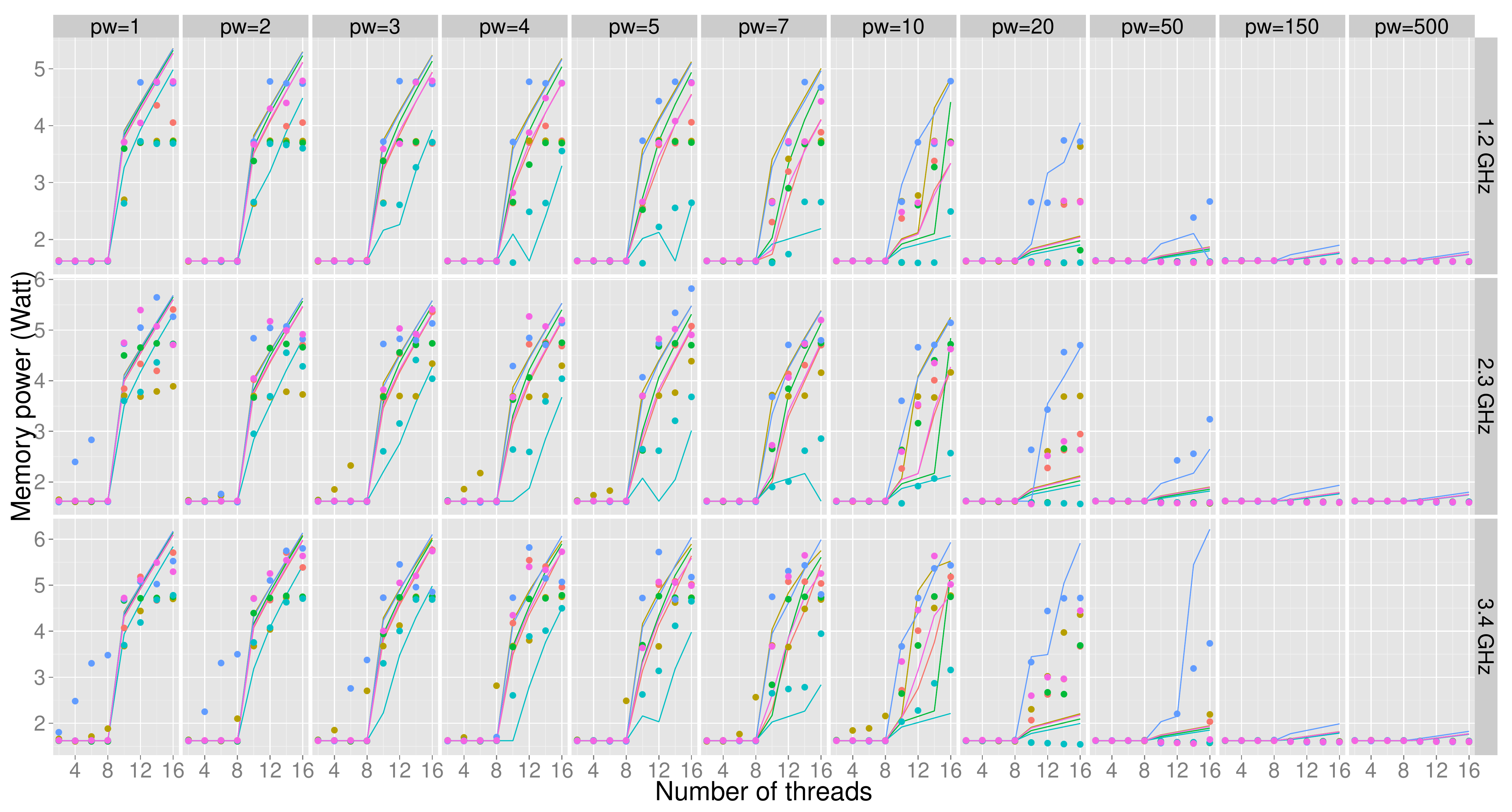}
\caption{Memory power based on estimated throughput\label{fig.mem-pow-thr}}
\end{figure}
\begin{figure}[h!]
\includegraphics[width=\textwidth]{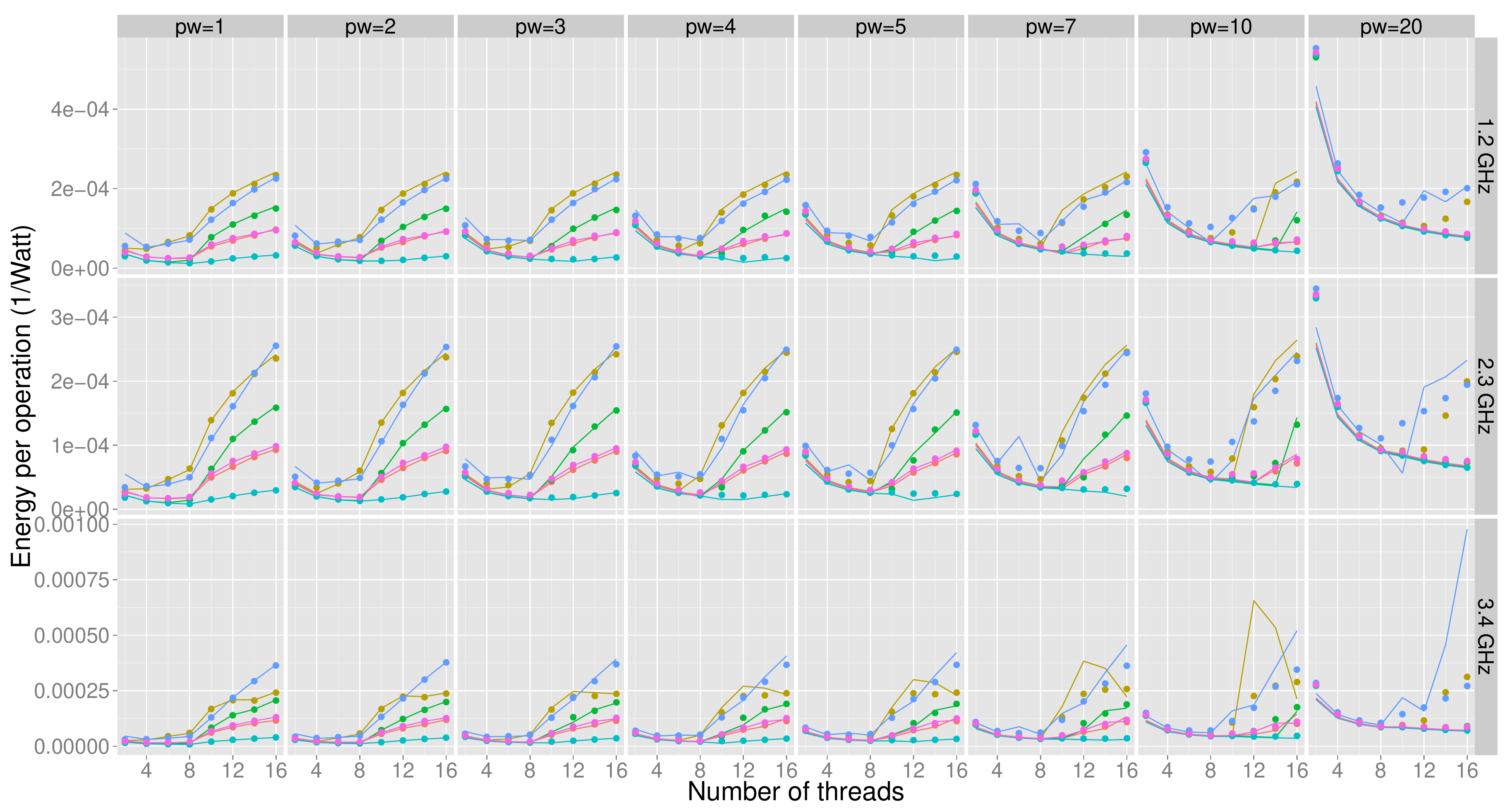}
\caption{Energy per operation\label{fig.nrg-pop1}}
\end{figure}
\begin{figure}[h!]
\includegraphics[width=\textwidth]{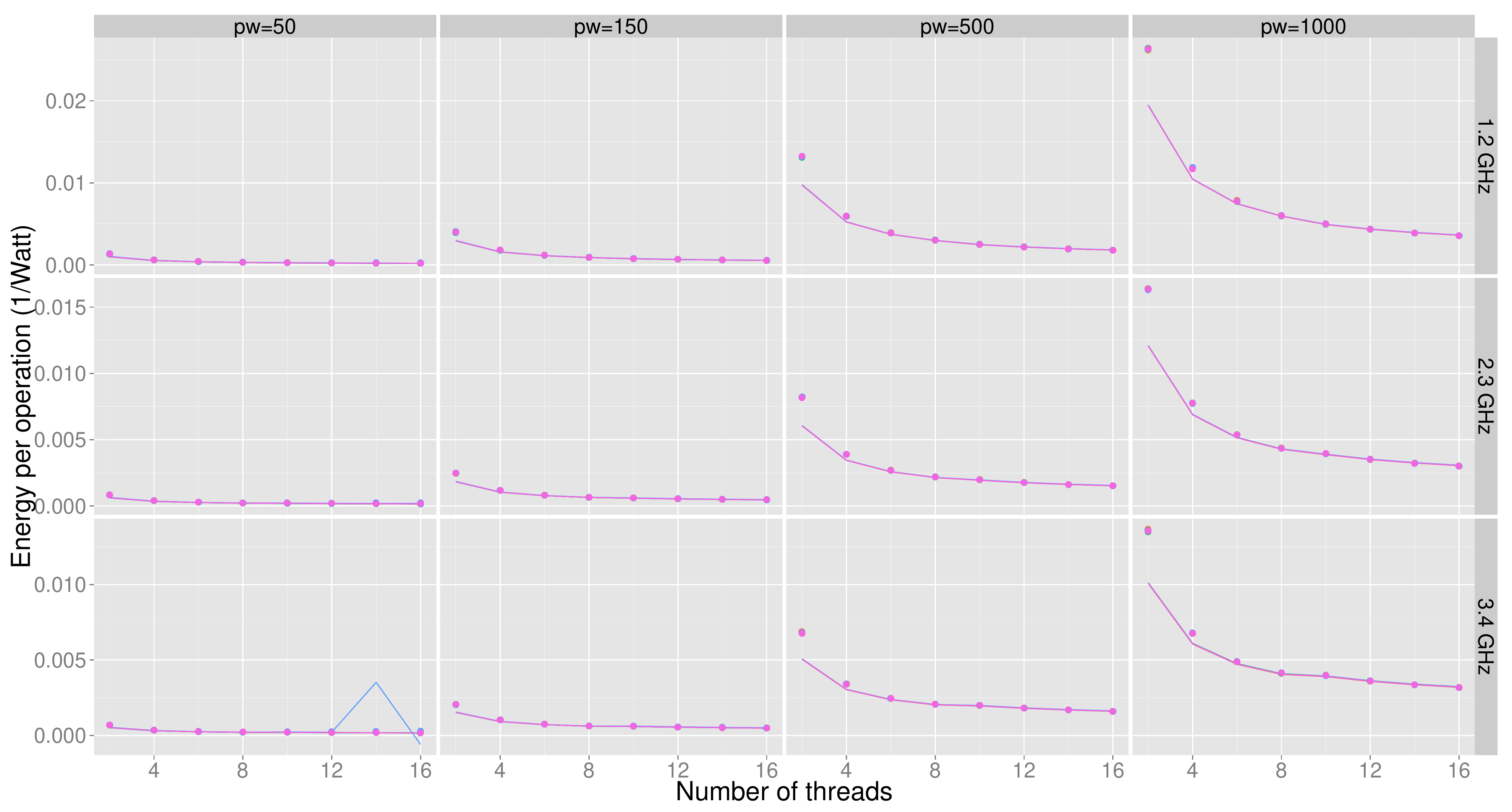}
\caption{Energy per operation\label{fig.nrg-pop2}}
\end{figure}

We plot in Figure~\ref{fig.mem-pow-thr} the estimate of the memory power dissipation when
we use the predicted throughput while computing the ratio in Equation~\ref{eq.rat},
instead of the measured throughput. It shows that the throughput prediction is good
enough, since there is no clear difference in the memory power, and we do not need to rely
on the throughput measurement. This is an important property since otherwise, we should
have run the benchmark for each value of the variables and measure the throughput, in
order to be able to compute the ratio, and then the memory power.

In Figures~\ref{fig.nrg-pop1} and~\ref{fig.nrg-pop2} is represented the energy per
operation. Except from the pathological case with $\pw=20$, the prediction is
accurate. This mistake in estimation occurs because the implementations {\bf a1} and {\bf
a5} have a long transient behavior between high and low contention cases, and the
throughput is harder to estimate in this range.


\subsubsection{Description of the Implementations}
\label{sec:shm-queue-algs}\label{sec:NOBLE}

\leaveout{ 
\remind{(PT) [Re-write a little]
A common approach to parallelizing applications is to divide the problem into separate threads that act as either producers or consumers. The problem of synchronizing these threads and streaming of data items between them, can be alleviated by utilizing a shared collection data structure.}

Concurrent FIFO queues and other producer/consumer
collections are fundamental data structures that are key components in
applications, algorithms, run-time and operating systems. 
\index{concurrent data structures!queue}
\index{queue}
The Queue abstract data type is a collection of items in which only the earliest added item may be accessed. Basic operations are \op{Enqueue} (add to the tail) and \op{Dequeue} (remove from the head). Dequeue returns the item removed. The data structure is also known as a ``first-in, first-out'' or FIFO buffer.
} 



\remind{(PT) [Re-write a little]}

\paragraph{NOBLE~\cite{Sundell02,Sundell08}}
Most of the implementations that we use are part of the NOBLE library.
The NOBLE library offers support for non-blocking multi-process synchronization in shared memory systems. NOBLE has been designed in order to: i) provide a collection of shared data objects in a form which allows them to be used by non-experts, ii) offer an orthogonal support for synchronization where the developer can change synchronization implementations with minimal changes, iii) be easy to port to different multi-processor systems, iv) be adaptable for different programming languages and v) contain efficient known implementations of its shared data objects. The library provides a collection of the most commonly used data types.
The semantics of the components, which have been designed
to be the very same for all implementations of a
particular abstract data type, are based on the sequential
semantics of common abstract data types and adopted for
concurrent use. The set of operations has been limited
to those which can be practically implemented using both
non-blocking and lock-based techniques. Due to the concurrent
nature, also new operations have been added, e.g.
Update which cannot be replaced by Delete followed
by Insert. Some operations also have stronger semantics
than the corresponding sequential ones, e.g. traversal in a List is not invalidated due to concurrent deletes, compared
to the iterator invalidation in STL. As the published algorithms
for concurrent data structures often diverge from the
chosen semantics, a large part of the implementation work
in NOBLE, besides from adoption to the framework, also
consists of considerable changes and extensions to meet
the expected semantics.

The various lock-free concurrent queue algorithms that we include in
this study are briefly described below.

\paragraph{Tsigas-Zhang~\cite{TsiZ01b}}
 Tsigas and Zhang \cite{TsiZ01b} presented a lock-free extension of \cite{Lam83} for any number of threads where synchronization is done both on the array elements and the shared head and tail indices using \op{CAS}, and the ABA problem is avoided by exploiting two (or more) null values. In \cite{TsiZ01b} synchronization is done both directly on the array elements and the shared head and tail indices using \op{CAS}\footnote{The Compare-And-Swap (CAS) atomic primitive will update a given memory word, if and only if the word still matches a given value (e.g. the one previously read). CAS is generally available in contemporary systems with shared memory, supported mostly directly by hardware and in other cases in combination with system software.}, thus supporting multiple producers and consumers. In order to avoid the ABA problem when updating the array elements, the algorithm exploits using two (or more) null values; the ABA problem is due to the inability of \op{CAS} to detect concurrent changes of a memory word from a value (A) to something else (B) and then again back to the first value (A). A {\em CAS} operation can not detect if a variable was read to be A and then later changed to B and then back to A by some concurrent processes. The {\em CAS} primitive will perform the update even though this might not be intended by the algorithm's designer. Moreover, for lowering the memory contention the algorithm alternates every other operation between scanning and updating the shared head and tail indices.

\paragraph{Valois~\cite{Val94}}
Valois \cite{Val94,Val95phd} makes use of linked list in his lock-free implementation which is based on the \op{CAS} primitive. He was the first to present a lock-free implementation of a linked-list. The list uses auxiliary memory cells between adjacent pairs of ordinary memory cells. The auxiliary memory cells were introduced to provide an extra level of indirection so that normal memory cells can be removed by joining the auxiliary ones that are adjacent to them. His design also provides explicit cursors to access memory cells in the list directly and insert or delete nodes on the places the the cursors point to.

\paragraph{Michael-Scott~\cite{Michael96}}
 Michael and Scott \cite{Michael96} presented a lock-free queue that is more efficient, synchronizing via the shared head and tail pointers as well as via the next pointer of the last node.
Synchronization is done via shared pointers indicating the current head and tail node as well via the next pointer of the last node, all updated using \op{CAS}.
The tail pointer is then moved
to point to the new element, with the use of a  \op{CAS} operation. This second step
can be performed by the thread invoking the operation, or by another thread that
needs to help the original thread to finish before it can continue. This helping
behavior is an important part of what makes the queue lock-free, as a thread
never has to wait for another thread to finish.
The queue is fully dynamic as more nodes are allocated as needed when new items are added. The original presentation used unbounded version counters, and therefore required double-width \op{CAS} which is not supported on all contemporary platforms. The problem with the version counters can easily be avoided by using some memory management scheme as e.g. \cite{Mic04b}.

\paragraph{Moir-et-al.~\cite{MoirNSS:2005:elim-queue}}
 Moir et al. \cite{MoirNSS:2005:elim-queue} presented an extension of the Michael and Scott \cite{Michael96}  lock-free queue algorithm where elimination is used as a back-off strategy to increase scalability when contention on the queue's head or tail is noticed via failed \op{CAS} attempts. However, elimination is only possible when the queue is close to empty during the operation's invocation. 

\paragraph{Hoffman-Shalev-Shavit~\cite{DBLP:conf/opodis/HoffmanSS07}}
Hoffman et al.~\cite{DBLP:conf/opodis/HoffmanSS07} takes another approach in their design in order to increase scalability by allowing concurrent \op{Enqueue} operations to insert the new node at adjacent positions in the linked list if contention is noticed during the attempted insert at the very end of the linked list. To enable these "baskets" of concurrently inserted nodes, removed nodes are logically deleted before the actual removal from the linked list, and as the algorithm traverses through the linked list it requires stronger memory management than \cite{Mic04b}, such as~\cite{DBLP:journals/tpds/GidenstamPST09} or \cite{HerLMM:2005:NMM} and a strategy to avoid long chains of logically deleted nodes.

\paragraph{Gidenstam-Sundell-Tsigas~\cite{Gidenstam10:OPODIS}}
Gidenstam et al. \cite{Gidenstam10:OPODIS} combines the efficiency of using arrays and the dynamic capacity of using linked lists, by providing a lock-free queue based on linked lists of arrays, all updated using \op{CAS} in a cache-aware manner. In resemblance to \cite{Lam83}\cite{GiaMoVa:2008:ff-queue}\cite{TsiZ01b} this algorithm uses arrays to store (pointers to) the items, and in resemblance to \cite{TsiZ01b} it uses \op{CAS} and two null values. Moreover, shared indices \cite{GiaMoVa:2008:ff-queue} are avoided and scanning \cite{TsiZ01b} is preferred as much as possible. In contrast to \cite{Lam83}\cite{GiaMoVa:2008:ff-queue}\cite{TsiZ01b} the array is not static or cyclic, but instead more arrays are dynamically allocated as needed when new items are added, making the queue fully dynamic. 


\remind{(AG) cite:\\
stack: \cite{Treiber86,Mic04b,Val95phd,HenSY10}\\
bag: \cite{Sundell11}\\
pool: \cite{AfeKNS10}\\
\\
}


\subsubsection{Towards Realistic Applications: Mandelbrot Set Computation} \label{sec:mandelbrot}
\remind{Is this the right place? Or further up to motivate why energy aspects ofdata structures are interesting.}


\paragraph{Mandelbrot Set Description}
\begin{figure}[p]
\includegraphics[width=\textwidth]{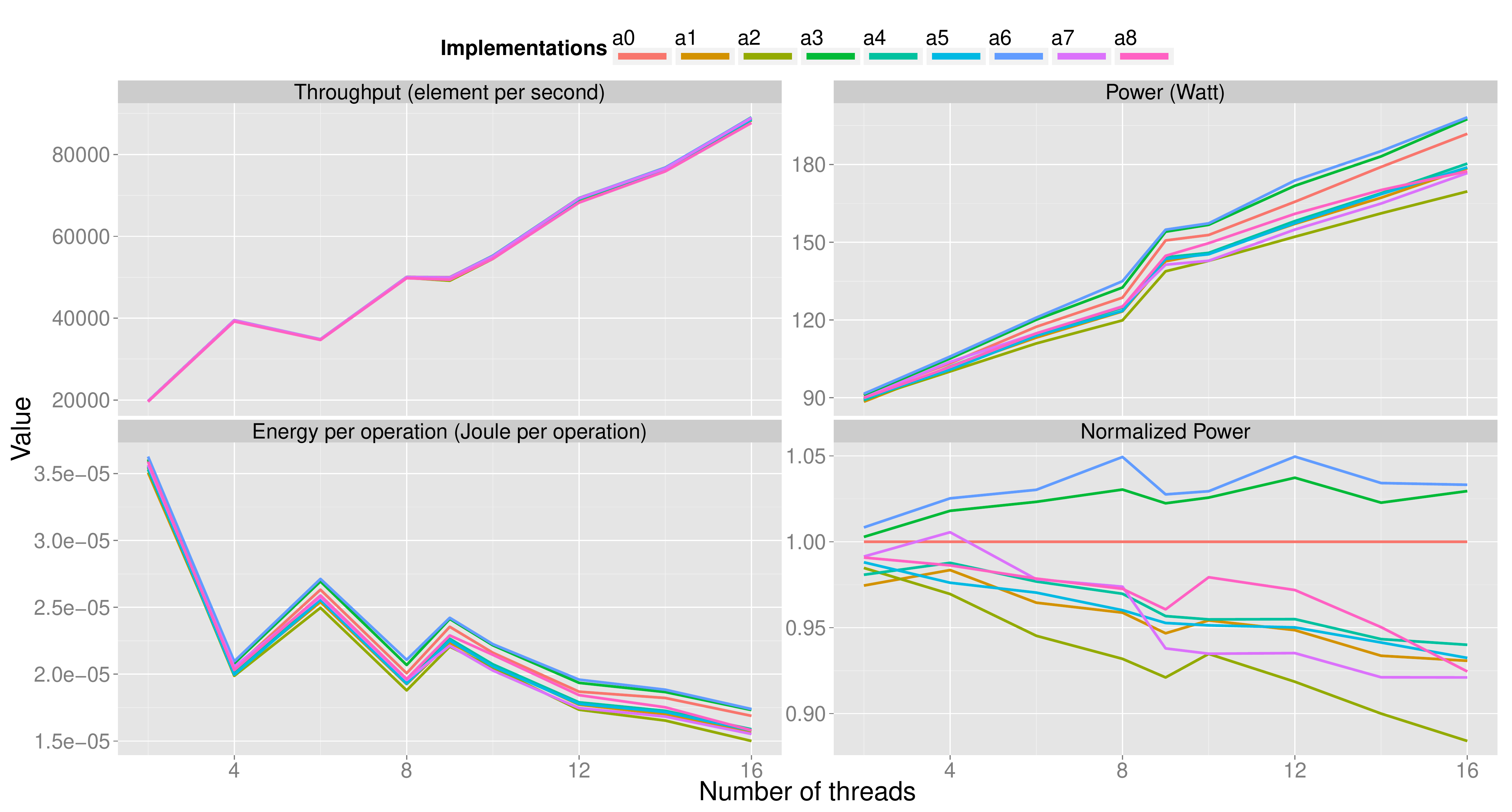}
\caption{Mandelbrot using $16\times 16$ pixel regions.}
\label{fig.mandel-p3}
\end{figure}
\begin{figure}[p]
\includegraphics[width=\textwidth]{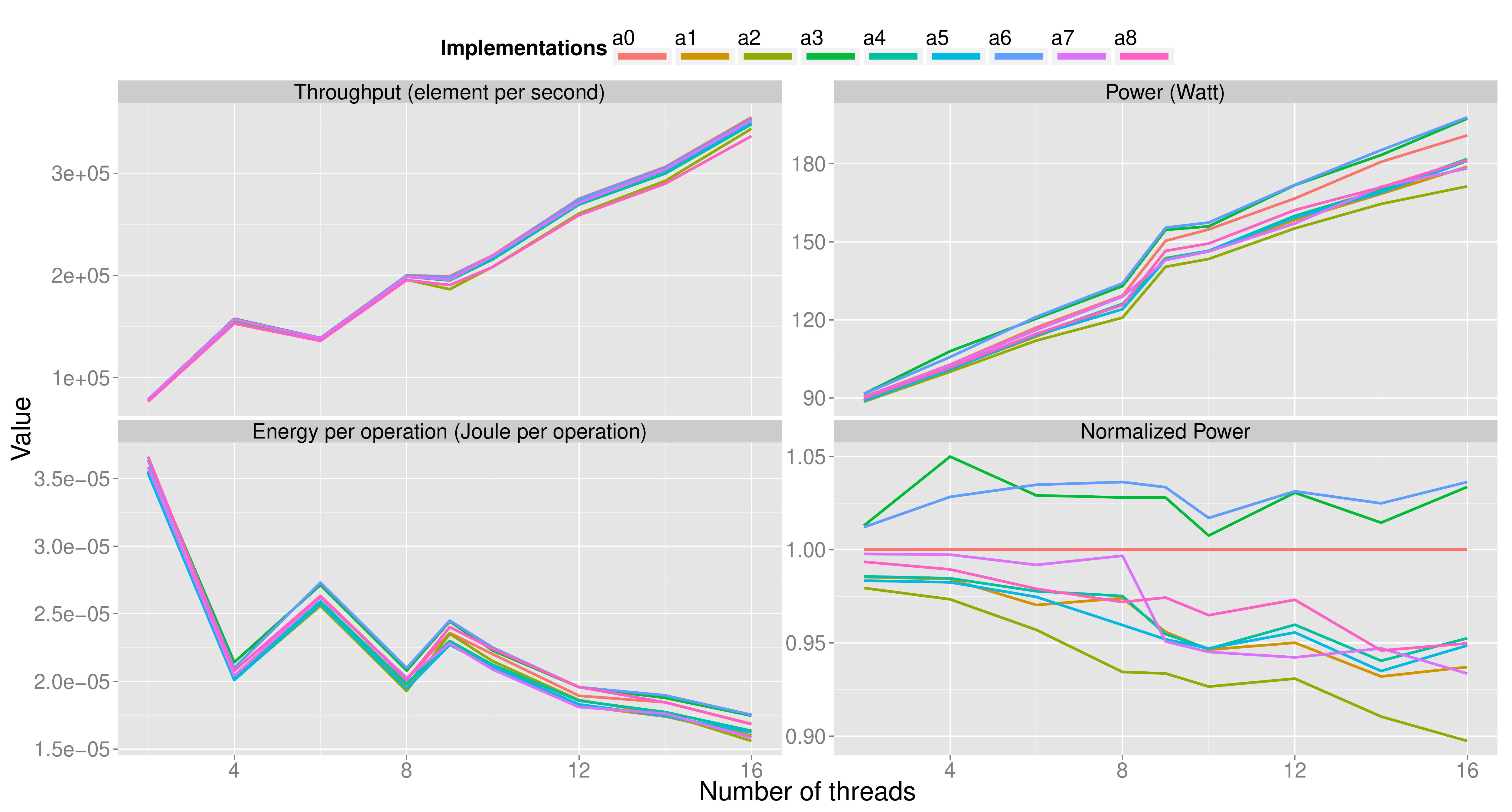}
\caption{Mandelbrot using $8\times 8$ pixel regions.}
\label{fig.mandel-p2}
\end{figure}
\begin{figure}[p]
\includegraphics[width=\textwidth]{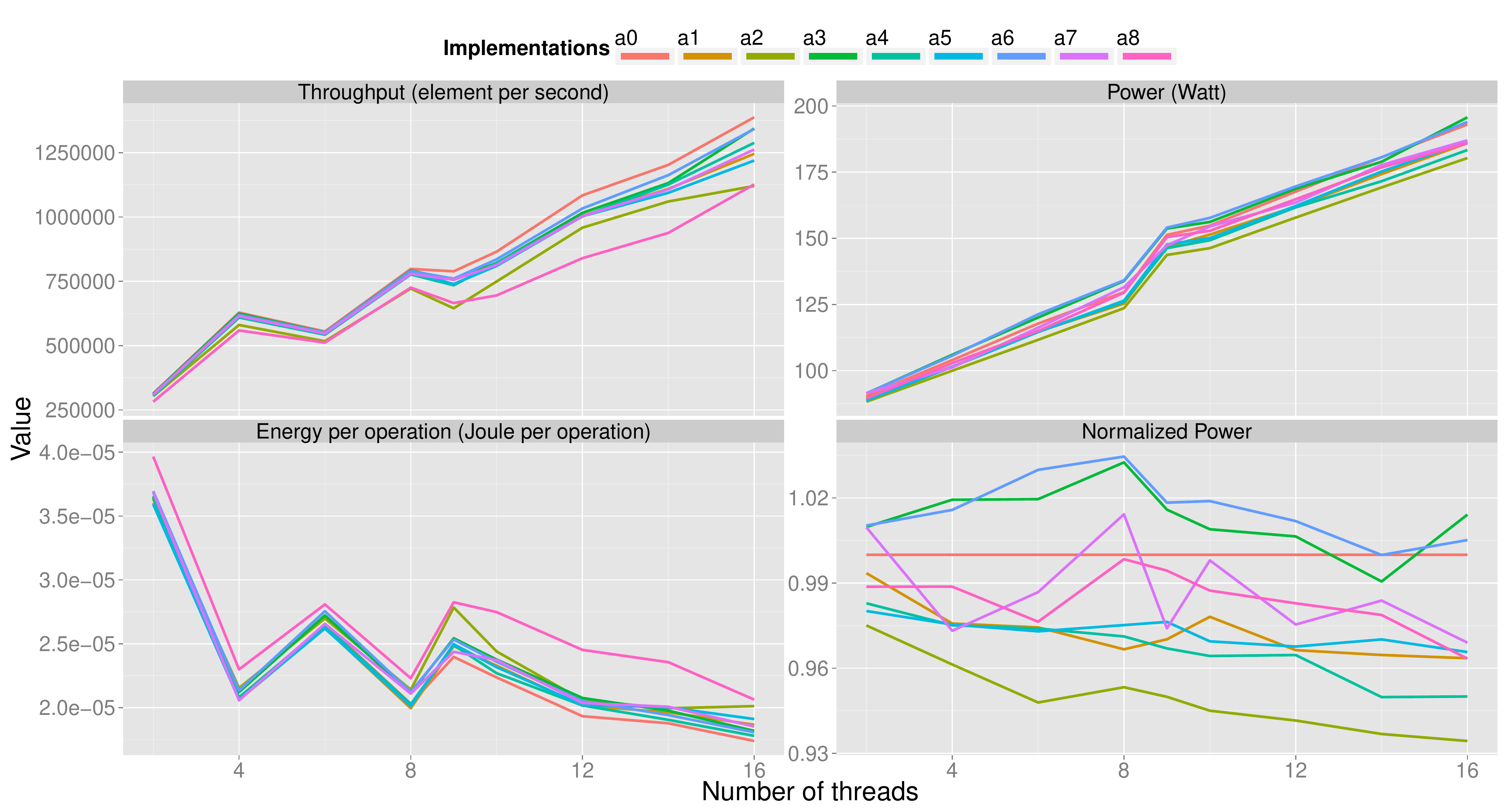}
\caption{Mandelbrot using $4\times 4$ pixel regions.}
\label{fig.mandel-p1}
\end{figure}
\begin{figure}[p]
\includegraphics[width=\textwidth]{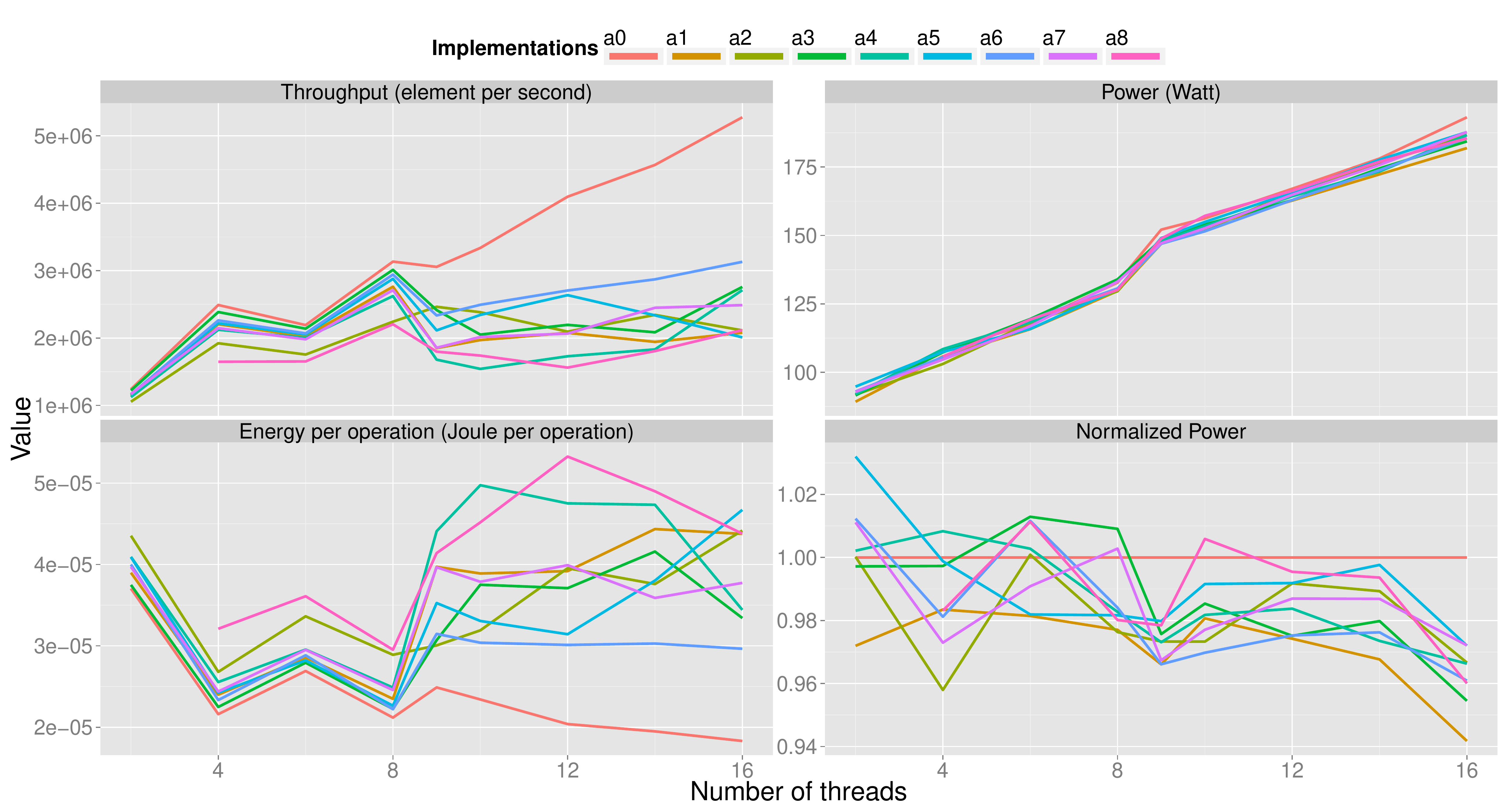}
\caption{Mandelbrot using $2\times2$ pixel regions.}
\label{fig.mandel-p0}
\end{figure}


As a simple case-study of parallel applications that use concurrent
data structures for communication we have used an application that
computes and renders an image of the Mandelbrot set~\cite{Man80} in
parallel using the producer/consumer pattern. The program, also used
as part of the evaluation in~\cite{Sundell11}, uses a shared
collection data structure for communication between the
program's two major phases:

\begin{itemize}
\item Phase 1 consists of computing the number (with a maximum of 255)
  of iterations for a given set of points within a chosen region of
  the image. The results for each region together with its coordinates
  are then put in the collection data structure.
\item Phase 2 consists of, for each computed region stored in the
  collection, computing the RGB values for each contained point and
  draw these pixels to the resulting image. The colors for the
  corresponding number of iterations are chosen according to a rainbow
  scheme, where low numbers are rendered within the red and high
  numbers are rendered within the violet spectrum.
\end{itemize}

Phase 1 is performed in parallel with phase 2, i.e., like a
pipeline. Half of the threads perform phase 1 and the rest perform
phase 2. The application is implemented in C, and renders a 32-bit
color image of 8192 times 8192 pixels of the Mandelbrot set. The size
of each square region is chosen to be one of $16 \times 16$ (i.e., 16 by
16), $8\times 8$, $4\times 4$, or $2\times 2$ pixels which also determines the number of work
units and, hence, the level of contention on the shared
collection. The whole image is divided into a number (equal to half
the number of threads) of larger row-oriented parts%
\footnote{Due to the nature of the Mandelbrot set, this way of
  deciding each part might not be fair in respect of workload per
  thread. As can be seen in the experimental results, this partition
  pattern causes 3 parts to take longer time than 2 parts in
  parallel, because the total execution time depends on the slowest
  part.}%
, where each producer thread (i.e., phase 1) work sequentially on the
regions contained within its own part. The consumer threads (i.e.,
phase 2) render the regions got from the collection in the order that
they were obtained, until the producer threads have finished and the
collection is empty.
The application uses a dense pinning strategy, pinning the producers
and then the consumers to consecutive cores, e.g. when 16 threads are
used the producers are pinned to cores on the first socket while the
consumers are pinned to cores on the second socket.
This is just one of many possible ways to divide the work and pin threads,
it remains as future work to explore other ways.

In this application the shared data structure used for communication
only need to offer one \op{Add} operation that adds an element to the
collection and one \op{TryRemove} operation that removes and returns
one element from the collection. The minimal semantic requirements are
that at most one \op{TryRemove} returns each \op{Add}ed element and
that a surplus of \op{TryRemove} operations eventually (e.g. after all
\op{Add}s have been issued) returns all \op{Add}ed elements.

Any linearizable concurrent queue, stack or bag data structure meets
these requirements and could be used as the shared collection. There are
even some non-linearizable data structures that could meet them.

The following concurrent shared collection data structures, most of
which are described in Section~\ref{sec:shm-queue-algs}, have been
considered:
\begin{itemize}
\item {\bf a0}. Lock-free and linearizable bag by Sundell et al.~\cite{Sundell11}.
\item {\bf a1}. Lock-free and linearizable queue by Michael and
  Scott~\cite{Michael96}.
\item {\bf a2}. Lock-free and linearizable queue by Valois~\cite{Val94}.
\item {\bf a3}. Lock-free and linearizable queue by Tsigas and
  Zhang~\cite{TsiZ01b}.
\item {\bf a4}. Lock-free and linearizable queue by Hoffman et
  al.~\cite{DBLP:conf/opodis/HoffmanSS07}.
\item {\bf a5}. Lock-free and linearizable queue by Moir et
  al.~\cite{MoirNSS:2005:elim-queue}.
\item {\bf a6}. Lock-free and linearizable stack by Michael~\cite{Mic04b}.
\item {\bf a7}. Lock-free and linearizable stack by Hendler et
  al.~\cite{HenSY10}.
\item {\bf a8}. Lock-free EDTree (a.k.a. pool or bag) by Afek~et
  al.~\cite{AfeKNS10}.
\leaveout{
\item {\bf a0}. Lock-free and linearizable queue by Michael and
  Scott~\cite{Michael96}.
\item {\bf a1}. Lock-free and linearizable queue by Valois~\cite{Val94}.
\item {\bf a2}. Lock-free and linearizable queue by Tsigas and
  Zhang~\cite{TsiZ01b}.
\item {\bf a3}. Lock-free and linearizable queue by Gidenstam et
  al.~\cite{Gidenstam10:OPODIS}.
\item {\bf a4}. Lock-free and linearizable queue by Hoffman et
  al.~\cite{DBLP:conf/opodis/HoffmanSS07}.
\item {\bf a5}. Lock-free and linearizable queue by Moir et
  al.~\cite{MoirNSS:2005:elim-queue}.
\item {\bf a6}. Lock-based (and linearizable) queue.
\item {\bf a7}. Lock-free and linearizable stack by Michael~\cite{Mic04b}.
\item {\bf a8}. Lock-free and linearizable stack by Hendler et
  al.~\cite{HenSY10}.
\item {\bf a9}. Lock-free and linearizable bag by Sundell et al.~\cite{Sundell11}.
\item {\bf a10}. Lock-free EDTree (a.k.a. pool or bag) by Afek~et
  al.~\cite{AfeKNS10}.
}
\end{itemize}

Each implementation has been run at each of the 4 work unit sizes
($2\times 2$, $4\times 4$, $8\times 8$ and $16\times 16$ pixels)
and with 2, 6, 8, 9, 10, 12, 14 and
16 threads on the EXCESS server at Chalmers. The results are
presented in Figures~\ref{fig.mandel-p3} to~\ref{fig.mandel-p0} in order
of decreasing work unit size, i.e. increasing contention.
For each case the following metrics are shown (clockwise starting from
the top left):
i)   throughput in pixels per second;
ii)  total system power in Watts;
iii) total system power normalized by {\bf a1} power; and
iv)  total energy in Joules consumed per pixel.

As mentioned above the method used to divide the Mandelbrot set into
regions does not share the work equally among the producer threads
which results in the decreases in throughput for 6 and 9 threads.

When the work units are large, such as in Figure~\ref{fig.mandel-p3},
the difference in throughput between the different collection
implementations is very small indeed for any number of threads. The
work load is dominated by independent parallel computation and
consequently the level of contention on the collection is low. There
is, however, a somewhat larger difference in energy per pixel. This
difference is interesting as it ought to be directly related to
properties of the collection implementation as all implementations
carry out the same total amount of parallel work and a very similar number of
successful collection operations per second. Moreover, the lowest energy per
pixel costs are achieved by the implementation, {\bf a2}, which is
among the worst at high contention (compare with
Figure~\ref{fig.mandel-p0}).
In this particular application the producers do a larger part of the
total work than the consumers which can lead to the shared collection
becoming empty at times. However, the cost is not distributed equally
across all work units -- some are cheaper for the producers than
others. Consumers finding the shared collection empty will retry the
\op{TryRemove} operation in a tight loop. This could could be one
reason for the difference in power as the effort needed to determine
that the collection is empty varies among the different
algorithms. E.g. for {\bf a1} and {\bf a2} this just requires reading
a small number of pointers (2 to 3), which however invokes memory barriers,
while for {\bf a0} it entails
scanning through (while invoking few memory barriers) at least one block of
pointers per thread using the
data structure.

When the work units are small, such as in Figure~\ref{fig.mandel-p0},
there are large differences in throughput from 4 threads and up. This
together with the fact that the total system power for the different
implementations (at the same number of threads) is even more close
together than when using larger work units the differences in energy
per pixel varies considerably. Here the contention level on the
collection is higher, above 8 threads where the throughput of the less
scalable implementations flatten or decrease it can be considered
high. In this case all but one of the implementations have their
energy per pixel sweetspot at less than or equal to 8 threads
(i.e. when using cores in only one socket). Implementation {\bf a0}
(the bag) is the only one that delivers the lowest energy per pixel
when using all cores of the machine. It is worth noting that the bag
data type has a potential to use less synchronization than a queue or
stack data type that must enforce an (illusion of) total order among
all their elements.

From this case-study some observations can be made about the
problem of making an informed selection of implementation for a
multi-variant shared data structure in a certain application and context:
\begin{itemize}
\item the semantic requirements of the application must be known (naturally)
  but should also not be overstated as that would limit the choices of
  implementation;
\item the required throughput of data structure operations (and their
  mix) needs to be predicted (bounded) from the parallel work-load to
  estimate the level of contention (which if too high would further
  bound the achievable throughput of data structure operations);
  and that, consequently,
\item a good prediction of achievable data structure operation
  throughput for each implementation and for a certain state will be
  needed to do that.
\end{itemize}


\paragraph{Simplified implementation}

As mentioned above, realistic applications introduce variety of additional parameters that
hardens the estimation of throughput and power. The Mandelbrot application has two main
differences from synthetic tests. In the first place, the parallel sections are composed
of a mixture of computations and memory accesses.  It is hard to estimate the memory
access delays and intensity which are important for our model. These metrics are used to
determine the parallel section size and the bandwidth requirements which are used to
obtain the memory power consumption.

Another complexity regarding the Mandelbrot application is the unequal load balance among producer
threads. Even though the problem domain is decomposed into equally sized chunks, some
radians require less work than others because they diverge rapidly and require less
iterations before determining that they do not belong to the Mandelbrot set. This fact creates
variability in the parallel section size which does not occur in the synthetic
tests. There are some ways to eliminate this load balance problem.  One very simple way is
to force each thread to iterate until maximum count even after determining that the point
does not belong to the Mandelbrot set. However, this is not an ideal approach since it leads
to waste of resources. Instead, one can decompose the domain in an interleaved manner to
obtain a better load balance.

For now, we apply the simple approach and leave the interleaved decomposition as future
work. Having obtained load balance with simple modifications, we make use of synthetic
applications to predict power and throughput values for the Mandelbrot application. We
determine the parallel section size and extrapolate the Mandelbrot power and throughput
metrics from the corresponding synthetic application.



\paragraph{Mandelbrot Prediction}

There are slight differences between the simplified Mandelbrot implementation that we
consider in this paragraph and the synthetic benchmarks that we have analyzed in the
previous subsection. Those differences have an impact on power dissipation through two
main components:
\begin{itemize}
\item CPU power. What is considered as \ps in Mandelbrot differs from the synthetic
      test since the operations that reside in this \ps are of a different nature: only
      floating point divisions for synthetic benchmark, and a complete mixture of
      arithmetic and memory operations for Mandelbrot. The dynamic CPU power that we have
      measured and extrapolated in synthetic test is then no longer valid for the new
      application.
\item Memory power. Again, as some memory operations take place inside the \ps, the
      amount of remote accesses per unit of time in the whole program changes; and we have
      seen that this metric impacts directly the memory power dissipation.
\end{itemize}

Because of those variations between the synthetic test and the new application, we need to
measure the power dissipation of memory and CPU for some more values of the
parameters. This requirement of new power measurements comes however naturally in the
process; we cannot expect to be able to predict the power dissipation of any application
that uses a queue without having any knowledge about the characteristics of the
application according to power. We are then able to extract from those power measurements
both power dissipation of the \rl (which is correlated to the queue implementation) and
power dissipation of the \ps (which depends on the application that actually uses one of
the queue implementations).

\bigskip

Concerning the CPU power dissipation, we do not reconsider the assumption that it mostly
depends on the CPU power dissipated in the \ps, \ie there is no clear difference of CPU
power in the different queue implementations that we have studied in this
deliverable. However, contrary to the floating point division case in synthetic benchmark,
we do not know what is the relation between frequency and CPU power. We then still rely on
the following equation:
\[ 
\pow{C} = \pstat{C} + \sock \times \pact{C} + \nth \times \pdyn{C}(\freq,\prog), \quad
\text{which is equivalent to}
\]
\[ \pdyn{C}(\freq,\prog) = \frac{1}{\nth} \times 
   \left(\pow{C} - \pstat{C} - \sock \times \pact{C} \right)  \]

Hence, we choose a pattern and a queue implementation, and for each value of the
frequency, we run the application and obtain the corresponding value of dynamic CPU power.

\medskip
\begin{figure}
\includegraphics[width=\textwidth]{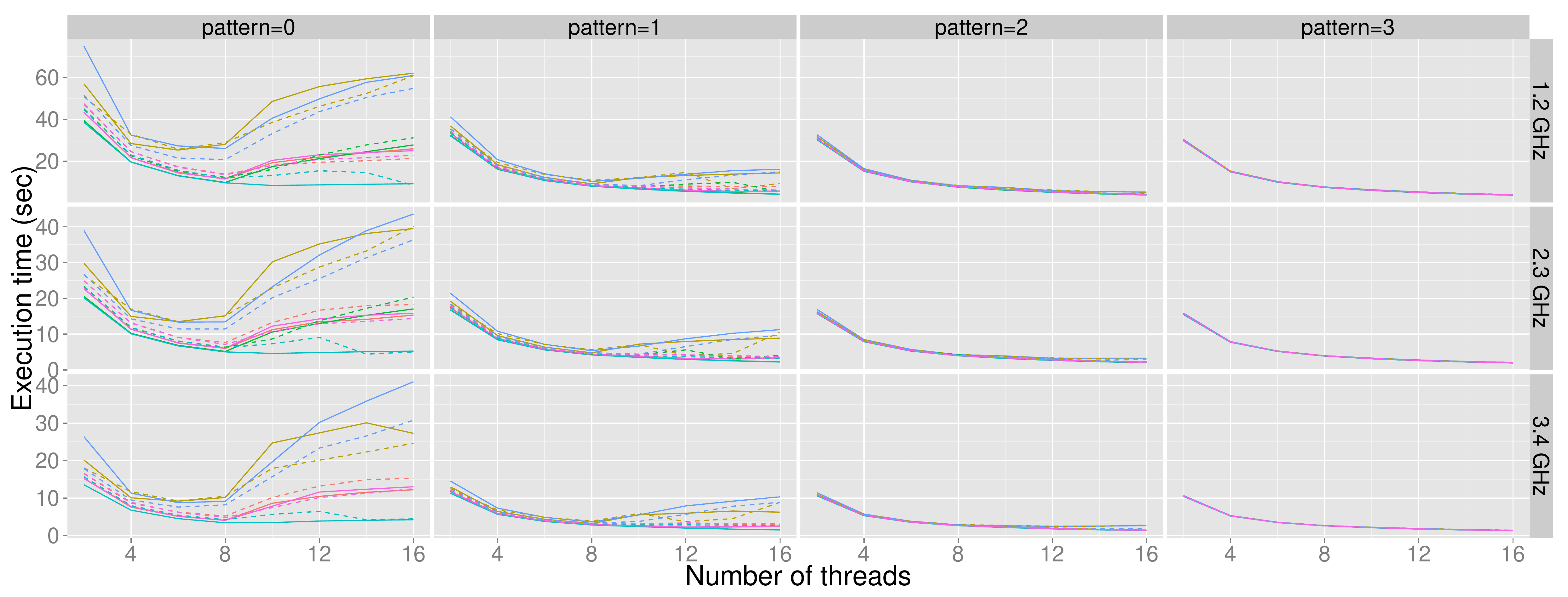}

\includegraphics[width=.49\textwidth]{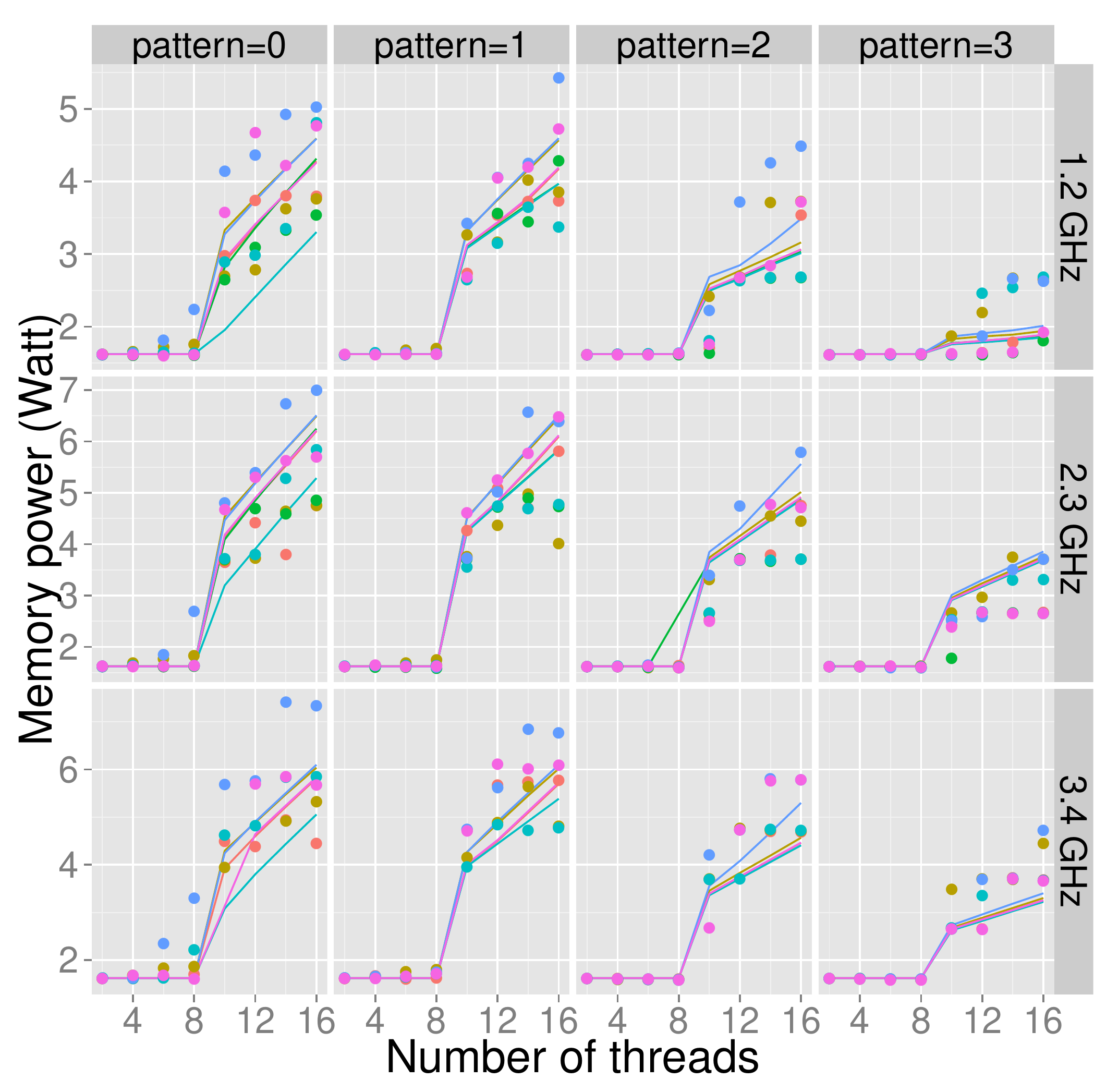}\hfill%
\includegraphics[width=.49\textwidth]{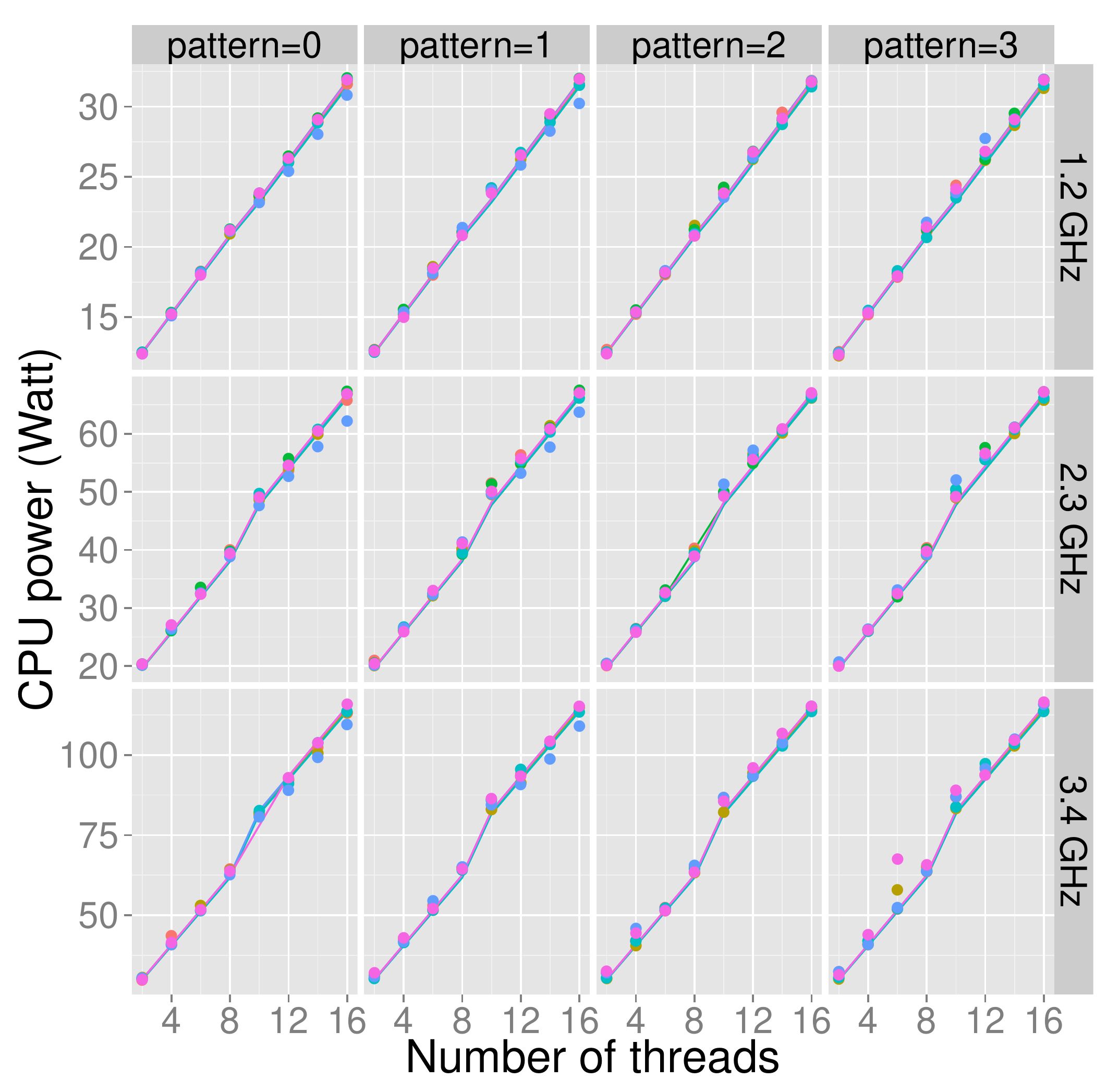}

\includegraphics[width=.49\textwidth]{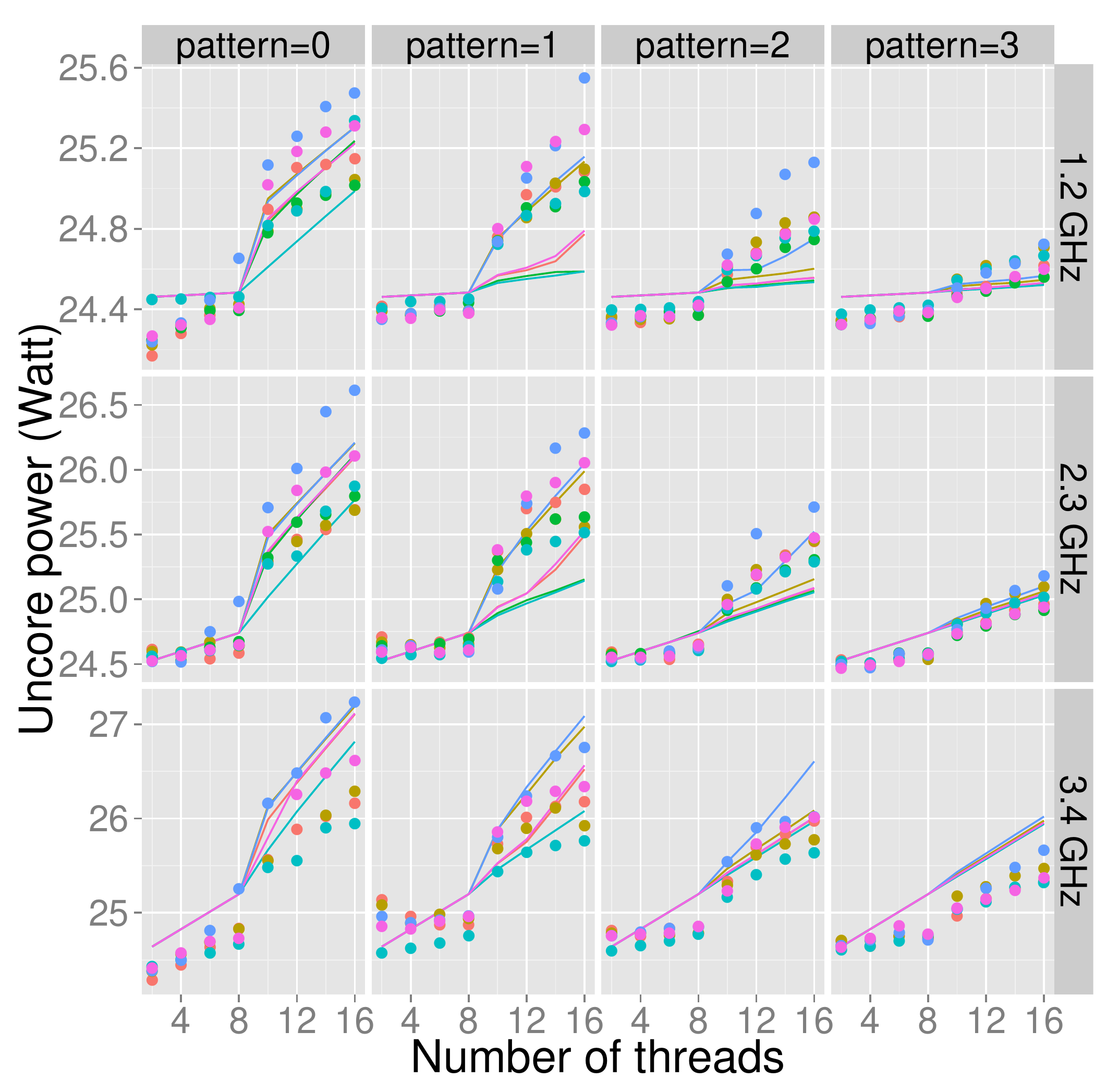}\hfill%
\includegraphics[width=.49\textwidth]{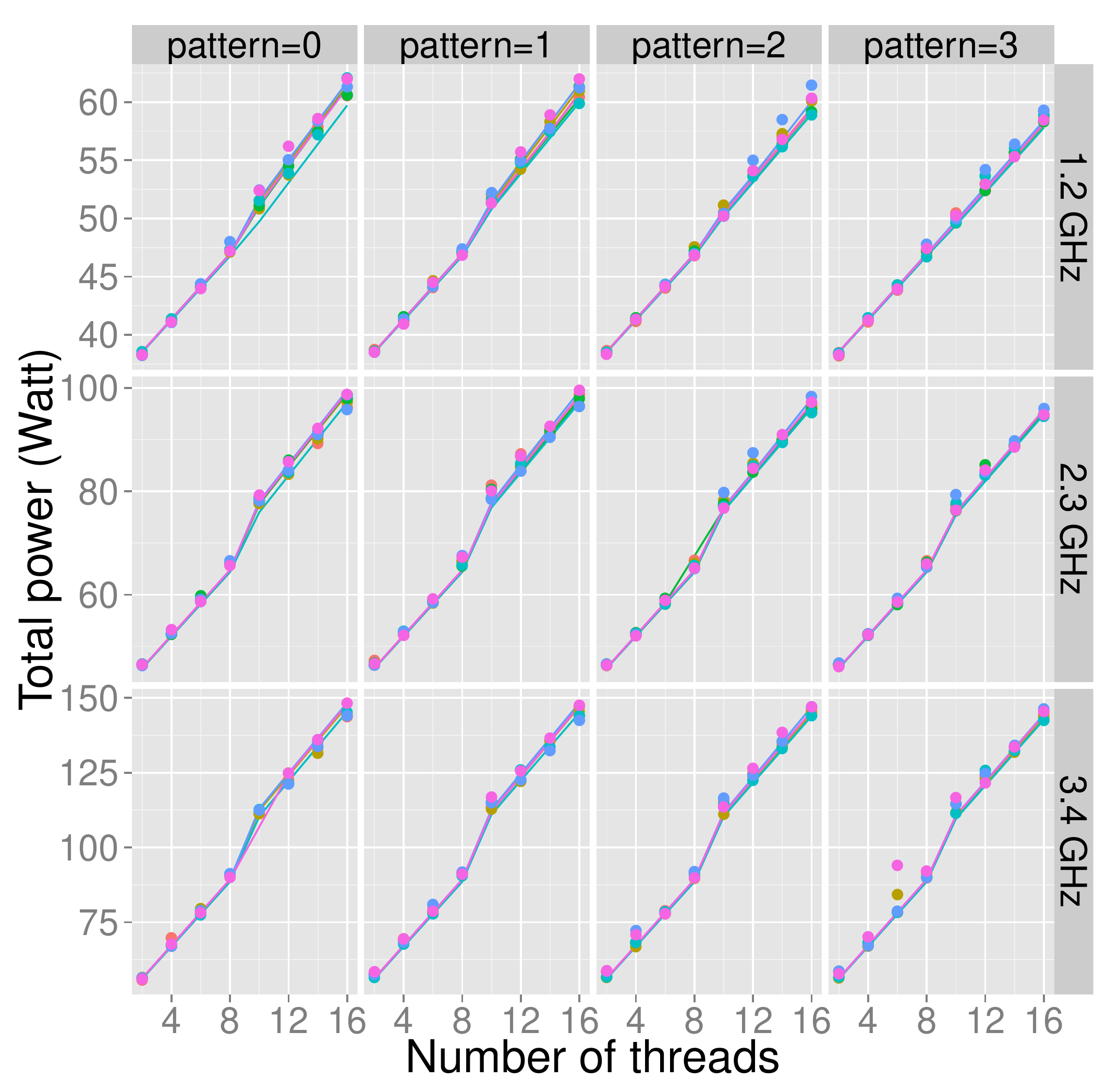}
\caption{Prediction on Mandelbrot simple implementation\label{fig.sim-man}}
\end{figure}

Regarding the memory consumption, we rely on a simple memory consumption model, in order
to estimate the intensity of remote accesses in the \ps. In the Mandelbrot application,
the number of cache misses in a given \ps is expected to depend on the pattern that is
used; more precisely, it should be roughly proportional to the size of the subregion
($2\times 2$, $4\times 4$ pixels, $\dots$) that is used in the application. We recall that
(i) the memory power is computed through $\pow{M} = \pstat{M} + \pdyn{M}$, where \pdyn{M}
is proportional to the number of accesses to remote or main memory per unit of time, (ii)
the ratio of time spent in the \rl is $\rat = 1 - (\thr \pw)/(\nth \facf \freq )$, (iii)
in the synthetic benchmark, \ie without memory accesses in the \ps, thanks to memory power
measurements, we have instantiated the value of \coda, such that the dynamic memory power
is $\pdyn{M} = \coda \nth \rat$.

Now, for the \ps of the Mandelbrot application, which contains memory accesses, we define
$\coda'$, such that $\coda' \nth$ would be the dynamic memory power if there were only
parallel sections and no retry-loop. As the application spends $\rat\%$ of the time in the
\ps, and $(1-\rat)\%$ in the \rl, the dynamic memory power can be
computed through:
\[ \pdyn{M} = \coda \times \nth \times (1-\rat) + \coda' \times \nth \times \rat. \]
Then, we run the application only once to obtain the value of $\coda'$.

\medskip

The results are presented in Figure~\ref{fig.sim-man}, where dashed lines and points are
the actual measurements, and plain lines are predictions. The key is again the one represented in
Figure~\ref{fig.key}, on page~\pageref{fig.key}.

\subsection{Concurrent Queues on Movidius Embedded Platform} \label{sec:concurrent-data-structures-Movidius}

As described in Section~\ref{sec:myriad1-sync} the Myriad platform
avails a number of different options for synchronizing SHAVEs and the
LEON processor.  Based on these we propose a number of different
concurrent queue implementations and evaluate them experimentally with
respect to performance and power consumption.

\subsubsection{Concurrent Queue Implementations}
\label{sec:4}
In this section we describe the concurrent queue implementations we evaluated on Myriad platform in the context of this work. Concurrent queues are used in a wide range of application domains, especially in the implementation of path-finding and work-stealing algorithms. The queue is implemented as a bounded cyclic array, accessed by all SHAVE cores. SHAVEs request concurrently access to the shared queue for inserting and removing elements.

Table \ref{tab:2} summarizes all different queue implementations we developed and evaluated on the platform. We used three different kinds of synchronization primitives: mutexes, message passing over shared variables and SHAVEs' FIFOs. Mutexes and SHAVEs' FIFOs were described in Section~\ref{sec:myriad1-sync}. With respect to the shared variables, we implemented communication buffers between the processors, used to exchange information for achieving synchronization. To reduce the cost of spinning on shared variables, we allocated these buffers in processor local memories, to avoid the congestion of the Myriad main buses.

The queue implementations can be divided in two basic categories: Lock based and Client-Server.

\paragraph{Lock-based Implementations}
\label{sec:4.1}
The lock-based implementations of the concurrent queue utilize the Mutex registers provided by the Myriad architecture. We implemented two different lock-based algorithms: In the first one, a single lock is used to protect the critical section of the \textit{enqueue()} function and a second one to protect the critical section of the \textit{dequeue()}. Therefore, simultaneous access to both ends of the queue can be achieved.   
The second implementation utilizes only one lock to protect the whole data structure.
%
\begin{figure}
\centering
\includegraphics[width=0.75\textwidth]{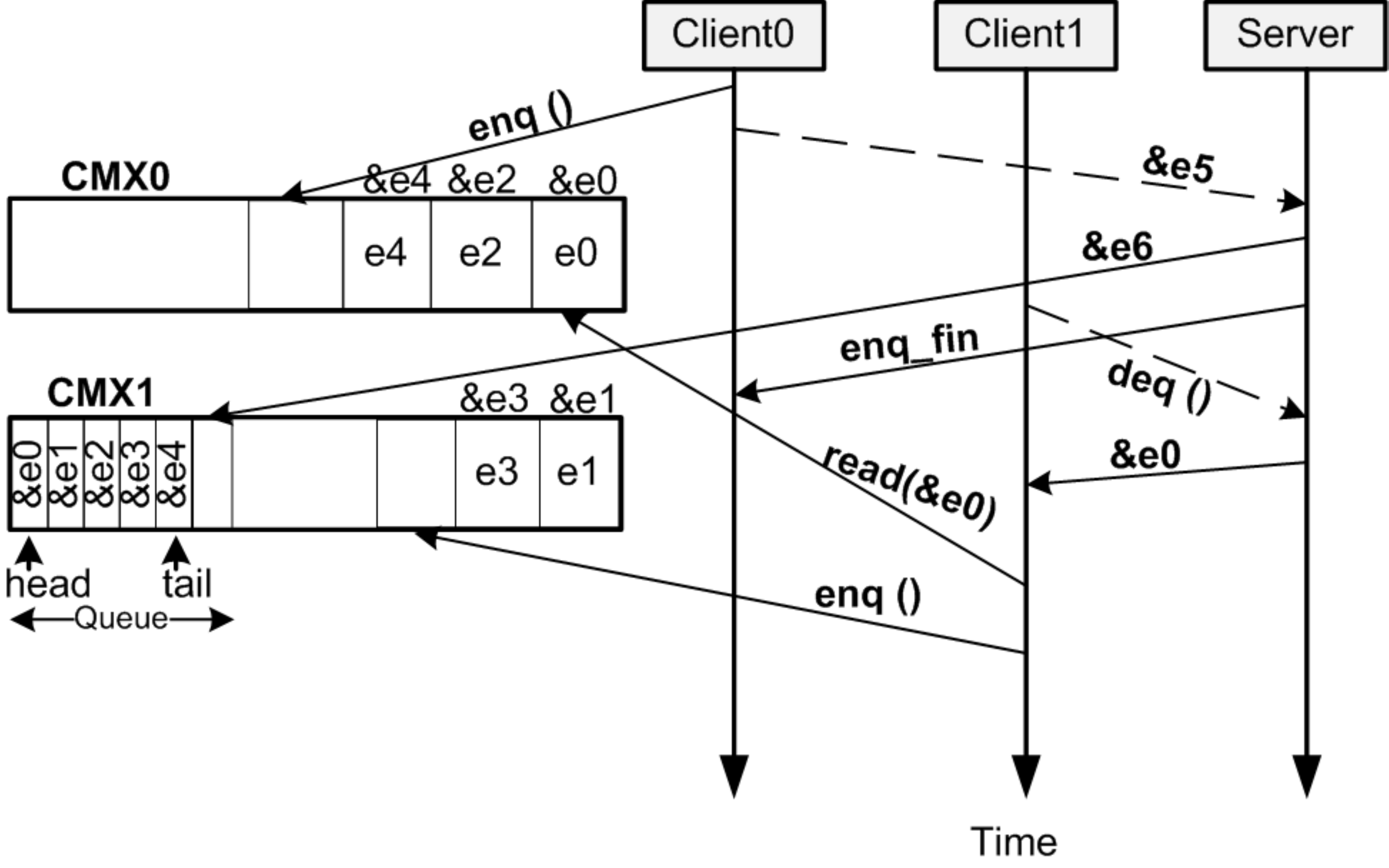}
\caption{Client-Server Implementation: The server maintains the order of the allocated elements by storing their addresses in a FIFO.}
\label{fig:2}       
\end{figure}
\paragraph{Client-Server Implementations}
\label{sec:4.2}
The Client-Server Implementations are based on the idea that a server arbitrates the access requests to the critical sections of the application and executes them. Therefore, the clients do not have direct access to the critical section. Instead, they provide the server with the required information for executing the critical section. This set of implementations is an adaptation of the Remote Core Locking algorithm (RCL)~\cite{PetrovicRS2014:LHM}\cite{LoziFGLM2012:RCL}. In Myriad platform the role of the server can be played either by LEON or a SHAVE core, as shown in Table \ref{tab:2}. The SHAVEs are the clients, requesting access to the shared data from the server.

To maximize the efficiency of the Client-Server implementations, each SHAVE allocates the elements to be enqueued in its local CMX slice. Although the CMX is much smaller in comparison with the DDR memory, it provides much smaller access time for the SHAVEs than, for example, with the DDR. The server allocates the queue in a CMX slice, since DRAM is much smaller.

We implemented two versions of the Client-Server synchronization algorithms. In the first one, the server maintains the FIFO order of the queue by storing the addresses of the allocated elements in a FIFO manner. In an enqueue, the client allocates the element in its local CMX slice and then sends the address of the element to the server, which pushes the address in the queue. In the dequeue case, the client sends a \textit{dequeue} message to the server and waits for the server to respond with the address of the dequeued element.

Figure~\ref{fig:2} illustrates this implementation with only two clients: \textit{Client0} enqueues element \textit{e5} in \textit{CMX0} and sends the address to the server. The server pushes the address in the queue and notifies the client that the enqueue has finished with an \textit{enq\textunderscore fin} message. \textit{Client1} requests a dequeue and the server responds with the address of the \textit{e0} element.

We evaluated this algorithm by designing several alternatives: In the first one, the server is the LEON processor and in the second is a SHAVE. In Table~\ref{tab:2} these are displayed as \textit{Leon--Srv--addr} and \textit{SHAVE--Srv--addr} respectively. Also, we experimented with both shared variables and SHAVEs' FIFOs synchronization primitives.

The main advantage of this algorithm is that it reduces the stalling of the clients, especially during the enqueue operations. The client sends the address to the server and then can proceed with other calculations, without waiting for the server to respond. This applies especially in the case where the SHAVE FIFO synchronization primitive is used. The client stalls only when the server's FIFO is full. Additionally, the fact that the element allocation takes place only in local memories reduces both the execution time and the power consumption. Another parameter that affects the efficiency of the algorithm is the synchronization primitive used. We expect SHAVEs' FIFOs primitive to be efficient both in terms of performance and power consumption, since it avoids memory accesses during the exchange of information between the clients and the server. However, the main disadvantage in this case is that it can be only implemented using a SHAVE as a server.
%
\begin{figure}
\centering
\includegraphics[width=0.75\textwidth]{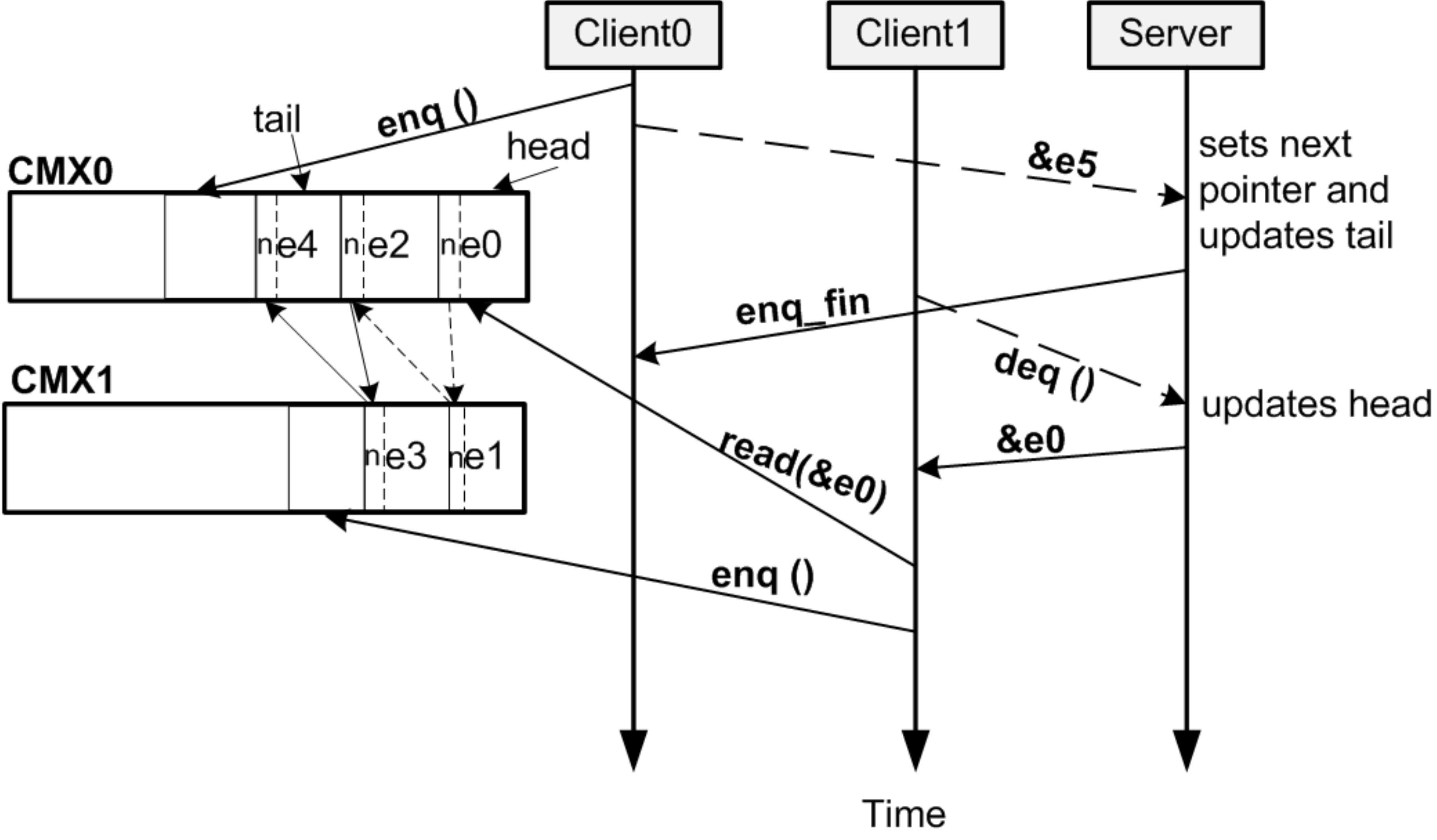}
\caption{Client-Server Implementation: The server maintains FIFO order of the allocated objects by pointing to the first and the last enqueued elements.}
\label{fig:3}       
\end{figure}

In the second Client-Server implementation we altered the queue structure as follows: The server, instead of managing a queue to store object addresses, utilizes two pointers: \textit{head} and \textit{tail} that point at the first and the last element allocated respectively. Additionally, each element has a \textit{next} pointer which points to the next element, keeping in this way the FIFO order. All these pointers are managed by the server, in order to improve the application parallelism by allowing the clients to perform tasks only outside the critical section.

When a client allocates an object in its local queue stored in CMX memory, it sends the address to the server, as in the first implementation. When the server receives the address, first it updates the \textit{next} pointer of the last element to point to the newly allocated element. Then, it updates the \textit{tail} pointer to point to the new element. (This is the same that happens in a singly linked list FIFO data structure). In the dequeue, as soon as the server receives the request, it sends to the client the address of the element pointed by the \textit{head} and then updates the \textit{head}, using the \textit{next} pointer of the dequeued element.

To illustrate this algorithm, Figure~\ref{fig:3} shows an example. Initially, five elements exist in the queue. \textit{e0} is the first allocated and \textit{e4} is the last one. Therefore, head points to \textit{e0} and \textit{tail} to \textit{e4}. \textit{Client0} allocates element \textit{e5} an element in \textit{CMX0} and sends the address to the server. The server sets the \textit{next} pointer of \textit{e4} to point to \textit{e5} and updates the \textit{tail} pointer to \textit{e5} as well. Then, it sends an \textit{enq\textunderscore fin} message to \textit{Client0}. \textit{Client1} requests a dequeue. The server receives the message, updates the \textit{head} pointer and sends the address of \textit{e0} to the client.

In comparison with the \textit{Leon--Srv--addr} or \textit{SHAVE--Srv--addr}, this implementation consumes less memory space. Therefore, the space available in each local CMX slice is affected only by the number of allocated elements of the corresponding client, unlike the previous implementation where the queue of stored addresses reduced the available memory of the slice where it was allocated. It is important to mention that each client accesses only its local CMX slice during the enqueue and the dequeue operations. Only the server makes inter-slice accesses. The disadvantage of this implementation is that it cannot be efficiently implemented using LEON as a server, since it accesses the \textit{next} pointers of each element with high cost. The implementation was designed using both shared variables and SHAVEs' FIFOs for communication between the server and the clients. In Table \ref{tab:2} is displayed as \textit{SHAVE-srv\textunderscore ht}.

\begin{table}
\begin{tabular}{llll}
\hline\noalign{\smallskip}
 & \multicolumn{3}{c}{Synchronization Primitive} \\
 & Mutex & Shared Var & SHAVE FIFO \\
 \noalign{\smallskip}\hline\noalign{\smallskip}
 no server & Y & - & - \\
 Leon-srv-addr & - & Y & - \\
 SHAVE-srv-addr & - & Y & Y \\
 Leon-srv-h/t & - & - & - \\
 SHAVE-srv-h/t & - & Y & Y \\
\noalign{\smallskip}\hline
\end{tabular}
\caption{Concurrent Queue Implementations: ("Y" indicates the ones that are evaluated in this work.\label{tab:2})}
\end{table}

\subsubsection{Experimental Evaluation}

The concurrent queue implementations were evaluated using a synthetic benchmark, which is composed by a fixed workload of 20,000 operations and it is equally divided between the running SHAVEs. In other words, in an experiment with 4 SHAVEs each one completes 5,000 operations, while in an experiment with 8 SHAVEs, each one completes 2,500 operations.  In the implementations where a SHAVE is utilized as a server, we run the experiments using up to 6 clients (to have the same number of enqueuers and dequeuers).

All algorithms were evaluated in terms of time performance, for the given fixed workload, which is expressed in number of execution cycles. More specifically, in Myriad platform the data flow is controlled by LEON. SHAVEs start their execution when instructed to do so by LEON and then LEON waits for them to finish. The number of cycles measured is actually LEON cycles from the point that SHAVEs start their execution until they all finish. This number represents accurately the execution time. Power consumption was measured using a shunt resistor connected at the 5V power supply cable. Using a voltmeter attached to the resistor's terminals we calculated the current feeding the board and therefore the power consumed by the Myriad platform. 
\remind{Must state that this different from the method used in the Myriad microbenchmark section.}
\begin{figure}[b!]
\centering
\includegraphics[width=0.80\textwidth]{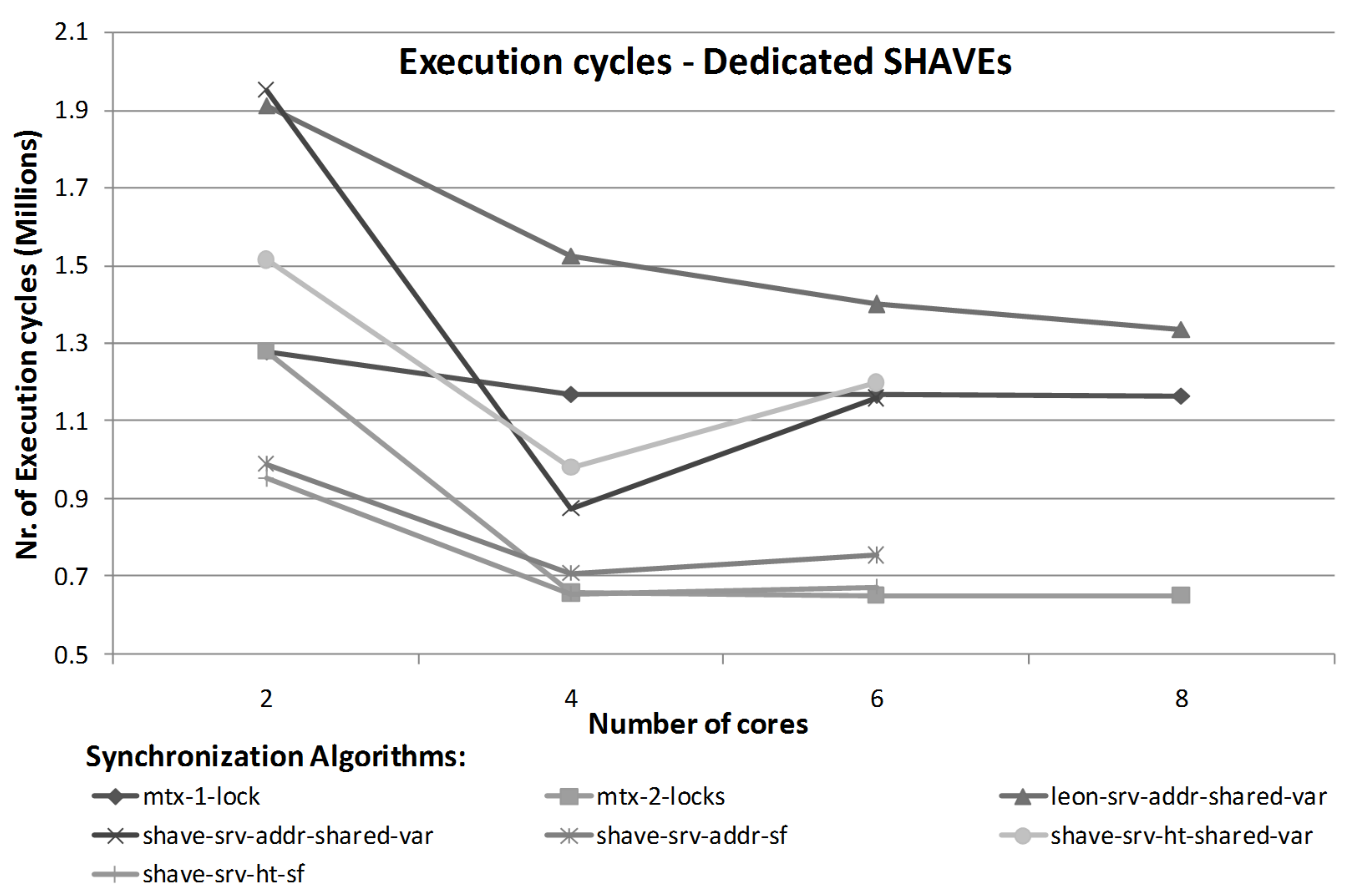}
\caption{Execution cycles for different synchronization algorithms, when half of the SHAVEs perform enqueue and half dequeue operations.}
\label{fig:4}       
\end{figure}

We performed two sets of experiments for evaluating the behavior of the designs: dedicated SHAVEs and random operations. In the ``dedicated SHAVEs'' experiment each SHAVE performs only one kind of operations. In other words, half of the SHAVEs enqueue and half dequeue elements to / from the data structure.  ``Random operations'' means that each SHAVE has equal probability to perform either an enqueue or a dequeue each time it prepares its next operation. 
%
\begin{figure}[t]
\centering
\includegraphics[width=0.80\textwidth]{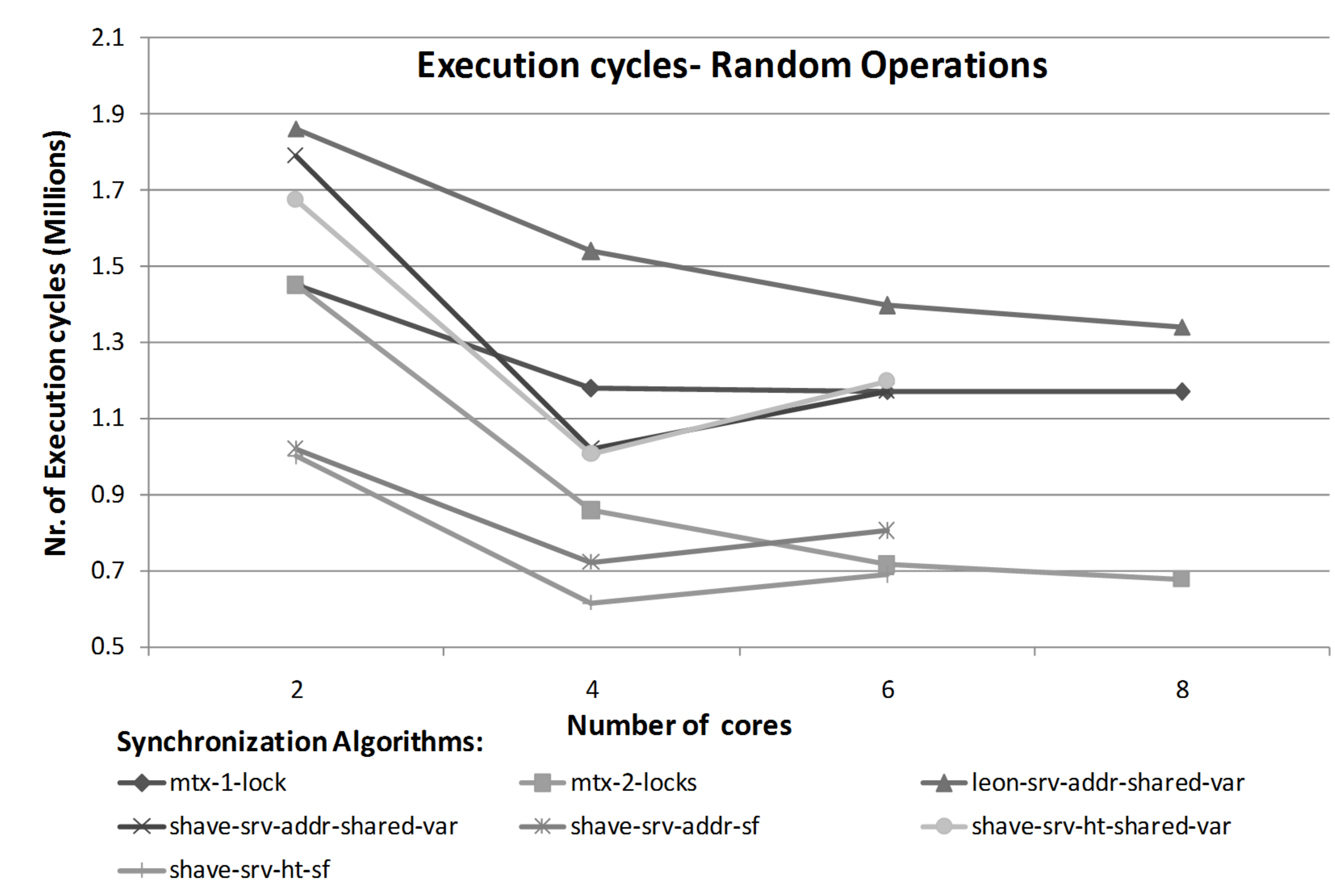}
\caption{Execution cycles for different synchronization algorithms, when the SHAVEs perform randomly enqueue and dequeue operations.}
\label{fig:5}       
\end{figure}
\paragraph{Execution Time Evaluation}
\label{sec:5.1}
In this subsection we present the experimental results for performance. \textit{mtx-2-locks} is the lock-based queue implementation with 2 locks, while \textit{mtx-1-lock} is the same implementation with a single lock. \textit{leon-srv-addr} refers to the Client-Server implementation, where the server is LEON and stores the addresses of the objects in a queue, while \textit{SHAVE-srv-addr} is the same implementation where a SHAVE is the server. \textit{leon-srv-ht} and \textit{shave-srv-ht} refers to the Client-Server implementation where the server (LEON and a SHAVE respectively) manages a head and a tail pointer to maintain the FIFO order. Finally, ``shared-var'' means that the communication is achieved using a shared buffer (i.e. shared variables) and ``sf'' means that the communication is made through the SHAVEs' FIFOs.

Figure~\ref{fig:4} and Figure~\ref{fig:5} show the execution time of dedicated SHAVEs and random operations respectively. We notice that the \textit{mtx-2-locks} implementation is the fastest one in the case of 8 SHAVEs and seems to scale well. This is expected, since it provides the maximum possible concurrency. It requires about half the number of execution cycles in comparison with the \textit{mtx-1-lock}. 

The SHAVEs' FIFOs implementations perform well, especially in the case of 4 SHAVEs in the random operations experiment (28.3\% in comparison to the mtx-2-locks). The utilization of SHAVEs' FIFOs for communication seems to be very efficient in terms of execution time. On the other hand, shared variables provide much lower execution time: For example, \textit{shave-srv-addr-shared-var} leads to 53.3\% more execution cycles in comparison with the \textit{shave-srv-addr-sf} in the dedicated SHAVEs experiment with 6 SHAVEs. Also, the implementations where the server maintains a head and tail pointer performs slightly better in most cases in comparison with the one where the server stores the addresses of the elements (up to 16.7\% in random operations for the 6 SHAVEs experiment). 

In most implementations, we notice a very large drop in the execution time from 2 to 4 SHAVEs, due to the increase of concurrency. In other words, in the 2 SHAVEs experiments there are time intervals where no client requests access to the shared data. However, in case of 4 clients or more, there is always a SHAVE accessing the critical section. Since the workload is fixed, there is a large drop in the execution time compared to the 2 client experiment. However, since access to the critical section is serialized, the execution time drop for more than 4 SHAVEs is much smaller (e.g. in case of \textit{mtx-2-locks}) or even non-existent (e.g. in \textit{mtx-1-lock}). 

Finally, in all experiments where a SHAVE is utilized as a server there is an increase in execution time from 4 to 6 SHAVEs. The reason for that is the overhead added by inter-slice accesses, which is larger than the decrease in execution time due to increased concurrency. The utilization of LEON as a server using shared variables for communication (i.e. \textit{leon-srv-addr-shared-var}) is inefficient since two overheads are accumulated: LEON is accessing variables in the CMX memory (which is more costly in comparison with SHAVEs) and the spinning on shared variables for achieving communication. 
%
\begin{figure}
\centering
\includegraphics[width=0.75\textwidth]{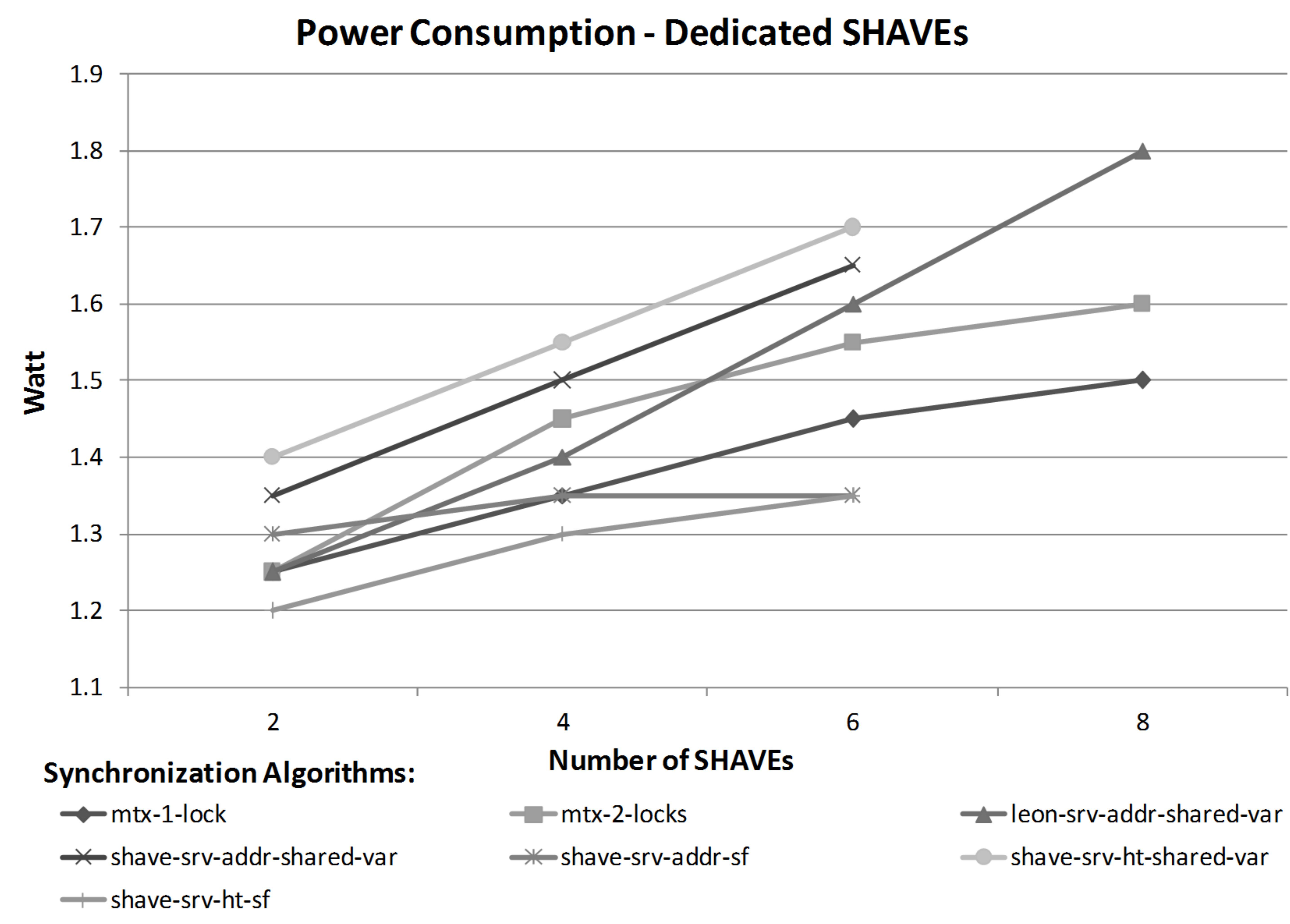}
\caption{Power consumption for different synchronization algorithms, when half of the SHAVEs perform enqueue and half dequeue operations.}
\label{fig:6}       
\end{figure}
%
\begin{figure}
\centering
\includegraphics[width=0.75\textwidth]{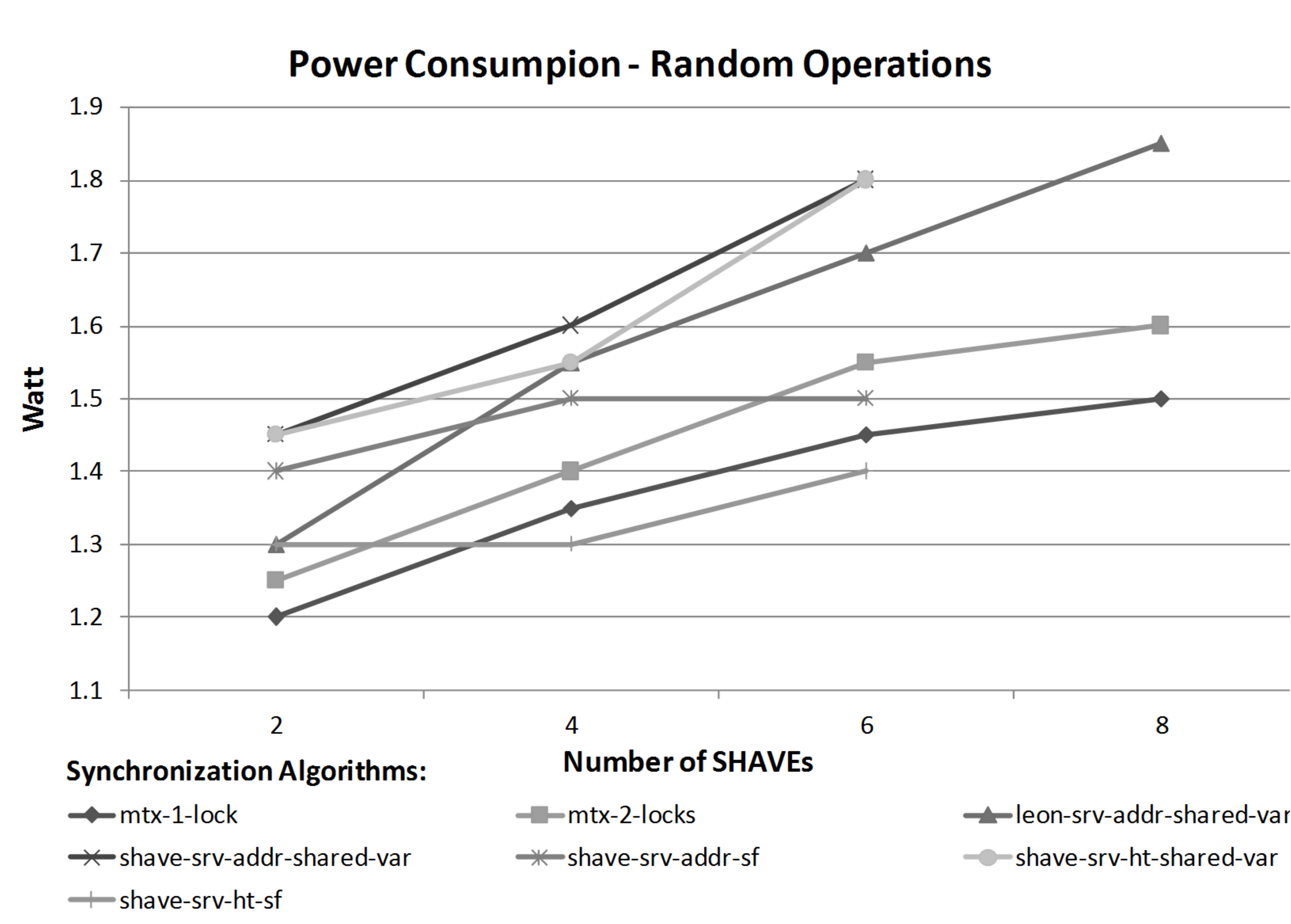}
\caption{Power consumption for different synchronization algorithms, when the SHAVEs perform randomly enqueue and dequeue operations.}
\label{fig:7}       
\end{figure}

\paragraph{Power Consumption Evaluation}
\label{sec:5.2}

%
\begin{figure}[b!]
\centering
\includegraphics[width=0.75\textwidth]{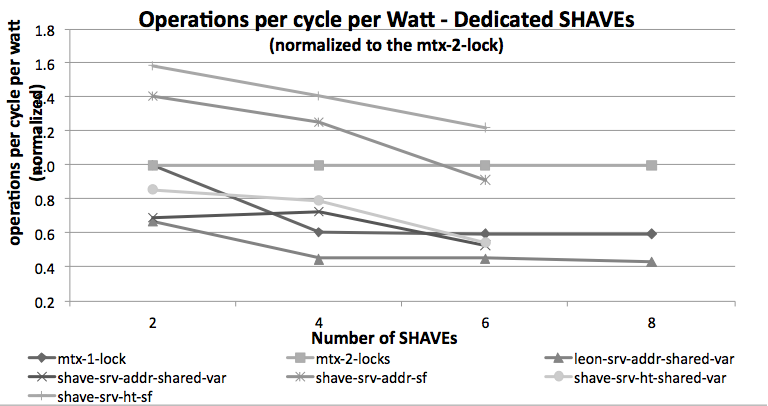}
\caption{Normalized energy per operation of the synchronization algorithms, when half of the SHAVEs perform enqueue and half dequeue operations.}
\label{fig:8}       
\end{figure}

As previously stated, power consumption was measured using a shunt resistor connected to the power supply of the platform. Figure~\ref{fig:6} and Figure~\ref{fig:7} show the power consumption for dedicated SHAVEs and random operations. For the 8 SHAVEs experiment the most power efficient implementation is the lock-based with a single lock (6.25\% in dedicated SHAVEs in comparison with the \textit{mtx-2-locks}). However, for a smaller number of clients, the SHAVEs' FIFOs implementations are the most power efficient. Indeed, power consumption drops up to 6.8\% for 6 SHAVEs in the dedicated SHAVEs experiment. 

We notice that the lock-based implementation with a single lock consumes less power than the 2-lock implementation. This is due to the fact that the power consumption is affected by the number of SHAVEs accessing the memory concurrently. In the single lock implementation only one SHAVE accesses the memory for performing operations. However, in the 2-locks implementation there are 2 SHAVEs which perform operations concurrently, while the rest are spinning on the locks. Therefore, the 2-lock algorithm consumes more power. 

Spinning on a lock consumes very low power, because mutexes are hardware implemented. Microbenchmarking experiments show that 8 SHAVEs spinning on a lock concurrently, consume about 20\% less power in comparison with the case where 8 SHAVEs access the memory concurrently. In fact, this is the case with the shared variables synchronization primitive. All shared variable implementations consume a lot of power, because spinning on a memory location is energy inefficient, even if the spinning takes place in a local CMX slice.  

SHAVEs' FIFOs communication method is the most energy efficient. When a SHAVE tries to write in a full FIFO or read from an empty FIFO stalls, until the FIFO gets non-full or non-empty respectively. Microbenchmarking experiments we performed show that 8 SHAVEs stalling in a FIFO consume about 28\% less power than spinning on a mutex. Indeed, stalling in FIFOs is common in our experiments, where the contention is high. The fact that this synchronization algorithm avoids spinning on memory locations and set SHAVEs to stall mode makes it very power efficient. 

\begin{figure}
\centering
\includegraphics[width=0.75\textwidth]{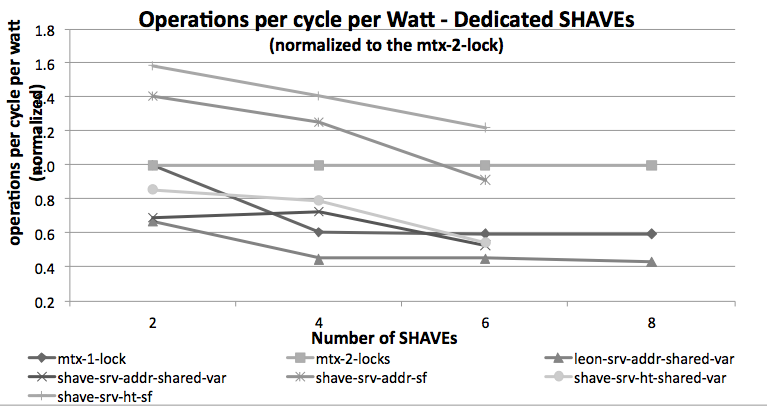}
\caption{Normalized energy per operation of the synchronization algorithms, when the SHAVEs perform randomly enqueue and dequeue operations.}
\label{fig:9}       
\end{figure}
\paragraph{Energy per operation Evaluation}
\label{sec:5.3}
To evaluate in more depth the synchronization algorithms, we present the energy per operation results in Figure~\ref{fig:8} and Figure~\ref{fig:9} for the dedicated and the random operations experiments respectively. The results are normalized to the \textit{mtx-2-locks} calculated values. We notice that the RCL implementations that utilize the SHAVE's FIFOs for communication between the clients and the server consume in almost all cases less energy per operation than the \textit{mtx-2-locks}. In the random operation experiment, \textit{shave-srv-ht-sf} consumes 33\% less energy per operation in comparison with the \textit{mtx-2-locks}. Indeed, when utilizing SHAVEs FIFOs instead of memory buffers for arbitration between the SHAVEs, the energy consumption is low. The shared buffer communication scheme is proven to be inefficient in terms of energy consumption. For instance, \textit{leon-srv-addr-shared-var} consumes more than two times energy per operation in comparison with the \textit{mtx-2-locks}.
\paragraph{Discussion}
\label{sec:5.4}
The mutex synchronization primitive is indeed efficient for the concurrent queue implementation in terms of both performance and power consumption (especially in the Myriad platform, where mutexes are hardware implemented and very optimized). However, our results show that RCL implementations provide very promising results for the queue implementations and in most cases perform similar to the lock-based ones. We expect that in future MPSoCs, where the number of cores will increase even further, client-server implementations will become even more efficient. 

The reason that the RCL implementations seem an attractive alternative to the lock-based ones is the fact that they transfer computational overhead of the critical sections from the cores, which are the queue workers (clients) to another dedicated core that plays the role of the server. The computational overhead of inserting and removing elements to / from the queue is transferred from the clients to the server. Therefore, while the server executes the critical section, the clients can proceed with other computations, thus increasing the parallelism and reducing the application execution time. Furthermore, by minimizing the communication overhead between the clients and the server (e.g. by utilizing the SHAVE's FIFOs), the results are very satisfactory. On the other hand, in the lock-based implementations, the computational overhead of accessing the queue elements is handled by the workers. However, in this case, simultaneous accesses to the data structure can be achieved, which is obviously not possible in the RCL algorithm. However, with these experiments we show that the RCL algorithm should be evaluated in embedded systems along with the lock-based solutions, especially in applications that use data structures which allow relatively low level of parallelism.

\clearpage

\section{Analysis of Energy Consumption of Concurrent Search Trees: $\Delta$Trees} \label{sec:DeltaTree_description}
\subsection{Introduction} \label{sec:DToverview}
Concurrent trees are fundamental data structures that are widely used in different contexts such as load-balancing \cite{DellaS00, HaPT07, ShavitA96} and searching \cite{Afek:2012:CPC:2427873.2427875, BronsonCCO10, Brown:2011:NKS:2183536.2183551, Crain:2012:SBS:2145816.2145837, DiceSS2006, EllenFRB10}.
Most of the existing highly-concurrent search trees are not considering the fine-grained
data locality. The non-blocking concurrent search trees 
\cite{Brown:2011:NKS:2183536.2183551, EllenFRB10} and Software Transactional
Memory (STM) search trees
\cite{Afek:2012:CPC:2427873.2427875, BronsonCCO10,
Crain:2012:SBS:2145816.2145837, DiceSS2006} have been regarded
as the state-of-the-art concurrent search trees. They have been proven
to be scalable and highly-concurrent. 
However these trees are not designed for fine-grained data locality. 
Prominent concurrent search trees which are often included in several benchmark 
distributions such as the concurrent red-black
tree \cite{DiceSS2006} by Oracle Labs and the concurrent AVL tree
developed by Stanford \cite{BronsonCCO10} are not designed for data locality either. It is challenging to devise search trees that are portable, highly concurrent and fine-grained locality-aware. A platform-customized locality-aware search trees \cite{KimCSSNKLBD10, Sewall:2011aa} are not portable while there are big interests of concurrent data structures for unconventional platforms~\cite{Ha:2010aa, Ha:2012aa}. Concurrency control techniques such as transactional memory~\cite{Herlihy:1993aa,Ha:2009aa} and 
multi-word synchronization~\cite{Ha:2005aa,Ha:2003aa,Larsson:2004aa} do not take into account fine-grained locality while fine-grained locality-aware techniques such as van Emde Boas layout \cite{Prokop99,vanEmdeBoas:1975:POF:1382429.1382477} poorly support concurrency.

\subsubsection{I/O model}
One of the most cited memory models is is the two-level I/O model
of Aggarwal and Vitter \cite{AggarwalV88}. In their seminal paper, Aggarwal and Vitter
postulated that the memory hierarchy consists of two levels, a main memory with size
$M$ (e.g. DRAM) and a secondary memory of infinite size (e.g. disks). Data is transferred 
in $B$-sized blocks between those two levels of memory and CPUs can only access data which are available in the main memory. In the I/O model, an algorithm time complexity is assumed to be dominated by how many block transfers are required.  

The simplicity and feasibility of this model has made this model popular. However 
to use this model, an algorithm has to know the $B$ and $M$ parameters in advance.
The problem is that these parameters are sometimes unknown 
and most importantly not portable between platforms.
For this I/O model, B-tree \cite{Bayer:1972aa} is an optimal search tree \cite{CormenSRL01}.

Concurrent B-trees \cite{BraginskyP12, Comer79, Graefe:2010:SBL:1806907.1806908,
Graefe:2011:MBT:2185841.2185842}
are optimised for a known memory block size $B$ (e.g. page size) to minimise the
number of memory blocks accessed during a search, thereby improving data
locality. 
In reality there are different block sizes at different levels of the memory hierarchy
that can be used in the design of locality-aware data layout for search trees. 
For example in \cite{KimCSSNKLBD10, Sewall:2011aa}, 
Intel engineers have come out with very fast search trees by crafting  
a platform-dependent data layout based on the register size, SIMD width, cache line size,
and page size. 
Namely, existing concurrent B-trees limits its spatial locality optimisation to the memory level
with block size $B$, leaving access to other memory levels with a different block size
unoptimised.

For example in this I/O model, a traditional B-tree that is optimised for searching data in disks (i.e.
$B$ is page size), where each node is an array of sorted keys, 
is optimal for transfers between a disk and RAM (cf. Figure \ref{fig:vEB-bfs}c).
However data transfers between RAM and last level cache (LLC) 
is no longer optimal.
For searching a key inside each $B$-sized block in RAM, the transfer complexity 
is $\Theta (\log (B/L))$ transfers between RAM and LLC, 
where $L$ is the cache line size.
Note that a search with optimal cache line transfers of $O(\log_L B)$ is achievable
by using the van Emde Boas layout \cite{Brodal:2004aa}.

\subsubsection{Ideal-cache model}
The \textit{ideal-cache} model was introduced by Frigo et. al. in \cite{Frigo:1999:CA:795665.796479},
which is similar to the I/O model except that the block size $B$ and memory size
$M$ are unknown.
This paper has coined the term \textit{cache-oblivious algorithms}. Using same
analysis of the Aggarwal and Vitter's two-level I/O model, an algorithm is categorised as 
\textit{cache-oblivious} if it has no variables that need to be tuned with respect to 
hardware parameters, 
such as cache size and cache-line length in order to achieve optimality, 
assuming that I/Os are performed by an optimal off-line cache replacement strategy.

Cache-oblivious algorithms by default have the optimal temporal locality, mainly
because of the unknown $M$. Thus, cache-oblivious algorithms mainly concentrate
on optimising spatial locality. Because it is optimal for an arbitrary size of the two levels of memory, a cache-oblivious
algorithm is also optimal for any adjacent pair of available levels of the memory hierarchy. Therefore without knowing anything about memory level hierarchy and the size of each level, a cache-oblivious
algorithm can automatically adapt to multiple levels of the memory hierarchy.

Empirical results have showed that a cache-oblivious algorithms are often outperform 
the basic RAM algorithms but 
not always as good as the carefully tuned (cache-aware) algorithms. However 
cache-oblivious algorithms perform better on multiple levels of memory hierarchy and are
more robust despite changes in memory size parameters compared to the cache-aware
algorithms \cite{Brodal:2004aa}. 

It is important to note that in the ideal-cache model, B-tree is no longer optimal
because of the unknown $B$. Instead, the trees that are described in  
the seminal paper \cite{vanEmdeBoas:1975:POF:1382429.1382477}  
by Peter van Emde Boas, are optimal.
The van Emde Boas (vEB) tree is an ordered dictionary data type which implements
the idea of a recursive structure of priority queues.
The efficient structure of the vEB tree arranges data
in a recursive manner so that related values are placed in contiguous memory
locations (cf. Figure \ref{fig:vEB}). This work has inspired many cache-oblivious 
data structures 
such as cache-oblivious B-trees \cite{BenderDF05, BenderFFFKN07, BenderFGK05} and cache-oblivious
binary trees \cite{BrodalFJ02}. These researches have demonstrated that 
vEB layout is suitable for cache oblivious algorithms as it lowers the upper bound on memory transfer complexity.

\begin{figure}[!t]
\centering \input{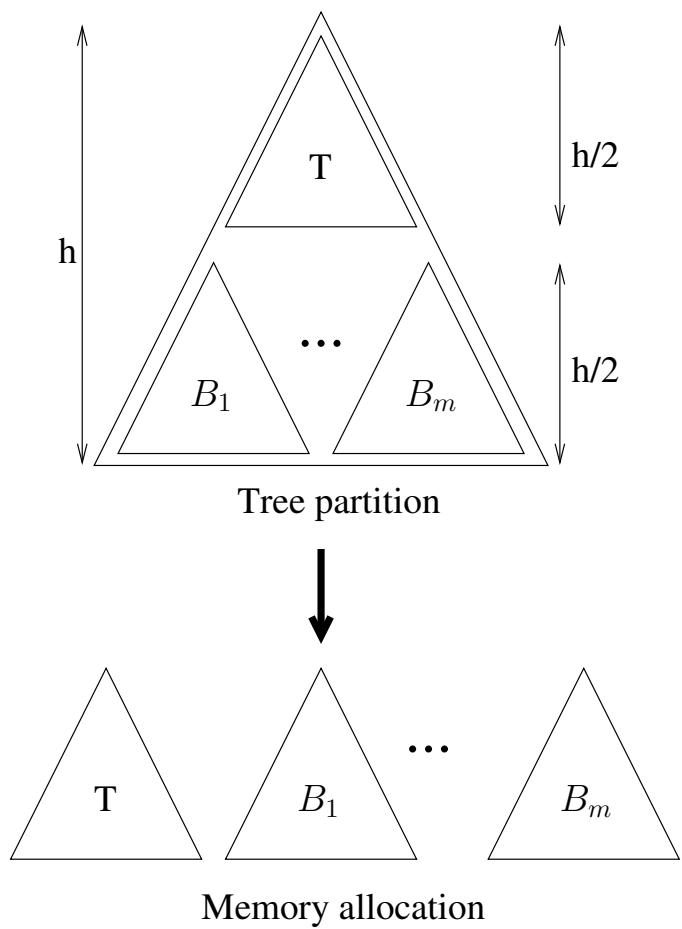}
\caption{Static van Emde Boas (vEB) layout: a tree of height $h$ is recursively split at height $h/2$. 
The top subtree $T$ of height $h/2$ and $m=2^{h/2}$ bottom subtrees $B_1;B_2; \ldots ;B_m$ 
of height $h/2$ are located in contiguous memory locations in the order of 
$T;B_1;B_2;\ldots;B_m$.}\label{fig:vEB}
\end{figure}

For example in a system where block size $B=3$, a search tree with Breadth First Search layout (or BFS tree for short) 
(cf. Figure \ref{fig:vEB-bfs}a) with height 4 will need to do three memory transfers to locate the key in leaf-node 13 in top-down traversing. The first two levels with
three nodes (1, 2, 3) will fit within a single block transfer, while the other two levels need to be loaded
in two separate memory transfers, where each of the transfer contains (6, 7, 8) and (13, 14, 15) nodes, 
respectively.
Therefore, required memory transfers is $(\log_2N-\log_2B) = \log_2 N/B \sim \log_2 N$
for $N \gg B$. 

However, for a vEB tree with the same height, the required memory transfers is only 
two. As seen in Figure \ref{fig:vEB-bfs}b, locating the key in leaf-node 12 requires only (1, 2, 3) nodes transfer 
followed by (10, 11, 12) nodes transfer.
This means the transfer complexity is now reduced to $\frac{\log_2N}{\log_2B} = \log_B N$, simply by crafting 
an efficient data structure so that nearby nodes are located in adjacent memory locations. 
If $B=1024$, traversing a BFS tree requires 10x more I/Os than a vEB tree. 

So far the vEB layout is shown to have $\log_2B$ less I/Os for two-level memory. 
On commodity machines where exists multiple-level memory, 
the vEB layout is outright efficient. In a typical machine having three 
levels of cache (with 64B cache line size), a RAM (with 4KB page size), and a disk;
vEB tree can deliver up to 640x less I/Os than BFS tree, assuming node size is 4 bytes (Figure \ref{fig:vEB-bfs}c).

\begin{figure}[!t]
\centering  \includegraphics[width=0.7\columnwidth]{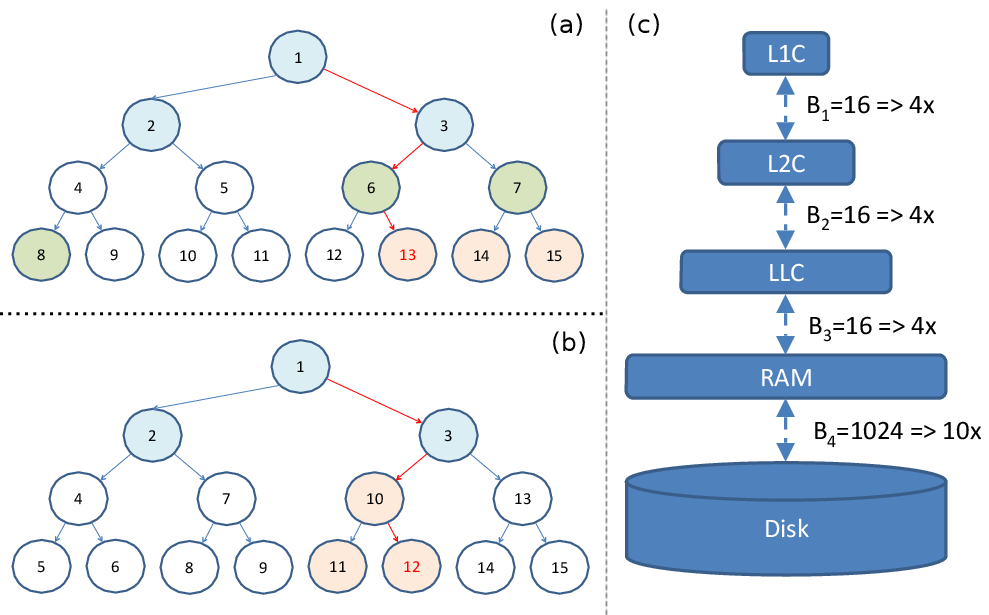}
\caption{Illustration of required memory transfers in searching for key 13 in (a) BFS tree layout and 
(b) vEB tree layout. An example of multiple levels of memory is shown in (c). $B_x$ is the block size $B$  between levels of memory.}\label{fig:vEB-bfs}
\end{figure}

However, while proven to perform well in searching, vEB trees poorly support concurrent update operations.
Inserting or deleting a node in a tree may result in relocating a large part of
the tree in order to maintain the vEB layout. The work in  \cite{BenderFGK05} has discussed this problem but a feasible implementation hasn't been reported yet \cite{BraginskyP12}. We would like to refer the readers to \cite{Brodal:2004aa, Frigo:1999:CA:795665.796479} 
for a more comprehensive overview of the I/O model and the ideal-cache model.

We introduce \textit{$\Delta$Tree} family, novel locality-aware concurrent search trees.
$\Delta$Tree is an unbalanced locality-aware concurrent search tree of $\Delta$Nodes 
whose Search, Insert, and Delete are non-blocking to each other, 
but Insert and Delete may be occasionally blocked by maintenance
operations within a $\Delta$Node. $\Delta$Node is a fixed size tree-container with the van Emde Boas layout (cf. Figure \ref{fig:dynamicVEB}(left)).

\textit{BalancedDT} is a balanced $\Delta$Tree with pointer-less $\Delta$Nodes, 
  enabling $\sim$200\% more keys to fit in a $\Delta$Node, resulting in 90\% improvement 
  in performance and 2x improvement in energy efficiency compared to $\Delta$Tree.

\textit{HeterogeneousBDT} is a "heterogeneous" BalancedDT where its inner $\Delta$Nodes
  are leaf-oriented but its leaf $\Delta$Nodes are not, enabling 100\% more keys fitting in the leaf $\Delta$Nodes, 
  resulting in 20\% improvement in performance and 50\% improvement in energy efficiency compared to BalancedDT.

\subsection{Overview on the van Emde Boas layout}

We propose a modification to the original (static) van Emde Boas (vEB) layout
to support high concurrency and fast update operations. This effort results in cache-oblivious concurrent search trees in the form of a {\em dynamic} vEB layout. 
We first define the following notations that will be use to elaborate more on the idea:
\begin{itemize}

\item $b_i$ (unknown): block size in terms of the number of nodes at level $i$ of the memory
hierarchy (like $B$ in the I/O model \cite{AggarwalV88}), which is unknown as in the cache-oblivious model \cite{Frigo:1999:CA:795665.796479}. When the specific level $i$ of the memory
hierarchy is irrelevant, we use notation $B$ instead of $b_i$ in order to be
consistent with the I/O model.

\item $UB$ (known): the upper bound (in terms of the number of nodes) on the
block size $b_i$ of all levels $i$ of the memory hierarchy.

\item {\em $\Delta$Node}: the largest recursive subtree of a van Emde Boas-based 
search tree that contains at most $UB$ nodes (cf. dashed triangles of height $2^L$ in
Figure \ref{fig:search_complexity}). $\Delta$Node is a fixed-size tree-container
with the vEB layout.

\item "level of detail" $k$ is a partition of the tree into recursive subtrees of 
height at most $2^k$. 

\item Let L be the level of detail of $\Delta$Node. Let $H$ be the height
of a $\Delta$Node, we have $H = 2^L$. For simplicity, we assume $H = \log_2
(UB+1)$.

\item $N, T$: size and height of the whole tree in terms of basic nodes (not in
terms of $\Delta$Nodes).

\end{itemize}  

\subsubsection{Static van Emde Boas (vEB) layout} \label{subsec:staticveb}

The conventional {\em static} van Emde Boas (vEB) layout has been introduced in
cache-oblivious data structures \cite{BenderDF05, BenderFFFKN07, BenderFGK05, BrodalFJ02,
Frigo:1999:CA:795665.796479}. Figure \ref{fig:vEB} illustrates the vEB layout.
Suppose we have a complete binary tree with height $h$. For simplicity, we
assume $h$ is a power of 2, i.e. $h=2^k$.
The tree is recursively laid out in the memory as follows. The tree is
conceptually split between nodes of height $h/2$ and $h/2+1$, resulting in a top
subtree $T$ and $m_1 = 2^{h/2}$ bottom subtrees $B_1, B_2, \cdots, B_m$ of
height $h/2$. The $(m_1 +1)$ top and bottom subtrees are then located in
consecutive memory locations in the order of subtrees $T, B_1, B_2, \cdots,
B_m$. Each of the subtrees of height $h/2$ is then laid out similarly to $(m_2 +
1)$ subtrees of height $h/4$, where $m_2 = 2^{h/4}$. The process continues until
each subtree contains only one node, i.e. the finest {\em level of detail}, 0.

The main feature of the vEB layout is that the cost of any search in this layout
is $O(\log_B N)$ memory transfers, where $N$ is the tree size and $B$ is the {\em
unknown} memory block size in the I/O model \cite{AggarwalV88} or 
ideal-cache \cite{Frigo:1999:CA:795665.796479} model. Namely, its search is cache-oblivious. 
The search cost is the
optimal and matches the search bound of B-trees that requires the memory block
size $B$ to be known in advance. Moreover, at any level of detail, each subtree
in the vEB layout is stored in a contiguous block of memory.

Although the static vEB layout is helpful for utilising data locality, it poorly
supports concurrent update operations. Inserting (or deleting) a node at
position $i$ in the contiguous block storing the tree may restructure a large
part of the tree. For example, inserting
new nodes in the full subtree $B_1$ (a leaf subtree) in Figure \ref{fig:vEB} will  affect the other
subtrees $B_2, B_3, \cdots, B_m$ by by rebalancing existing nodes between $B_1$ 
and the subtrees in order to have
space for new nodes. Even worse, we will need to allocate a new contiguous block of
memory for the whole tree if the previously allocated block of memory for
the tree runs out of space \cite{BrodalFJ02}. Note that we cannot use dynamic
node allocation via pointers since at {\em any} level of detail, each subtree in the vEB layout must be stored in a {\em
contiguous} block of memory.

\subsubsection{Relaxed cache-oblivious model and dynamic vEB layout} \label{sec:relaxed-veb}

\begin{figure}[!htbp]
\centering
\scalebox{0.5}{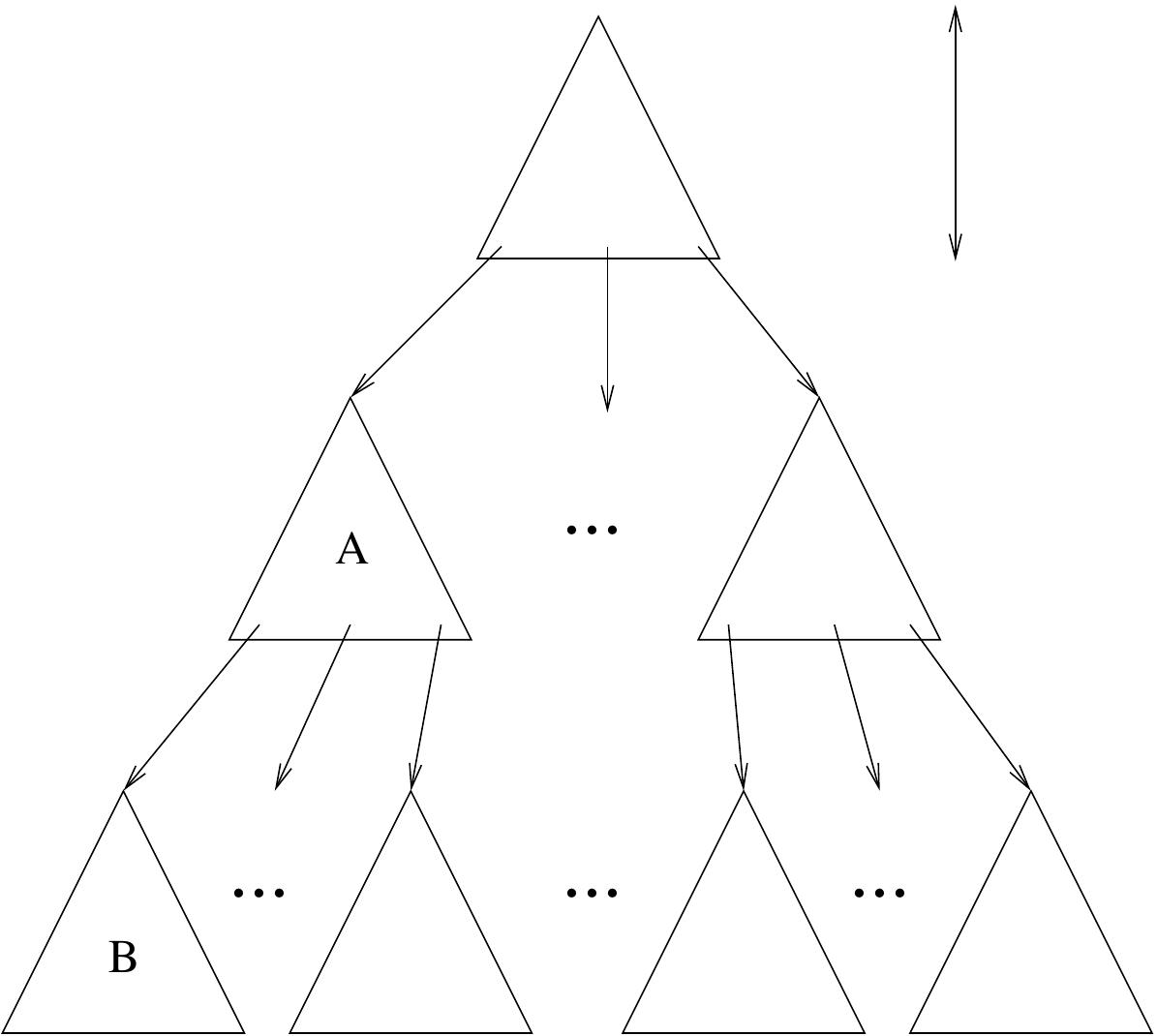}
\scalebox{1}{}
\scalebox{0.5}{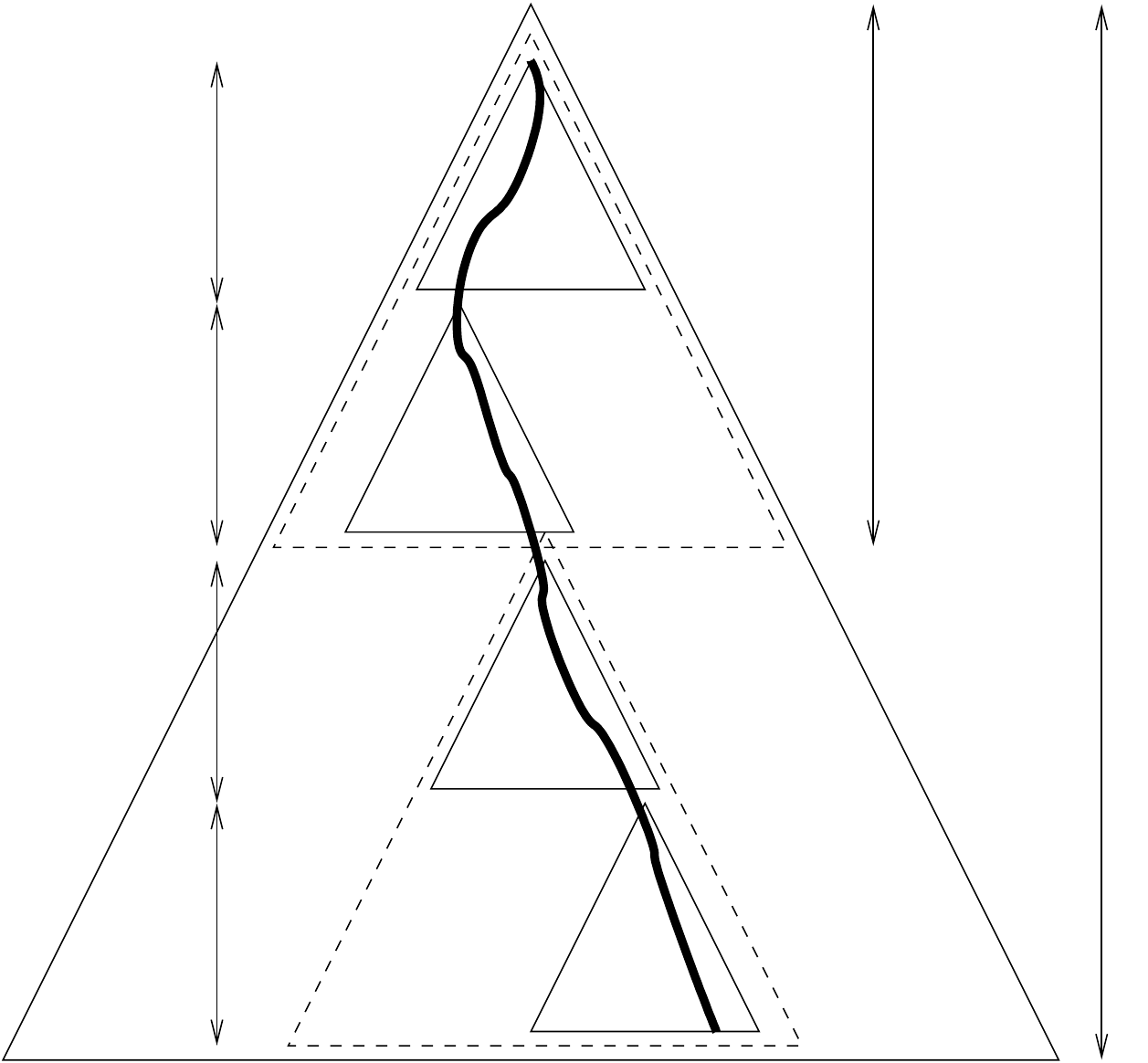}
\caption{{\em(left)}: New dynamic vEB layout. {\em(right)}: 
Search using dynamic vEB layout.}\label{fig:search_complexity}\label{fig:dynamicVEB}
\end{figure}

In order to make the vEB layout suitable for highly concurrent data structures
with update operations, we introduce a novel {\em dynamic} vEB layout. Our key
idea is that if we know an upper bound $UB$ on the unknown memory block size
$B$, we can support dynamic node allocation via pointers while maintaining the
optimal search cost of $O(\log_B N)$ memory transfers without knowing $B$ (cf.
Lemma \ref{lem:dynamic_vEB_search}). 

We define {\em relaxed cache oblivious} algorithms to be cache-oblivious (CO)
algorithms with the restriction that an upper bound $UB$ on the unknown memory
block size $B$ is known in advance. 
As long as an upper bound on all the block
sizes of multilevel memory is known, the new relaxed CO model maintains the key
feature of the original CO model \cite{Frigo:1999:CA:795665.796479}: 
First, temporal locality is exploited perfectly as there is no constraints on cache size
$M$ in the model. With this an optimal offline cache replacement policy can be
assumed. In practice, the Least Recently Used (LRU) policy with memory of size 
$(1+\epsilon)M$, where $\epsilon>0$, is nearly as good as the optimal replacement policy
with memory of size $M$ \cite{Sleator:1985:AEL:2786.2793};
Second, analysis for a simple two-level memory
are applicable for an unknown multilevel memory (e.g. registers, L1/L2/L3 caches
and memory). Namely, an algorithm that is optimal in terms of data movement for a 
simple two-level memory is asymptotically optimal for an unknown multilevel memory. 
This feature enable designing algorithms that can utilise fine-grained data locality 
in deep memory hierarchy of modern architectures. In
practice, although the exact block size at each level of the memory hierarchy is
architecture-dependent (e.g. register size, cache line size), obtaining a single upper
bound for all the block sizes (e.g. register size, cache line size and page size)
is easy. For example, a page size obtained from the operating system is such
an upper bound.
        
Figure \ref{fig:dynamicVEB} illustrates the new dynamic vEB layout based on the
relaxed cache oblivious model. Let $L$ be the coarsest level of detail such that
every recursive subtree contains at most $UB$  nodes. Namely, let $H$ and $S$ 
be the height and size of such a balanced subtree then $H=2^L$ and $S=2^H < UB$.
The tree is recursively
partitioned into level of detail $L$ where each subtree represented by a
triangle in Figure \ref{fig:dynamicVEB},  is stored in a contiguous memory block
of size $UB$. Unlike the conventional vEB, the subtrees at level of detail $L$
are linked to each other using pointers, namely each subtree at level of detail
$k > L$ is not stored in a contiguous block of memory.  Intuitively, since $UB$
is an upper bound on the unknown memory block size $B$, storing a subtree at
level of detail $k > L$ in a contiguous memory block of size greater than $UB$,
does not reduce the number of memory transfers, provided there is perfect alignment. 
For example, in Figure
\ref{fig:dynamicVEB}, a travel from a subtree $A$ at level of detail $L$, which
is stored in a contiguous memory block of size $UB$, to its child subtree $B$ at
the same level of detail will result in at least two memory transfers: one for
$A$ and one for $B$. Therefore, it is unnecessary to store both $A$ and $B$ in a
contiguous memory block of size $2UB$. As a result, the memory transfer cost of any search in
the new dynamic vEB layout is intuitively the same as that of the conventional
static vEB layout (cf. Lemma \ref{lem:dynamic_vEB_search}) while the dynamic vEB supports
highly concurrent update operations.

\begin{figure}[!htbp]\centering  \includegraphics[width=0.7\columnwidth]{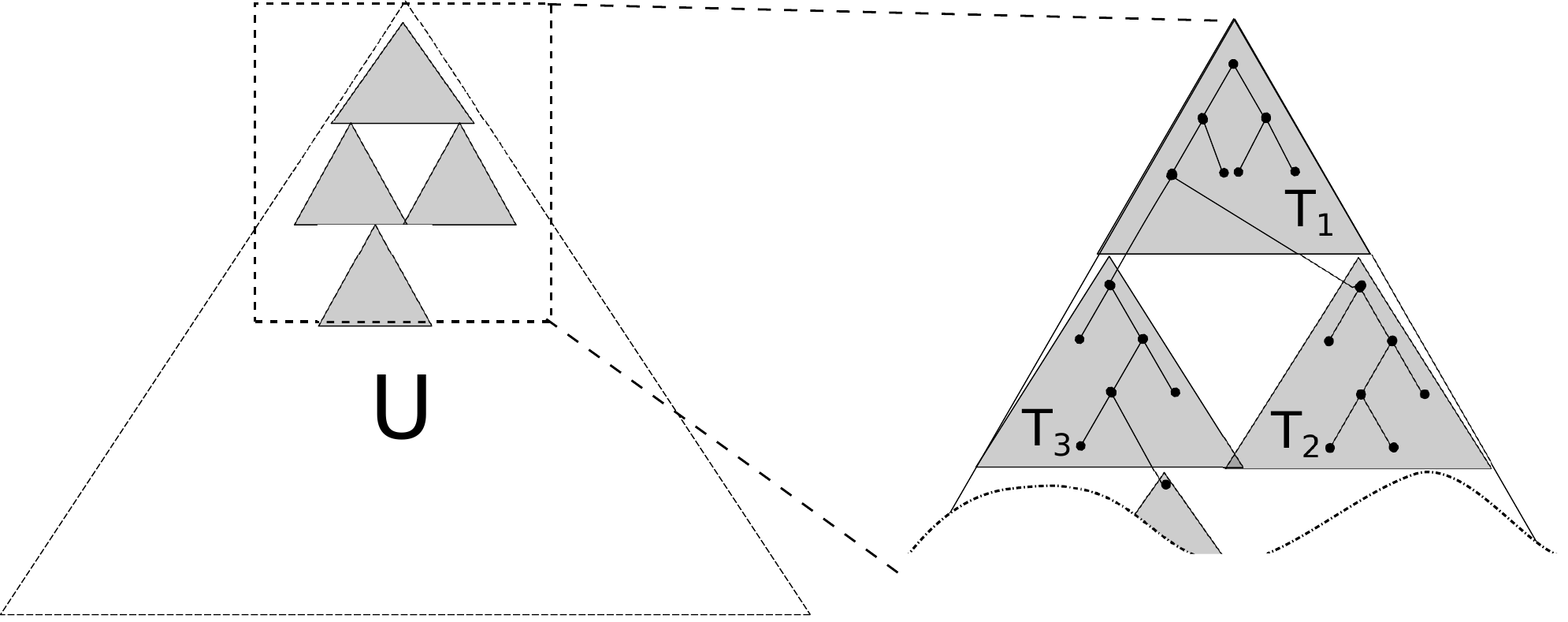}
\caption{Depiction of a DeltaNode $U$. Triangles $T_x$ represent the $\Delta$Nodes.}
\label{fig:treeuniverse}
\end{figure}

Let $\Delta$Node be a subtree at level of detail $L$, which is stored in a
contiguous memory block of size $UB$ and is represented by a triangle in Figure
\ref{fig:dynamicVEB}.

\begin{lemma}
A search in a complete binary tree with the new dynamic vEB layout needs 
$O(\log_B N)$ memory transfers, where $N$ and $B$ is the tree size and 
the {\em unknown} memory block size in the ideal cache model
\cite{Frigo:1999:CA:795665.796479}, respectively.
\label{lem:search_mem} \label{lem:dynamic_vEB_search}
\end{lemma}
\begin{proof} (Sketch)
Figure \ref{fig:search_complexity} illustrates the proof.  
Let $k$ be the coarsest level of detail such that every recursive subtree 
contains at most $B$ nodes. Since $B \leq UB$, $k \leq L$, where $L$ is 
the coarsest level of detail at which every recursive subtree 
contains at most $UB$ nodes. That means there are at most $2^{L-k}$ subtrees
along the search path in a $\Delta$Node and no subtree of depth $2^k$ is split 
due to the boundary of $\Delta$Nodes. Namely, triangles of height $2^k$ fit
within a dashed triangle of height $2^L$ in Figure \ref{fig:search_complexity}. 

Due to the property of the new dynamic vEB layout that at any level of detail 
$i \leq L$, a recursive subtree of depth $2^i$ is stored in a contiguous block 
of memory, each subtree of depth $2^k$ {\em within} a $\Delta$Node is stored in
at most 2 memory blocks of size $B$ (depending on the starting location of
the subtree in memory). Since every subtree of depth $2^k$ fits in a
$\Delta$Node (i.e.
no subtree is stored across two $\Delta$Nodes), every subtree of depth $2^k$ is 
stored in at most 2 memory blocks of size $B$.

Since the tree has height $T$, $\lceil T / 2^k \rceil$ subtrees of depth $2^k$ 
are traversed in a search and thereby at most  $2  \lceil T / 2^k \rceil$ memory 
blocks are transferred. 

Since a subtree of height $2^{k+1}$ contains more than $B$ nodes, 
$2^{k+1} \geq \log_2 (B + 1)$, or $2^{k} \geq \frac{1}{2} \log_2 (B+ 1)$. 

We have $2^{T-1} \leq N \leq 2^T$ since the tree is a {\em complete} binary tree. 
This implies $ \log_2 N \leq T \leq \log_2 N +1$.  

Therefore, the number of memory blocks transferred in a search is 
$2  \lceil T / 2^k \rceil \leq 4 \lceil \frac{\log_2 N + 1}{\log_2 (B + 1)} \rceil 
= 4 \lceil \log_{B+1} N + \log_{B+1} 2\rceil$ $= O(\log_B N)$, where $N \geq 2$.
\end{proof}

\subsection{$\Delta$Tree implementation}\label{sec:implementation}

Figure \ref{fig:treeuniverse} illustrates a $\Delta$Tree $U$. 
$U$ uses the dynamic vEB layout presented in Section \ref{sec:relaxed-veb}.
The $\Delta$Tree consists of $|U|$ $\Delta$Nodes of fixed size $UB$ each of
which contains a \textit{leaf-oriented} binary search tree (BST) $T_i, i=1,
\dots,|U|$. $\Delta$Node's internal nodes are put together in cache-oblivious
fashion using the vEB layout.

\begin{figure}[!htbp]
\centering 
\begin{algorithmic}[1] 
\Start{\textbf{node} $n$} \label{lst:line:nodestruct}
	\START{member fields:}
\State $tid \in \mathbb{N}$, if $> 0$ indicates the node is \textit{root} of a \par
	\hskip\algorithmicindent $\Delta$Node with an id of $tid$ ($T_{tid}$) 
\State $value \in
	\mathbb{N}$, the node value, default is \textbf{empty} 
\State $mark \in \{true,false\}$, a value of \textbf{true} indicates a logically \par 
	\hskip\algorithmicindent deleted node
\State $left, right \in \mathbb{N}$, left / right child pointers 
\State $isleaf \in {true,false}$, indicates whether the \par 
	\hskip\algorithmicindent node is a leaf of a $\Delta$Node, default is \textbf{true} 	\label{lst:line:leafdefault}
\END \End

\Statex
\Start{\textbf{triangle} $S$} \label{lst:line:trianglestruct}
\START{member fields:}
\State $\mathit{nodes}$, a group of pre-allocated node $n$ $\{n_1,n_2,\ldots,n_{UB}\}$
\State $\mathit{buffer}$, an array of value with a length \par
\hskip\algorithmicindent  of the current number of threads 
\END \End

\Statex
\Start{\textbf{$\Delta$Node} $T$} \label{lst:line:deltanodestruct}
\START{member fields:}
\State $\mathit{locked}$, indicates whether a $\Delta$Node is locked 
\State $\mathit{opcount}$, a counter for the active update operations 
\State $\mathit{root}$, pointer the root node of the $\Delta$Node ($S_x.n_1$)
\State $\mathit{rootbuffer}$, pointer the buffer of the $\Delta$Node ($S_x.\mathit{buffer}$)
\State $\mathit{mirror}$, pointer to root node of the $\Delta$Node's mirror \par 
	\hskip\algorithmicindent ($S_{x'}.n_1$)
\State $\mathit{mirrorbuffer}$, pointer to buffer of the $\Delta$Node's mirror \par 
	\hskip\algorithmicindent ($S_{x'}.\mathit{buffer}$)
\END \End

\Statex
\Start{\textbf{universe} $U$} 		\label{lst:line:universe}
	\START{member fields:}
\State $root$, pointer to the $root$ of the topmost $\Delta$Node ($T_1.root$) 
\END \End
\end{algorithmic}
\caption{$\Delta$Tree's data structures.} \label{lst:datastruct}
\end{figure}

The $\Delta$Tree $U$ acts as a dictionary of abstract data types. It maintains
a set of values which are members of an ordered universe \cite{EllenFRB10}. 
The $\Delta$Tree $U$ provides the following operations: \textsc{insertNode($v, U$)}, 
which adds value $v$ to the set $U$, \textsc{deleteNode($v, U$)} 
for removing a value $v$ from the set, and
\textsc{searchNode($v, U$)}, which determines whether value $v$ exists in the
set. We use the term \textit{update} operation for either insert or delete
operation. We assume that duplicate values are not allowed inside the set and a
special value, for example $0$, is reserved as an indicator of an \textsc{Empty}
value.

\subsubsection{Data structures}

The implementation of $\Delta$Tree utilises the data structures defined in
Figure \ref{lst:datastruct}. The topmost level of $\Delta$Tree is represented by
a struct \textsc{universe} (line \ref{lst:line:universe}) that only contains a pointer to 
the root of the topmost $\Delta$Node.

Each $\Delta$Node that forms the $\Delta$Tree is represented by the struct
\textsc{$\Delta$Node} (line \ref{lst:line:deltanodestruct}). Each $\Delta$Node has an associated mirror. 
This structure consists of a field \textbf{opcount}, which is a counter that indicates 
how many insert/delete threads that
are currently operating within that $\Delta$Node; field \textbf{locked} that indicates 
whether a $\Delta$Node is currently locked by maintenance operations: 
when it is set to \textit{true}, no insert/delete threads are allowed to get in. The \textbf{root} pointer serves as the root
of a $\Delta$Node, while the pointer \textbf{mirror} references root of the $\Delta$Node's mirror. 
Also there is \textbf{rootbuffer} and \textbf{mirrorbuffer} pointers
that reference the $\Delta$Node's buffer and the mirror's buffer, respectively. 

Each \textsc{node} structure (line \ref{lst:line:nodestruct}) contains field \textbf{value}, 
which holds a value that will guide the search or a data value at a leaf-node. 
Field \textbf{mark} is used to indicate whether a value
is logically deleted. A \textit{true} value of \textbf{isleaf} indicates a leaf
node (as in the leaf-oriented tree), and \textit{false} otherwise. Field \textbf{tid} is a unique
identifier of a corresponding $\Delta$Node and it is used to let a thread know whether it has moved
between $\Delta$Nodes. 

\subsubsection{Function specifications}

\begin{figure}[!htbp]
\centering \begin{algorithmic}[1] 
\Function{searchNode}{$v, U$}
\State $lastnode, p \gets U.root$
\While{$p \neq$ end of tree $\And p.isleaf \neq$ TRUE}	\label{lst:line:searchifleaf}
    \State $lastnode \gets p$	\label{lst:line:lasnode-p}
        \If{$p.value < v$}		\label{lst:line:searchless}
            \State $p \gets p.left$
        \Else					\label{lst:line:searchelse}
            \State $p \gets p.right$  \label{lst:line:search-end}
        \EndIf
\EndWhile
\If{$lastnode.value = v$} 		\label{lst:line:linsearch3}
	\If{$lastnode.mark =$ FALSE}	 \Comment{lastnode is not deleted}	\label{lst:line:linsearch1}
		\State\Return TRUE 
	\Else 
		\State\Return FALSE 
	\EndIf 
\Else \State Search (last visited $\Delta$Node's \textit{rootbuffer}) for $v$	\label{lst:line:searchbuffer}
	\If{$v$ is found}							\label{lst:line:linsearch2}
		\State\Return TRUE 
	\Else 
		\State\Return FALSE 
	\EndIf 
\EndIf 
\EndFunction
\end{algorithmic}
\caption{$\Delta$Tree's wait-free search algorithm.}\label{lst:nodeSearch}
\end{figure}

\begin{figure}[!th]
\centering
\scriptsize
\begin{minipage}[t]{0.47\linewidth}
\begin{algorithmic}[1]

\Function{insertNode}{$v,U$} 
\State $t \gets U.root$
\State \Return \textsc{insertHelper}($v,t$)
\EndFunction

\Function{deleteNode}{$v,U$} 
\State $t \gets U.root$
\State \Return \textsc{deleteHelper}($v,t$)
\EndFunction

\Function{deleteHelper}{$v, node$}	\label{lst:line:deletefunc}					
\State $success \gets$ TRUE

\If{Entering new $\Delta$Node $T_x$}                 			                                                 
   \State  $T'_x \gets$ \textsc{getParent$\Delta$Node}($T_x$)
   \State  \textsc{decrement}($T'_x.opcount$) 				
   \State  \textsc{waitandcheck}($T_x.locked$, $T_x.opcount$)
\EndIf

\If{($node.isleaf =$ TRUE) \textbf{OR} ($!node.left \And !node.right$)}
        \If{$node.value = v$}
            \If{\textsc{CAS}($node.mark$, FALSE, TRUE) != FALSE)}              	\label{lst:line:markdel}       
                	\State $success \gets$ FALSE                                           \Comment{already deleted!}
		\State \textsc{decrement}($T_x.opcount$)
            \Else	
             	\If{($node.left.value$\&$node.right.value$=\textbf{empty})} 	\label{lst:line:markdel-check}
			\State \textsc{decrement}($T_x.opcount$) 
                		\State \textsc{mergeNode}(\textsc{parentOf}($T_x)) \gets$ TRUE          \label{lst:line:merge}                 
		\Else		 
			\State $success \gets$\textsc{deleteHelper}($v, node$)	 	\Comment{re-try} \label{lst:line:retrydel}
		\EndIf
            \EndIf
        \Else
        		\State Search ($T_{x}.\mathit{rootbuffer}$) for $v$
	   	\If{$v$ is found in $T_{x}.\mathit{rootbuffer.idx}$}
                 	\State $T_{x}.\mathit{rootbuffer.idx} \gets \textbf{empty}$	\label{lst:line:bufdel} 	\Comment{buffered delete}
            	\Else	
        	   		\State $success \gets$ FALSE
	    	\EndIf
		\State \textsc{decrement}($T_x.opcount$)                                             
    	\EndIf
\Else \If{$v < node.value$}
        \State $success \gets$\textsc{deleteHelper}($v, node.left$)		\Comment{go left}
\Else
        \State $success \gets$\textsc{deleteHelper}($v, node.right$)		\Comment{go right}
\EndIf	
\EndIf
\State \Return $success$
\EndFunction

\Function{insertHelper}{$v, node$}	\label{lst:line:insertfunc}
\State $success \gets$ TRUE
\If{Entering new $\Delta$Node $T_x$}                                            	                      
   \State  $T'_x \gets$ \textsc{getParent$\Delta$Node}($T_x$)
   \State  \textsc{decrement}($T'_x.opcount$) 				
   \State  \textsc{waitandcheck}($T_x.locked$, $T_x.opcount$)
\EndIf

\algstore{bkbreak}
\end{algorithmic}
\end{minipage}
\qquad
\begin{minipage}[t]{0.47\linewidth}
\begin{algorithmic}[1]
\algrestore{bkbreak}
    
\If{$node.left \And node.right$}                                                    
        \If{$v < node.value$}                                                                              
            \If{($node.isleaf =$ TRUE)}  									\Comment{insert to the left:}                                
                \If{CAS($node.left.value$, \textbf{empty}, $v$) =
                \textbf{empty}}       \label{lst:line:ins-lp1} 	\label{lst:line:growins1} 
                	\State $node.right.value \gets node.value$   	                                                   
                    \State $node.right.mark \gets node.mark$
                    \State $node.isleaf \gets$ FALSE		\label{lst:line:growins1-end} 
                    \State  \textsc{decrement}($T_x.opcount$)
                \Else
                    \State $success \gets$\textsc{insertHelper}($v$, $node$)    					\Comment{re-try} \label{lst:line:retryinsleft}                    
                \EndIf
            \Else
                \State $success \gets$\textsc{insertHelper}($v$, $node.left$)                       		\Comment{go left}
            \EndIf
        \ElsIf {$v > node.value$}                                                                      
            \If{($node.isleaf =$ TRUE)}   \Comment{insert to the right:}                                
                \If{CAS($node.left.value$, $\textbf{empty}$, $node.value$) = \par 
			\hskip4em \textbf{empty}}  \label{lst:line:ins-lp2}  \label{lst:line:growins2} 
                	\State $node.right.value \gets v$                                                       
                    \State $node.left.mark \gets node.mark$
                    \State \textsc{atomic} $\{$				\label{lst:line:atomic-isleaf}
                    	$node.value \gets v$
                    	\State $node.isleaf \gets$ FALSE	\label{lst:line:growins2-end} 
                    $\}$	
                    \State  \textsc{decrement}($T_x.opcount$)
                \Else
                    \State $success \gets$\textsc{insertHelper}($v$, $node$)        		\Comment{re-try}      \label{lst:line:retryinsright}             
                \EndIf
            \Else
                \State $success \gets$\textsc{insertHelper}($v$, $node.right$)                	\Comment{go right}
            \EndIf
        \ElsIf {$v = node.value$}                                                                                            
            \If{($node.isleaf =$ TRUE)}                                   
                \If{$node.mark =$ FALSE}							\Comment{is deleted?}		
                    \State $success \gets$ FALSE            				\Comment{value exist!} \label{lst:line:valexist}
                    \State  \textsc{decrement}($T_x.opcount$)                                   
                \Else
                	   \State Goto \ref{lst:line:growins2} 					\Comment{goto insert right}
                \EndIf
            \Else
                \State $success \gets$\textsc{insertHelper}($v, node.right$)                  	\Comment{go right}     
            \EndIf
        \EndIf
\Else			
		\If{$v$ already in $T_x.\mathit{rootbuffer}$}
			\State $success \gets FALSE$   
			\State \textsc{decrement}($T_x.opcount$)
		\Else
			\qquad \textit{put $v$ inside $T_x.\mathit{rootbuffer}$} \label{lst:line:insertbuffer}		\Comment{buffered insert}	
			\If{\textsc{TAS}($T_x.locked$)} \Comment{Acquired maintenance lock} \label{lst:line:lockbuf}
			\State \textsc{decrement}($T_x.opcount$)	\label{lst:line:dec-opcount}
			\State \textsc{spinwait}($T_x.opcount$)  \Comment{Waits updates to finish} \label{lst:line:spinwait-main}
			\State \ldots \textit{do \textsc{rebalance}($T_x$) \textbf{or} \textsc{expand($node$)}} \ldots  \label{lst:line:expand}
			\EndIf 
		\EndIf
\EndIf
\State \Return $success$
\EndFunction
\end{algorithmic}
\end{minipage}
\caption{$\Delta$Tree's update algorithms and their helper functions.}
\label{lst:pseudo-ops}
\end{figure}

\clearpage

Operation \textsc{searchNode($v,U$)} (cf. Figure \ref{lst:nodeSearch}) is going to walk 
over the $\Delta$Tree (Figure \ref{lst:nodeSearch}, lines \ref{lst:line:searchifleaf}--\ref{lst:line:search-end}) to
find whether the value $v$ exists in $U$. 
This particular operation is
guaranteed to be wait-free (cf. Lemma \ref{lem:waitfreesearch}), 
and it returns \textbf{true} whenever $v$ has been
found, or \textbf{false} otherwise (Figure \ref{lst:nodeSearch}, line \ref{lst:line:linsearch1}). 
Operation \textsc{insertNode($v, U$)} (cf. Figure \ref{lst:pseudo-ops}, line \ref{lst:line:insertfunc}) 
inserts value $v$ at a leaf of $\Delta$Tree, 
provided $v$ does not yet exist in the
tree (Figure \ref{lst:pseudo-ops}, line \ref{lst:line:valexist}). 
Following the nature of a leaf-oriented tree, a successful insert operation 
will replace a leaf with a subtree of three nodes \cite{EllenFRB10} (cf. Figure
\ref{fig:treetransform}a and pseudocode in Figure \ref{lst:pseudo-ops}, 
line \ref{lst:line:growins1} \& \ref{lst:line:growins2}).
\textsc{deleteNode($v, U$)} (cf. Figure \ref{lst:pseudo-ops}, line \ref{lst:line:deletefunc}) 
simply {\em mark} the leaf that contains value $v$ as deleted (Figure \ref{lst:pseudo-ops}, 
line \ref{lst:line:markdel}), instead of physically removing the leaf or changing
its parent pointer (as in \cite{EllenFRB10}). 

\begin{lemma} \label{lem:waitfreesearch}
$\Delta$Tree search operation is wait-free.
\end{lemma}
\begin{proof}(Sketch) The proof can be served based on these observations on Figure \ref{lst:nodeSearch}:
\begin{enumerate}
\item \textsc{SearchNode} and invoked \textsc{SearchBuffer} (line \ref{lst:line:searchbuffer}) don't wait for any locks.
\item The number of iterations in the \textit{while} loop (line 3) is bounded by the \textit{height} of the tree, $\mathcal{O}(N)$.
\item \textsc{SearchBuffer} time complexity is bounded by the buffer size, which is a constant.
\end{enumerate}
Therefore the  \textsc{SearchNode} time is bounded by $\mathcal{O}(N)$.
\end{proof}

There is a difference between inserting to the left
(Figure \ref{lst:pseudo-ops}, lines \ref{lst:line:growins1}--\ref{lst:line:growins1-end}) 
and inserting to the right (Figure \ref{lst:pseudo-ops}, lines \ref{lst:line:growins2}--\ref{lst:line:growins2-end}) 
because an insert to the right will need to change the value of the root of 
the new subtree in order to guide the tree search. And it's not necessary to 
do that when inserting a value to the left.

\paragraph{Maintenance functions}
Apart from the basic operations, three maintenance $\Delta$Tree operations are
invoked in certain cases of inserting and deleting a node from the tree.
Operation \textsc{rebalance($T_v.root$)} (cf.  Figure \ref{lst:pseudo-ops} 
line \ref{lst:line:expand}) is responsible for rebalancing a
$\Delta$Node after an insertion.
Figure \ref{fig:treetransform}a illustrates the rebalance operation. Whenever a
new node $v$ is to be inserted at the last level $H$ of $\Delta$Node $T_1$, the
$\Delta$Node is rebalanced to a complete BST by setting a new depth for all
leaves (e.g. $a,v,x,z$ in Figure \ref{fig:treetransform}a) to $\log N + 1$,
where $N$ is the number of leaves. In Figure \ref{fig:treetransform}a, right
after the rebalance operation, tree $T_1$ becomes balanced and its
height reduces from 4 to 3.

We also define \textsc{expand($l$)} operation (cf.  Figure \ref{lst:pseudo-ops} 
line \ref{lst:line:expand}), that will be responsible for
creating a new $\Delta$Node and connecting it to the parent of a leaf node $l$ 
(cf. Figure \ref{fig:treetransform}b).
Expanding will be triggered only if after \textsc{insertNode($v, U$)}, leaf $l$
will be at the last level of a $\Delta$Node and rebalancing will no longer
reduce the current height of the subtree $T_i$ stored in the $\Delta$Node. 
For example if expanding is taking place at a node $l$ located at the bottom
level of the $\Delta$Node (Figure \ref{fig:treetransform}b, node $l$ contains value $v$), or $depth(l) = H$,
then a new $\Delta$Node ($T_2$ for example) will be referred by the parent of
node $l$.
Namely, the parent of $l$ swaps one of its pointer that previously points to
$l$, into the root of the newly created $\Delta$Node, $T_2.root$.

\begin{figure}[!htbp]
\centering 
\includegraphics[width=0.7\columnwidth]{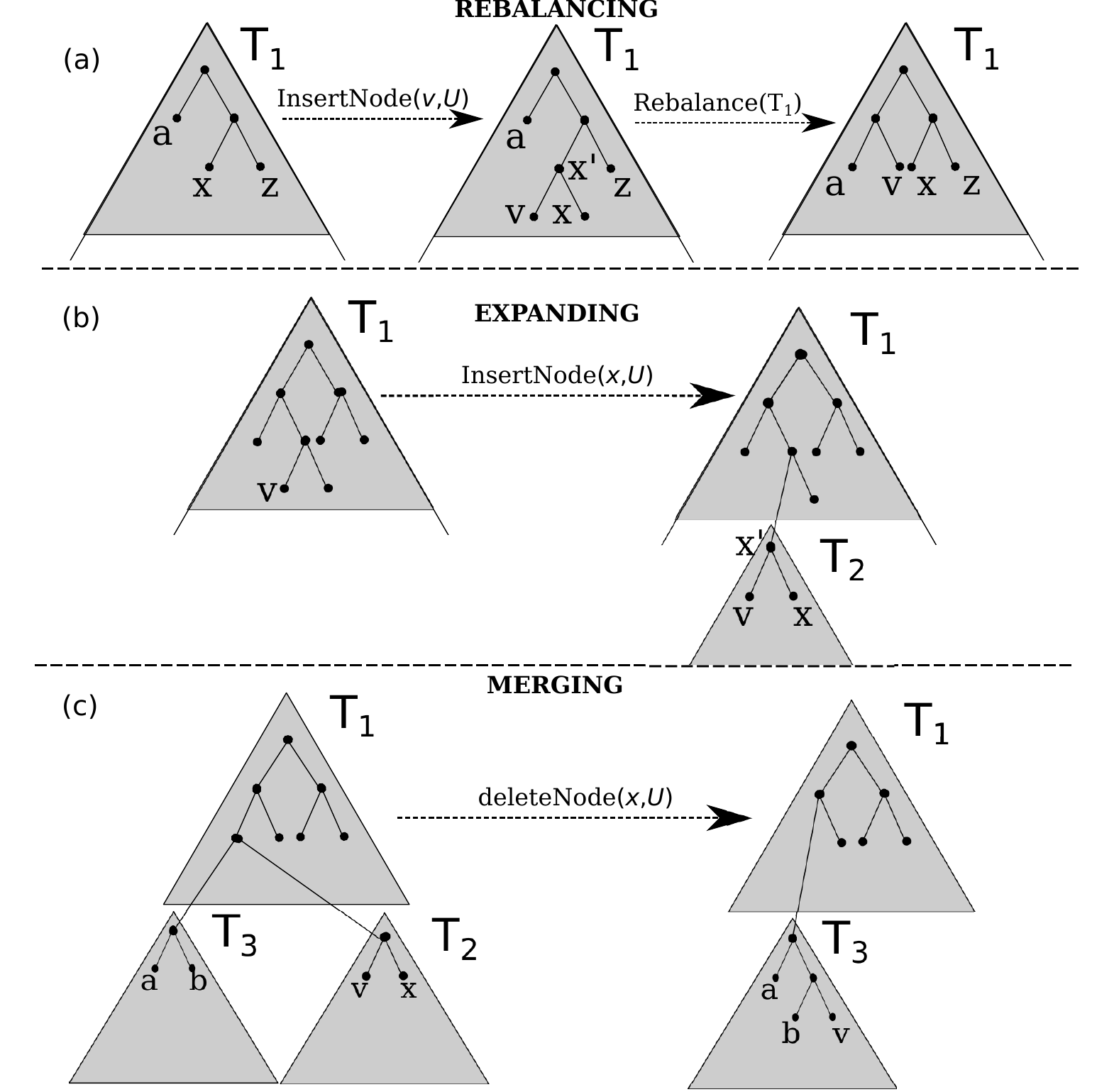}
\caption{(a)\textit{Rebalancing}, (b)\textit{Expanding}, and (c)\textit{Merging} operations on $\Delta$Tree.}
\label{fig:treetransform}
\end{figure}

Operation \textsc{merge($T_x.root$)} (cf.  Figure \ref{lst:pseudo-ops} 
line \ref{lst:line:merge}) is for merging $T_x$ with its sibling 
after a node deletion. For example in Figure \ref{fig:treetransform}c, $T_2$ is
merged into $T_3$. Then
the pointer of $T_3$'s grandparent that previously points
to the parent of both $T_3$ and $T_2$ is replaced by a pointer to $T_3$.
Merge operation is invoked provided that a particular $\Delta$Node, where 
a deletion has taken place, is filled less than $2^t$ of its capacity 
(where $t = 1/2H$) and all
values of that $\Delta$Node and its siblings could be fitted into a $\Delta$Node. 

To minimise block transfers required during tree traversal, the height of the
tree should be kept minimal. These auxiliary operations are the unique
feature of $\Delta$Tree in the effort of maintaining a small height.

These \textsc{insertNode} and \textsc{deleteNode} operations are linearisable 
to other  \textsc{searchNode, insertNode} and \textsc{deleteNode} operations
(cf. Lemma \ref{lem:linear-search} and \ref{lem:nonblock-update}).
Both of the operations are using single word CAS (Compare and Swap)
and "leaf-checking" (cf. Figure \ref{lst:pseudo-ops}, line \ref{lst:line:markdel} \& 
\ref{lst:line:markdel-check} for delete and \ref{lst:line:ins-lp1}
\& \ref{lst:line:ins-lp2} for insert) to achieve that. 

\begin{lemma} 
For a value that resides on the leaf node of a $\Delta$Node,
\textsc{searchNode} operation (Figure \ref{lst:nodeSearch}) has the linearisation 
point to \textsc{deleteNode} at line \ref{lst:line:linsearch1} and the
linearisation point to \textsc{insertNode} at line \ref{lst:line:linsearch3}. 
For a value that stays in the buffer of a $\Delta$Node, \textsc{searchNode} 
operation has the \textit{linearisation point} at line \ref{lst:line:linsearch2}.
\label{lem:linear-search}
\end{lemma}
\begin{proof}(Sketch) It is trivial to demonstrate this in relation to
deletion algorithm in Figure \ref{lst:pseudo-ops} since only an atomic operation
is responsible for altering the $mark$ property of a node (line
\ref{lst:line:markdel}).
Therefore \textsc{deleteNode} has the linearisation point to \textsc{searchNode} 
at line \ref{lst:line:markdel}. 
 
For \textsc{searchNode} interaction with an insertion that grows new subtree, we
rely on the facts that: 

\begin{enumerate}

\item a snapshot of the current node $p$ is recorded on
$lastnode$ as a first step of searching iteration (Figure \ref{lst:nodeSearch},
line \ref{lst:line:lasnode-p});

\item $v.value$ change, if needed, is not done
until the last step of the insertion routine for insertion of $v > node.value$
and will be done in one atomic step with $node.isleaf$ change (Figure
\ref{lst:pseudo-ops}, line \ref{lst:line:atomic-isleaf}); and 

\item $isleaf$
property of all internal nodes are by default \textbf{true} (Figure
\ref{lst:datastruct}, line \ref{lst:line:leafdefault}) to guarantee that values
that are inserted are always found, even when the leaf-growing (both
left-and-right) are happening concurrently. 
\end{enumerate}

Therefore
\textsc{insertNode} has the linearisation point to \textsc{searchNode} at line 
\ref{lst:line:ins-lp1} when inserting a value $v$ smaller than the leaf node's 
value, or at line \ref{lst:line:ins-lp2} otherwise. 

A search procedure is also able to cope well with a "buffered" insert, that is
if an insert thread loses a competition in locking a $\Delta$Node for expanding
or rebalancing and had to dump its carried value inside a buffer (Figure
\ref{lst:pseudo-ops}, line \ref{lst:line:insertbuffer}). Any value inserted to
the buffer is guaranteed to be found. This is because after a leaf $lastnode$
has been located, the search is going to evaluate whether the $lastnode.value$
is equal to $v$. Failed comparison will cause the search to look further inside
a buffer ($T_x.\mathit{rootbuffer}$) located in a $\Delta$Node where the leaf resides
(Figure \ref{lst:nodeSearch}, line \ref{lst:line:searchbuffer}).
By assuring  that the switching of a root $\Delta$Node with its mirror includes
switching $T_x.rootbuffer$ with $T_x.\mathit{mirrorbuffer}$, we can show that any new
values that might be placed inside a buffer are guaranteed to be found
immediately after the completion of their respective insert procedures.
The "buffered" insert has the linearisation point to \textsc{searchNode} at
line \ref{lst:line:insertbuffer}.

Similarly, deleting a value from a buffer is as trivial as exchanging that value
inside a buffer with an \textbf{empty} value. The search operation will failed
to find that value when doing searching inside a buffer of $\Delta$Node. This
type of delete has the linearisation point to \textsc{searchNode} at the same
line it's emptying a value inside the buffer (line \ref{lst:line:bufdel}).
\end{proof}

\begin{lemma} \label{lem:nonblock-update}
In the absence of maintenance operations, the linearisation point of $\Delta$Tree's 
\textsc{insertNode} to \textsc{deleteNode} is in line \ref{lst:line:growins1} 
and \ref{lst:line:growins2} of Figure \ref{lst:pseudo-ops}. 
Linearisation points of \textsc{deleteNode} operations to \textsc{insertNode} are in line \ref{lst:line:markdel}
of Figure \ref{lst:pseudo-ops}.
\end{lemma}
\begin{proof} (Sketch)
In a case of concurrent insert operations (Figure \ref{lst:pseudo-ops}) 
at the
same leaf node $x$, insert threads that need to "grow" the node 
(for illustration, cf. Figure \ref{fig:treetransform}) are
going to do \textsc{CAS}($x.left, \textbf{empty},\ldots$) 
(line  \ref{lst:line:growins1} and \ref{lst:line:growins2}) as their
first step. This CAS is the only thing needed since the whole $\Delta$Node
structure is pre-allocated and the CAS is an atomic operation. Therefore, only
one thread will succeed in changing $x.left$ and proceed populating the
$x.right$ node. Other threads will fail the CAS operation and they are going to
try restart the insert procedure all over again, starting from the node $x$.

To assure that the marking delete (line \ref{lst:line:markdel}) behaves nicely
with the "grow" insert operations, \textsc{deleteNode}($v, U$) that has found
the leaf node $x$ with a value equal to $v$, will need to check again whether
the node is still a leaf (line \ref{lst:line:markdel-check}) after completing
\textsc{CAS}($x.mark, FALSE, TRUE$).
The thread needs to restart the delete process from $x$ if it has found that $x$
is no longer a leaf node.
\end{proof}

\paragraph{Mirroring} \label{sec:mirroring}
Whenever a $\Delta$Node is undergoing a
maintenance operation (balancing, expanding, or merging), a mirroring operation
also takes place. Mirroring works by maintaining the original nodes but write 
the results into the
mirror nodes. After this is done, the pointer will be switched and 
now the mirror nodes become the original, vice-versa. 
The $\Delta$Node's lock will be released, 
and all the waiting update threads can continue with their respective operation.
The original nodes and helping buffer served as the latest snapshot, which
enables wait-free search on that $\Delta$Node. 

Despite \textsc{insertNode} and \textsc{deleteNode} are non-blocking, they'll still
need to wait at a tip of a $\Delta$Node whenever a maintenance
operation is currently operating within
that $\Delta$Node. We employ TAS (Test and Set) on $\Delta$Node's \textit{locked} 
field (cf. Figure \ref{lst:pseudo-ops}, line \ref{lst:line:lockbuf}) 
before any maintenance operation starts. Advanced locking techniques \cite{HaPT07_JSS, KarlinLMO91, LimA94} can also be used. This is to ensure that
no basic update operations will interfere with the maintenance
operation. This is also necessary to prevent the buffer overflow.

\paragraph{Performance concerns}
Note that the previous description has shown that \textsc{rebalance}
and \textsc{merge} execution are actually sequential within a $\Delta$Node.
Rebalancing
and merging only involve a maximum of two $\Delta$Node with size $UB$. 
Their operation consist
of traversing and re-inserting all members of one or two $\Delta$Nodes.
Because $UB \ll N$, \textsc{rebalance} and \textsc{merge} operations are not
affecting much on the $\Delta$Tree performance.   

\subsection{Balanced $\Delta$Tree}\label{sec:plDTv1}

$\Delta$Tree implementation served as an initial proof of concept
of a dynamic VEB-based search trees. This tree however has major weaknesses that 
would affect its performance. The first is the fact that the 
\textit{left} and \textit{right} pointers occupies too much space. In 
a $\Delta$Node with 127 nodes, the sets of pointers will occupy 2032 bytes of memory in a 64-bit operating
system, twice of the \textit{key} nodes (in integer) that only require 1016 bytes. 
A cache oblivious data structures will gain most of its benefits if only more data could occupy
a small amount of memory space. Thus, a single cache-line transfer will relocate more data 
between any levels of memory. Secondly, inserting a consecutive increasing or decreasing 
numbers into $\Delta$Tree will results a linked-list of $\Delta$Node.
It will be hard to guarantee the optimal search performance
of Lemma \ref{lem:dynamic_vEB_search} in this particular case.

\begin{figure}[!htbp]
\centering 
\begin{algorithmic}[1] 
\Start{\textbf{Map}} \label{lst:line:map}
	\START{member fields:}
\State $left \in
	\mathbb{N}$, interval of the \textit{left} child pointer address
\State $right \in
	\mathbb{N}$, interval of the \textit{right} child pointer address 
\END \End

\Statex 

\Function{right}{p, base}  \label{lst:line:right}
    \State $nodesize \gets  \textsc{sizeOf}(\textbf{single node})$
    \State $idx \gets (p - base)/nodesize$
    
    \If {($map[idx].right$ != 0)}
        \State \Return $base + map[idx].right$
    \Else
        \State \Return 0
    \EndIf
\EndFunction

\Statex 

\Function{left}{p, base} \label{lst:line:left}
    \State $nodesize \gets  \textsc{sizeOf}(\textbf{single node})$
    \State $idx \gets (p - base)/nodesize$
    
    \If {($map[idx].left$ != 0)}
        \State \Return $base + map[idx].left$
    \Else
        \State \Return 0
    \EndIf
\EndFunction
\end{algorithmic}
\caption{\textit{Mapping} functions.}
\label{lst:map-func}
\end{figure}

\subsubsection{\textit{Map} instead of pointers} \label{sec:mapdesc}

We have developed an improved $\Delta$Tree, namely the balanced DT 
by completely eliminating (\textit{left} and \textit{right}) 
pointers within a $\Delta$Node. We replaced them with
\textsc{left} and \textsc{right} functions instead (Figure \ref{lst:map-func}, lines \ref{lst:line:left} 
\& \ref{lst:line:right}). These two functions, given an arbitrary node and
its container $\Delta$Node \textit{root} memory address, will return the left and right child node address
of that arbitrary node, respectively. A $\Delta$Node is now slim-lined into just an array of keys.
Each $\Delta$Node is also coupled with a metadata that contains
an array of pointers for the inter-$\Delta$Node connection, and a structure that holds 
lock and counters.

With this mapping, we need only a single $\Delta$Tree's pointer-based $\Delta$Node to be created
in the initialisation phase.
This $\Delta$Node is used to populate the \textit{map}, by calculating the memory address differences between a node and
its left and right children, respectively. 
This \textit{map} is then used for every Balanced DT $\Delta$Node's \textsc{left} and \textsc{right} operations.
The memory for the pointer-based $\Delta$Node can be freed after a $map$ is created. 
And since we are re-using one $map$ array of size
$UB$ for traversing, memory footprint for the Balanced DT's 
$\Delta$Node operations can be kept minimum. We ended up having 200\% more
node counts in a $\Delta$Node given the same $UB$, compared to the $\Delta$Tree.

\begin{figure}[!htbp]
\centering 
\begin{algorithmic}[1] 

\Function{pointerLessSearch}{key, $\Delta$Node, maxDepth}
    
  \While{$\Delta$Node is not leaf}
   	\State $bits \gets 0$
    	\State $depth \gets 0$
   	\State $p \gets {\Delta}Node.root$;
   	\State $base \gets p$
    	\State $link \gets {\Delta}Node.link$
    
    	\While{$p \And p.value$ != EMPTY}  \Comment{continue until leaf node} 
        		\State $depth \gets depth + 1$ \Comment{increment depth}
        		\State $bits \gets bits << 1$ \Comment{either left or right, shift one bit to the left}
        		\If {$key < p.value$}
             		\State $p \gets$ \textsc{left}($p, base$)
        		\Else
            		\State $p \gets$ \textsc{right}($p, base$)
            		\State $bits \gets bits+1$ \Comment{right child colour is 1}
        		\EndIf
    	\EndWhile
    	\State $bits \gets bits >> 1$
	\Statex \Comment{pad the $bits$ to get the index of the child $\Delta$Node:}
    
    	\State $bits \gets bits << (maxDepth - depth)$ 
    	\Statex \Comment{follow nextRight if highKey is less than searched value:}
    	\If{${\Delta}Node.highKey <= key$}
     		\State ${\Delta}Node \gets {\Delta}Node.nextRight$
   	\Else
     		\State ${\Delta}Node \gets link[bits]$ \Comment{jump to child $\Delta$Node}
   	\EndIf
  \EndWhile
  \State \Return ${\Delta}Node$
\EndFunction

\end{algorithmic}
\caption{Search within pointer-less $\Delta$Node. This function will return the {\em leaf} $\Delta$Node containing
the searched key. From there, a simple binary search using \textsc{left} and \textsc{right} functions
is adequate to pinpoint the key location}
\label{lst:scannode-func}
\end{figure}

The inter-$\Delta$Node connection works by using the tree encoding. Here we gave colour
to each nodes using either $0$ or $1$ with a condition that adjacent nodes at the same level 
have different colours. 
With this the path traversed from the $root$ of a $\Delta$Node to reach 
any internal node will produce a bit-sequence of colours. This bit representation
will be translated into an array index that contains a pointer of another $\Delta$Node. 
We are using leaf-oriented tree and allocate a pointer array with
the length equal to the number of nodes in that $\Delta$Node. Figure \ref {lst:scannode-func}
illustrate how the inter-$\Delta$Node connection works in a pointer-less search function.

\subsubsection{Concurrent and balanced tree}

To solve the the worst case of consecutive numbers insertion, we adopt the structure and the 
algorithm of B-Link trees \cite{Lehman:1981:ELC:319628.319663} coined by Lehman and Yao. 
This tree is a highly concurrent variation of B-Tree which sometimes referred as Blink tree. 
We maintain the concept of  a dynamic-vEB $\Delta$Node and used these in place of 
the array nodes of B-Link trees.

To implement this, two new variables were added into the $\Delta$Nodes' metadata,
namely $nextRight$ or pointer to the right sibling $\Delta$Node and a $highKey$ value
that contains the upper-bound value of that specific $\Delta$Node. 
The insertion were done bottom-up and
searches were in top-down, left-to-right direction. 
With these additional variables and restrictions, 
\textsc{search} operations are guaranteed as \textit{wait-free} 
\cite{Lehman:1981:ELC:319628.319663}.
Bottom-up insertion also ensures 
that the tree is always in a balanced condition as mandated by Lemma \ref{lem:dynamic_vEB_search}. 
The same rebalancing procedures 
(Figure \ref{fig:treetransform}a) were also employed to ensure a $\Delta$Node
is full before it splits. The rebalancing also help
to clean-up the nodes marked for deletion, keeping $\Delta$Nodes always in good shape.

\subsection{Heterogenous balanced $\Delta$Tree}\label{sec:plDTv2}

The reason why we maintained the leaf-oriented (or external tree) layout for $\Delta$Node 
is to make sure the inter-$\Delta$Node mechanism works. Thus, 
it is not necessary for leaf $\Delta$Nodes or the last level
$\Delta$Nodes to have leaf-oriented layout since they don't have any child $\Delta$Nodes.
 
Based on this observation, we implement a special layout for the leaf $\Delta$Nodes,
making the balanced $\Delta$Tree is having heterogenous $\Delta$Nodes.
This special layout is using internal tree for the key nodes, therefore 
100\% more key nodes can fit into leaf $\Delta$Nodes
compare to the non-leaf $\Delta$Nodes given the same $UB$ limit.
To save space even more, we omit the array of pointers for intra-$\Delta$Node connection
in the leaf $\Delta$Nodes' metadata.

Stepping up to this improved version of $\Delta$Tree, we found that the efficiency of 
searches were greatly improved. Compared to original pointer-based $\Delta$Tree and balanced
DT, this heterogenous BDT delivered lower cache misses and more
efficient branching.


\subsection{Performance evaluation} \label{sec:perval_main}
To evaluate our conceptual ideas of dynamic-vEB implemented 
in the $\Delta$Trees (section \ref{sec:implementation}), 
balanced $\Delta$Tree (plDTv1) (section \ref{sec:plDTv1}), and heterogenous BDT
(plDTv2) (section \ref{sec:plDTv2}), we compare their 
performance with other prominent concurrent trees. 
The benchmark include the non-blocking binary search tree (NBBST) \cite{EllenFRB10}, 
concurrent AVL tree (AVLtree) \cite{BronsonCCO10}, concurrent red-black tree (RBtree) \cite{DiceSS2006}, and speculation friendly tree (SFtree) \cite{Crain:2012:SBS:2145816.2145837} from the Synchrobench benchmark \cite{synchrobench}. We also develop a concurrent version of the static vEB binary search tree in \cite{BrodalFJ02} using software transactional memory (STM). We utilise the 
GNU C Compiler's STM implementation from the version 4.9.1 for this tree and named it VTMtree. An optimised Lehman and Yao concurrent B-tree
\cite{Lehman:1981:ELC:319628.319663} implementation (CBTree) is also included in the benchmark.
The tree (sometimes known as B-link tree) is a highly-concurrent B-tree implementation 
and it is being used as the back-end in popular database systems such as PostgreSQL
\footnote{\url{https://github.com/postgres/postgres/blob/master/src/backend/access/nbtree/README}}. 
We use Pthread for concurrent threads and pin the threads to distinct available physical cores. We use GCC 4.9.1 with -O2 for all compilations. 

\begin{figure}[!htbp]\centering
\footnotesize
\input{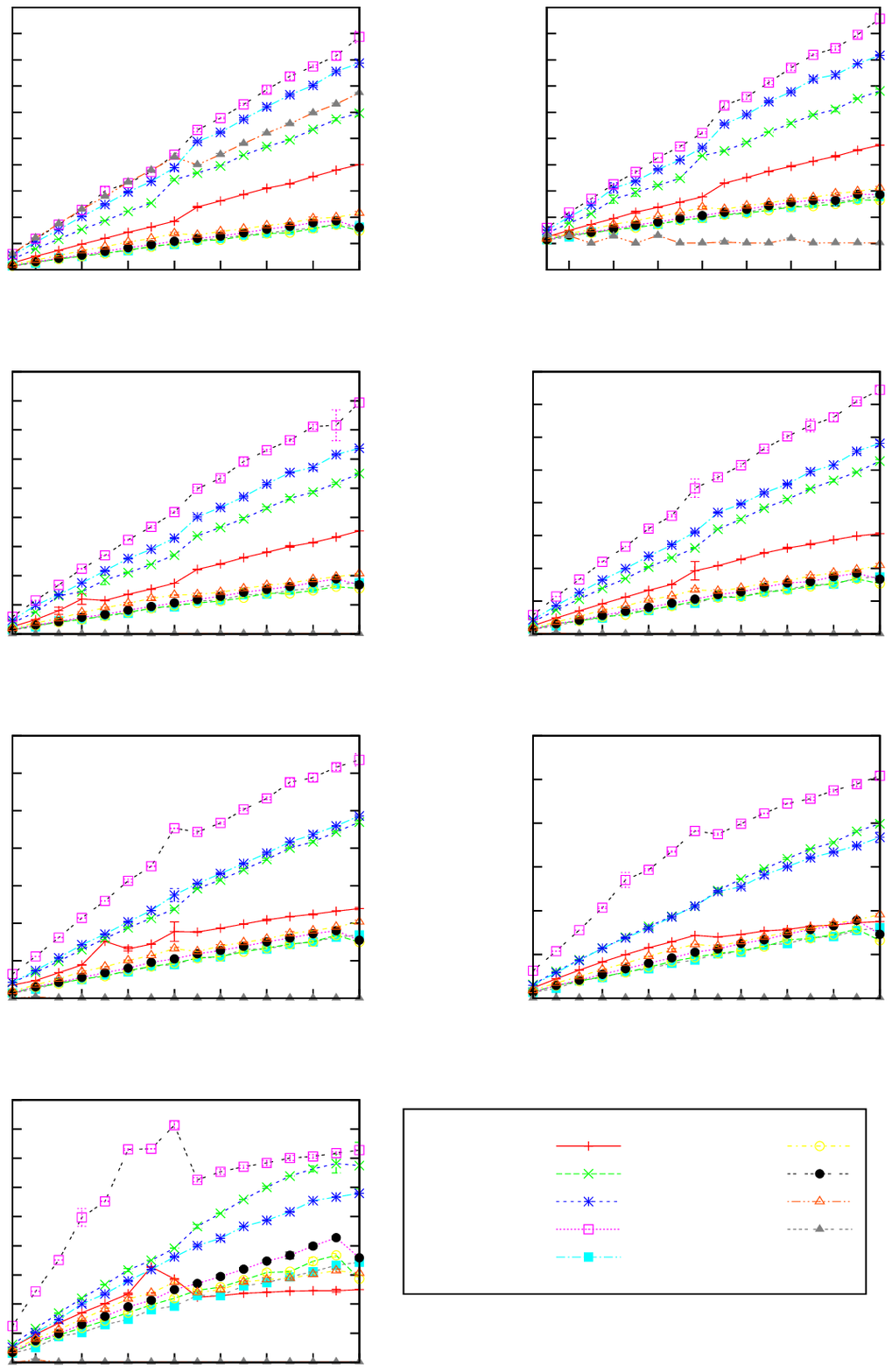}
\caption{Performance rate (operations/second) of the tested trees with 2,500,000 initial members. The y-axis indicates the rate of operations/second.}
\label{fig:perf-2500000}
\end{figure}

The base of the conducted experiment consists of running a series
of ($rep=5,000,000$) operations. Assuming we have $nr$ as the number of
threads, the maximum time for any of the threads to finish a 
sequence of $\frac{rep}{nr}$ operations is recorded. We
also define an update rate $u$ that translates to $upd = (u\%\times rep)$ number of
insert and delete operations and $src = (rep - upd)$ number of search operations. 
We conduct experiments based on the combinations of
update rate $u=\{0, 20, 50\}$ and the number of thread 
$nr=\{1, 2, \ldots, 16\}$.
Update rate of 0 means that only searching operations are conducted (100\% search), 
while update rate 50 indicates that 50\% insert and delete operations are being done
out of $rep$ operations. 
For each of the combination above, we pre-fill the tree with $(2^{22}-1)$ (or 4,194,303) random 
values before starting the benchmarks.

The initial size ($init$) of $(2^{22}-1)$ was chosen to simulate initial trees 
that partially fit into the last level cache (LLC).
All involved operations, namely search, insert, and delete invoked 
during the tests, use random
values $v \in (0, init \times 2], v \in \mathbb{N}$, as their parameter. 
Note that since VTMtree's static vEB layout is fixed, 
we set its layout size to $(2^{23}-1)$ for running the experiments. Namely, this setting is the best case for VTMtree since the memory allocated for its static vEB layout is large enough to accommodate all the values $v \in (0, init \times 2]$ and therefore its layout never needs to expand and rebuild during the experiments.
To make a fair
comparison, we set the $UB$ values of the $\Delta$Nodes and the CBTree's node-size to respective values so that each $\Delta$Node and each 
CBTree node will fit into the system page size of 4KB. 

\begin{figure}[!htbp] 
\centering
\footnotesize
\resizebox{0.8\columnwidth}{!}{ \includegraphics{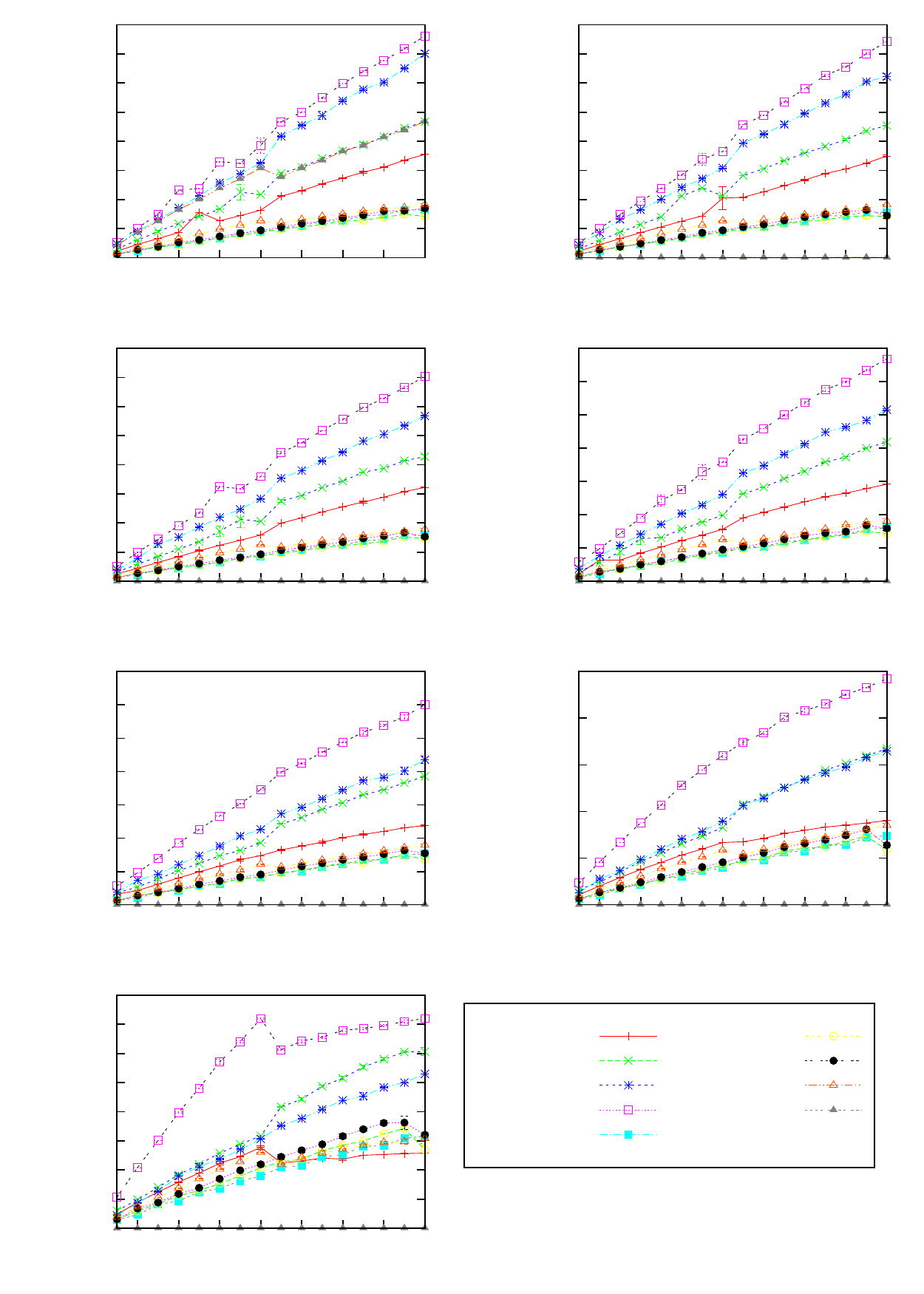}}
\caption{Performance comparison of the tested trees with 4,194,303 initial members using X86 based Intel Xeon system. 
There were 2 CPUs with 8 cores each and threads were pinned to the cores, therefore several "spikes" in performance could be 
observed for 9-thread tests.
The y-axis indicates operations/second.}
\label{fig:perf-4194303}
\end{figure}

The conducted experiments run on a dual Intel Xeon CPU E5-2670 machine, 
with total of 16 available cores. The machine has 32GB of memory, with a 2MB (8$\times$256KB) 
L2 cache and a shared 20MB L3 cache for each processor. 
The Hyperthread feature of the processor is turned off to 
make sure that the experiments only run on physical cores. Linux OS with Red Hat's 
kernel version 2.6.32-358 are installed in this system.
All performance result (in operations/second) for concurrent operations are calculated 
by dividing the number of iterations by the maximum time to finish the whole operations. 

\subsubsection{Experimental results} \label{sec:perfeval}

Before comparing the trees performance, there is 
one interesting thing to note from the benchmark results.
The x86 test system consists of 2 CPUs
with 8-cores each. Therefore we could see some "spikes" in performance
for couple of trees when it goes from 8-thread to 9-thread. This is expected though. For trees exploiting data-locality in cache (such as the VTMTree's 100\% search
and plDTv2's 50\% update in Figure \ref{fig:perf-4194303}), maintaining the cache-coherence between the two CPUs (when more than 8 threads are used) via shared memory reduces the benefit gained from cache-locality. Among the proposed trees, plDTv1 is up to 100\% faster 
than $\Delta$Tree for 100\% search and is up to 2.5x and 3.5x faster in 20\% and 50\% updates, respectively. Using map instead pointers and keeping the tree balanced
manage to lower the cache-misses of plDTv1 by 40\%. The plDTv2
is up to 5\% faster than plDTv1 in 100\%
searching. However in 20\% and 50\% updates,
plDTv2 is up to 40\% faster than the plDTv1 (cf.
Figure \ref{fig:perf-4194303}). It is because heterogeneous leaf $\Delta$Nodes that can hold
twice many keys, manage to lower the search time and $\Delta$Nodes' re-balancing overheads.
Unix {\it perf} utility shows that plDTv2 has up to 30\% less cache-misses 
and 20\% more efficient branching than plDTv1.  

\paragraph{$\Delta$Trees versus VTMtree} \label{sec:vtmversus}
It is expected that VTMtree is among the fastest in 100\% search 
along with the plDTv2 (Figure \ref{fig:perf-4194303}). 
As the cache-oblivious tree implementation, VTMtree is able to
exploit perfectly data-locality in all levels of memory. 
Our dynamic-vEB plDTv2 is the only contender
and even beats VTMtree by up to 20\% past the 8-thread mark. 
In $init=4,194,303$, the performance gap is even higher with 
plDTv2 leading by up to 30\%, after 9-thread.
The plDTv1 and $\Delta$Tree are both beaten by VTMtree
since it can only pack less data inside a memory page. CBTree and
other trees are not exploiting data-locality, which makes them slower
in 100\% search benchmark.

All other trees are demonstrating better performance compared to 
VTMtree whenever update operations are involved.
The bad performance of VTMtree's concurrent update is inevitable 
because of its static vEB tree layout. With this, the VTMtree needs to
always maintain a small height, which is done by {\em incrementally}
re-balancing different portions of its structure \cite{BrodalFJ02}. 
However in the worst case the whole tree must be blocked whenever 
a rebalance operation is being executed, blocking other operations as a result. 
While \cite{BrodalFJ02} explained that amortised cost for this is small, 
it will hold true only when implemented in the sequential fashion. 
In all variants of $\Delta$Tree, maintenance operations only block a $\Delta$Node, which is 
discernible in size compared to the whole tree.

\paragraph{$\Delta$Trees versus other trees}

In comparison with the other trees, the benchmark result in Figure 
\ref{fig:perf-2500000} and \ref{fig:perf-4194303} shows that the heterogenous BDT (plDTv2)
is the fastest among other trees. In 50\% update using a single CPU socket 
(8 threads), plDTv2 is up to 140\% faster than CBTree. 
plDTv1 and CBTree are 
trailing closely behind plDTv2. CBTree manages to outperform plDTv1 in higher 
update ratios because plDTv1 sometimes need to do rebalancing, which is
more expensive compared to array shifting in CBTree. However rebalancing
doesn't affect much the plDTv2 because its leaf $\Delta$Nodes are not leaf oriented.
Thus, the amortised cost of rebalancing is much lower compared to plDTv1. 

NBBST performs similarly with the $\Delta$Tree, 
mainly because the latter is modelled after NBBST for achieving concurrency. 
However being not a locality-aware structures makes NBBST unable 
to lead in the search-intensive benchmarks.
The balanced $\Delta$Tree (plDTv1) 
performs well only in the low-contention situations, as 
it is able to deliver good performance only up to 20\% update.

The good performance of $\Delta$Tree, plDTv1, and plDTv2 can 
be attributed to the dynamic vEB
layout that permits fast search. Also the fact that several 
$\Delta$Node can be concurrently updated and
restructured is also one of the leading factor over the static vEB layout. 
The CBTree layout, although fast and highly-concurrent, still suffers from
high branching operations, based on Unix perf profiling. CBTree
branching is up to 90\% more than plDTv1's and plDTv2's. Its cache references
is also 150\% higher than plDTv1's and plDTv2's, as expected in B-tree 
since block $B$ is not optimal for transfers between memory levels.

The software transactional memory (STM) based trees have a significant overhead in
maintaining transaction. Therefore their performances are the slowest.

\subsubsection{On the worst-case insertions}
One of the motivations in improving the $\Delta$Tree into balanced DT and subsequently 
heterogenous BDT, is to solve the poor performance of $\Delta$Tree in the worst-case insertions.
Inserting a sequence of increasing numbers to the $\Delta$Tree
will result in a linked-list of $\Delta$Nodes. 

Therefore we compare CBTree and plDTv2 for the worst-case insertions. 
GCC standard library {\it std::set} is included as the baseline. Starting from a blank 
tree, we insert a sequence of increasing number within $(0, 5,000,000]$ range using
single thread. 

\begin{figure}[!htbp] \centering
\resizebox{0.9\columnwidth}{!}{ 
\begin{tikzpicture}
  \begin{axis}
    [
    scaled x ticks=false,
    xbar, xmin=0,
    width=12cm, height=3.5cm, enlarge y limits=0.5,
    xlabel={operations/second},
    symbolic y coords={CBTree, plDTv2, std::set},
    ytick=data
    ]
    \addplot coordinates {(2113967,CBTree) (3011759,plDTv2) (2181831,std::set)};
  
    \end{axis}
\end{tikzpicture}
}
\caption{Worst case insertion of 5,000,000 increasing numbers}
\label{fig:worstcase}
\end{figure}
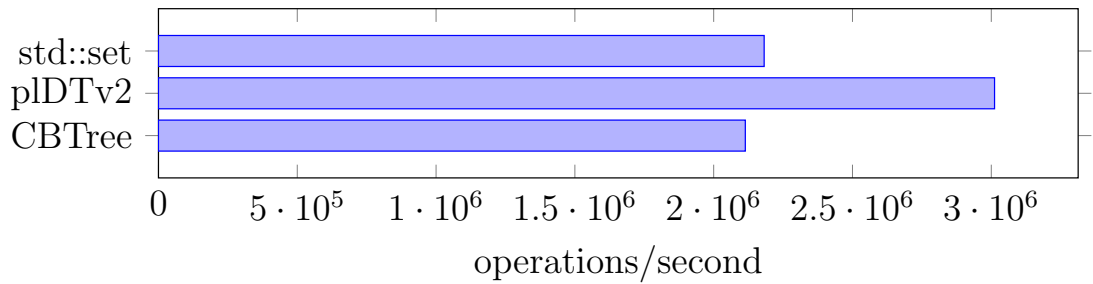

The result in Figure \ref{fig:worstcase} shows that plDTv2 is 50\% faster than 
CBTree and  {\it std::set} from GCC standard library. 
This test is done using the same x86 experimental system as in Section \ref{sec:perfeval}.

\subsubsection{Performance on different upper-bounds $UB$}

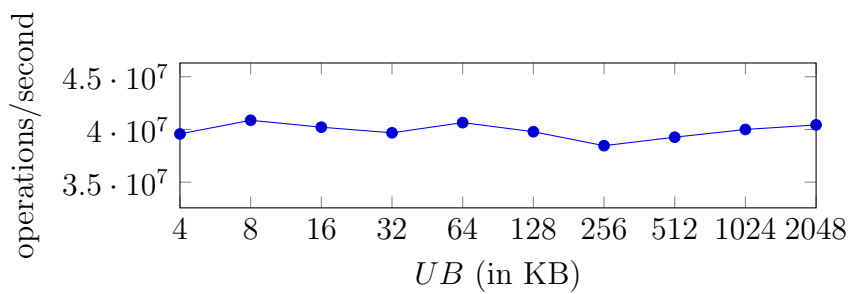
\begin{figure}[!hb] \centering
\resizebox{0.7\columnwidth}{!}{ 
\begin{tikzpicture}[baseline]
\begin{axis}
[
		ylabel={operations/second},
		scaled y ticks=false,
		width=10cm, height=3.5cm, enlarge y limits=0.5,
		xlabel={$UB$ (in KB)},
		ymin= 36000000,
		ymax= 42871934,
		xmin = 4,
		xmax = 2048,
		symbolic x coords={4, 8, 16, 32, 64, 128, 256, 512, 1024, 2048},
		xtick=data
]

\addplot+  plot [error bars/.cd, y dir=both, y explicit]
	coordinates{
	(4, 39577836.4116095) 	+- 	(2886, 2886)
	(8, 40871934.6049046) 	+- 	(1443, 1443)
	(16, 40214477.2117963) 	+- 	(1443, 1443)
	(32, 39682539.6825397) 	+- 	(0,0)
	(64, 40650406.504065) 	+- 	(0,0)
	(128, 39787798.4084881) +- 	(1443, 1443)
	(256, 38461538.4615385) +- 	(0,0)
	(512, 39267015.7068063) +- 	(1443, 1443)
	(1024, 40000000) 		+- 	(2500, 2500)
	(2048, 40431266.8463612)	+- 	(8036, 8036)
};
\end{axis}
\end{tikzpicture}}
\caption{Benchmarks of 100\% search using different sizes of $UB$. Tested on X86 platform.}
\label{fig:diff-UB}
\end{figure}

Dynamic vEB requires that an upper-bound $UB$ be specified or known in advance. 
One may argue that $UB$ is a fine-tuned value that will determine 
the performance of dynamic vEB trees. To get the conclusion of  
whether this is the case, we test the heterogeneous BDT using different upper-bounds $UB$, starting from 4KB (normal page size) up to 2MB (huge page size). 
This tree is filled with $(2^{22} -1)$ random values and time is recorded to conduct 5 million operations of 100\% search.

The result shows that the heterogeneous BDT is resilient to different upper-bounds $UB$ (cf. Figure \ref{fig:diff-UB}).
This is an expected result according to Lemma \ref{lem:dynamic_vEB_search}. In fact, $UB$ can
be as big as the whole tree and searching performance is still optimal, provided the meta-data 
(e.g. the tree map in Section \ref{sec:mapdesc}) occupies only a small fraction of the last level cache LLC.
 As mentioned in Section \ref{sec:relaxed-veb}, small $UB$ benefits concurrent tree updates.

\subsection{Energy consumption evaluation}\label{sec:energyeval}

\begin{figure}[!h] \centering
\includegraphics{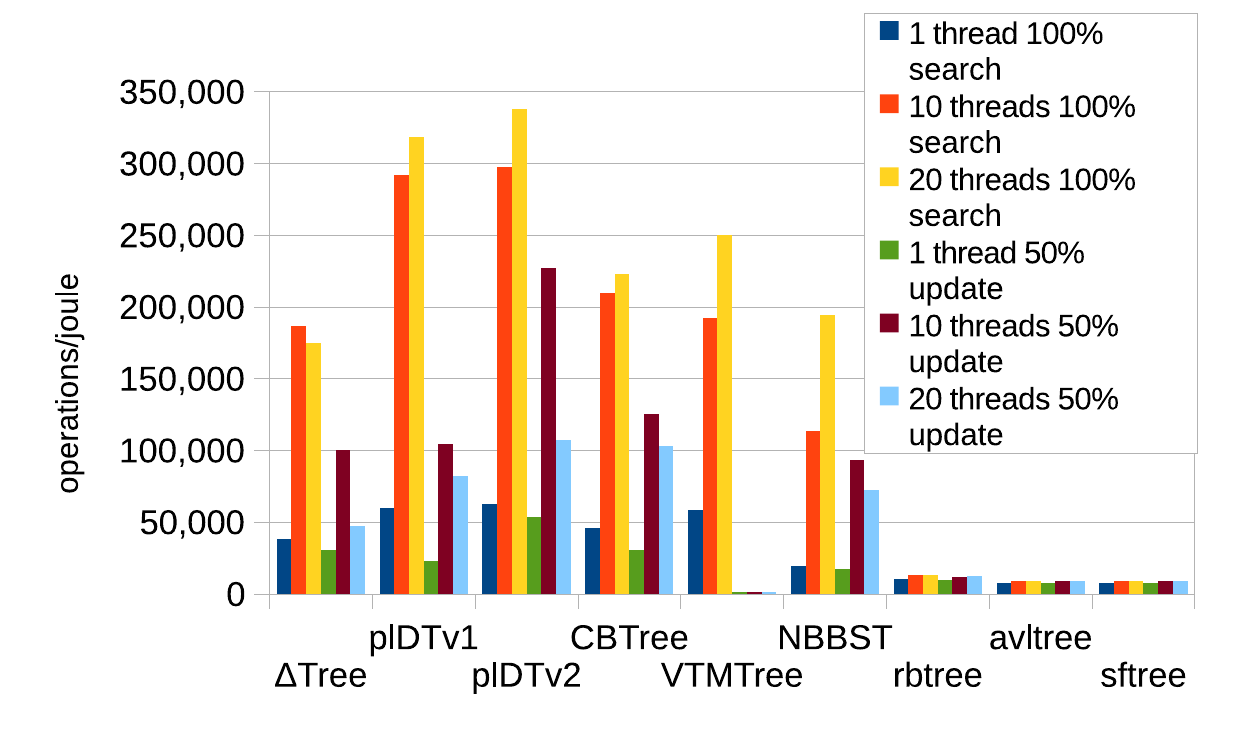}
\caption{Energy profile on X86 processor. Note that the energy efficiency goes down for some trees in 50\% update
on 20 threads (with 2 CPUs) because of the same reason discussed in Section \ref{sec:perfeval}}
\label{fig:X86-power}
\end{figure}

\begin{figure}[!hb] \centering
\resizebox{0.8\columnwidth}{!}{ \includegraphics{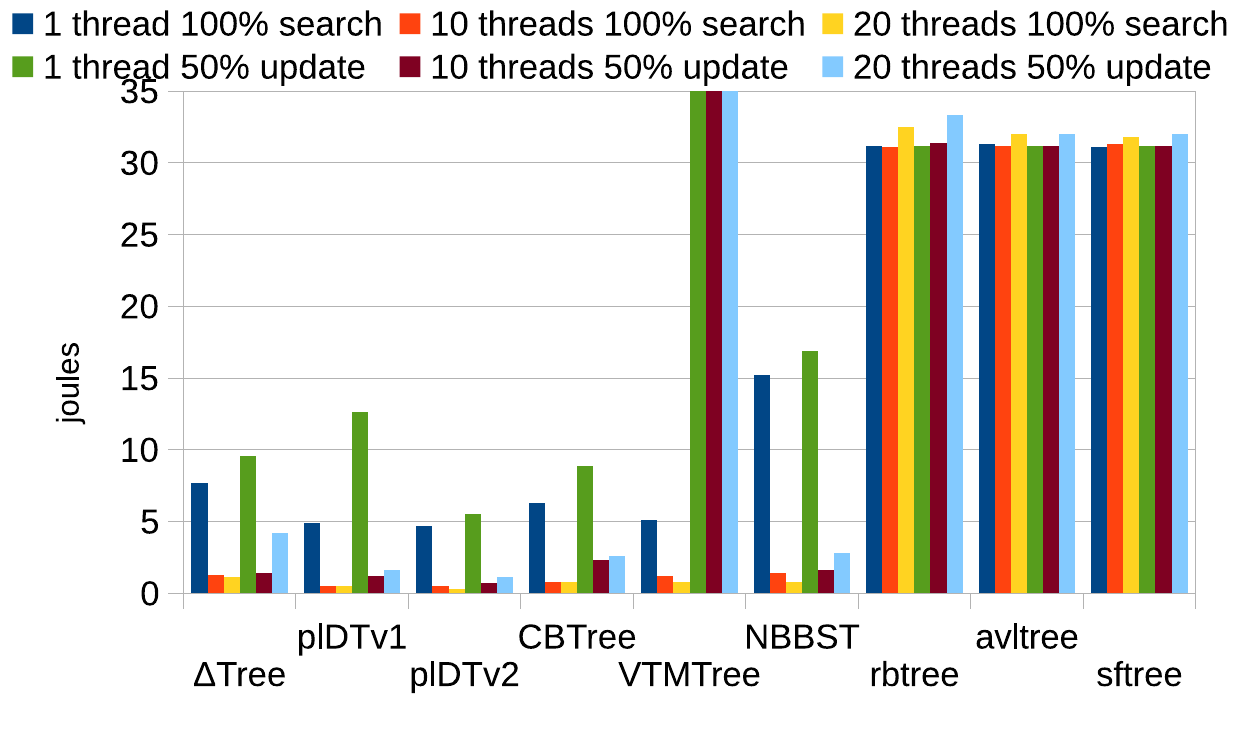}}
\caption{Memory (DRAM) energy profile on X86 platform. DRAM energy requirement goes up considerably for some trees in 50\% update
on 20 threads (with 2 CPUs) because of the cache-coherence mechanism (cf. Section \ref{sec:perfeval}). Also it is reflected in the 
memory write counters chart on Figure \ref{fig:mem-w}. Measured using Intel PCM.}
\label{fig:mem-energy}
\end{figure}

\begin{figure}[!h] \centering
\resizebox{0.8\columnwidth}{!}{ \includegraphics{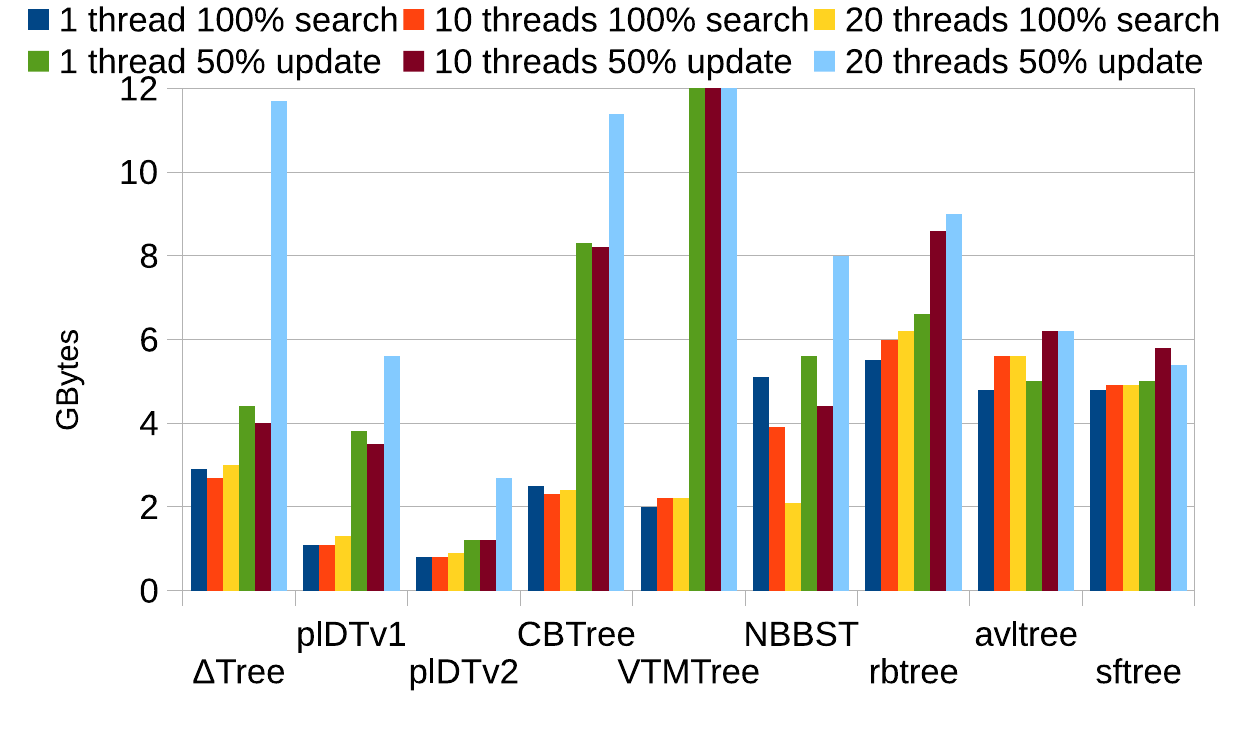}}
\caption{Amount data transferred between RAM and CPU for Read + Write operations. Measured using Intel PCM.}
\label{fig:mem-rw}
\end{figure}

\begin{figure}[!h] \centering
\resizebox{0.8\columnwidth}{!}{ \includegraphics{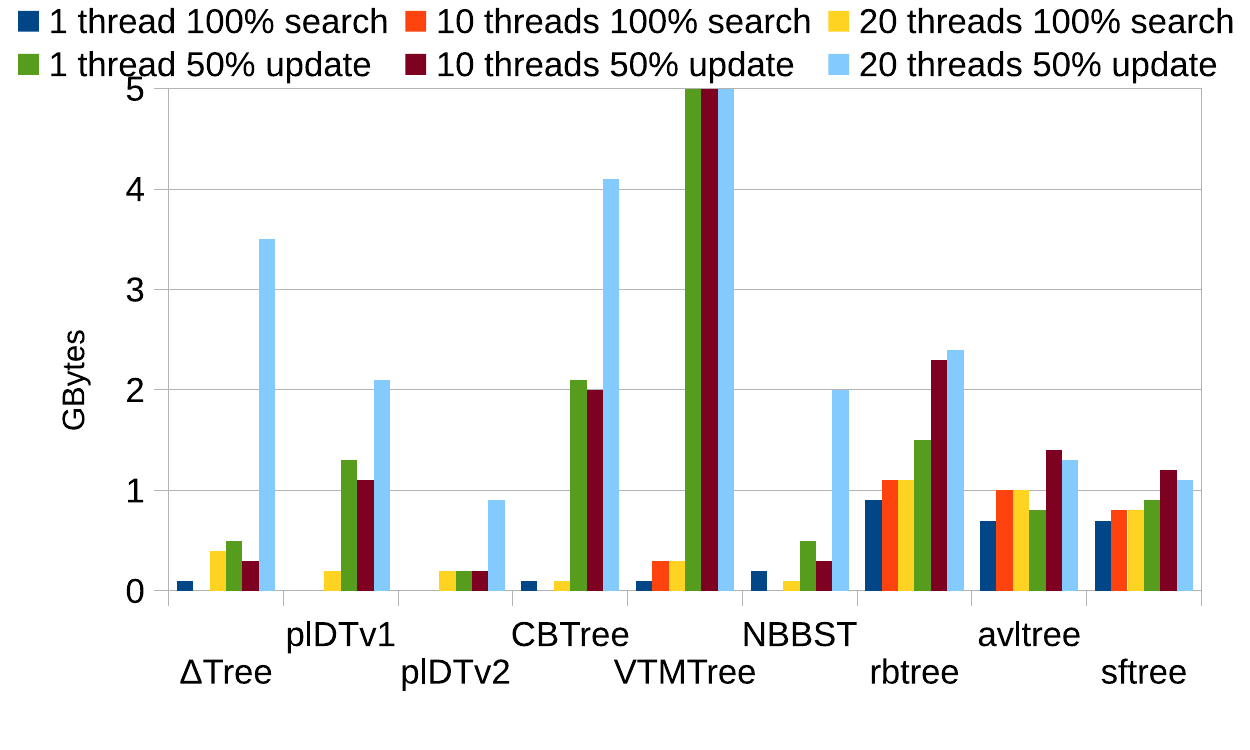}}
\caption{Amount data transferred between RAM and CPU for write-only operations as measured using Intel PCM. 
Data transfers goes up considerably for some trees in 50\% update
on 20 threads (with 2 CPUs) because of the cache-coherence mechanism. }
\label{fig:mem-w}
\end{figure}

To assess the energy consumption of the trees, the energy indicators 
are subsequently collected during specified benchmarks.
For these tests we use a specialised server with 2x Intel Xeon E5-2690 v2
for 20 total cores. We use Intel PCM library that can measure 
the energy for each CPU and DRAMs using the built-in CPU counters. The collected
energy measurements do not include the initialisations of trees.

In this experiment, we conduct 5 million operations 
of 100\% search and 50\% update on 
the trees. The trees are pre-filled with initial $(2^{22}-1)$ random 
values. A combination of minimum and maximum
available physical cores are used as one of the benchmarks parameters. The total energy
used for all CPUs and memory (in Joule) are divided by the number of
operations to produce operations/joule results.

The experimental result over the test system (Figure
\ref{fig:X86-power}) shows that using the dynamic-vEB layouts is able to reduce the CPU and memory
energy consumption. In the search-only benchmark, plDTv1, plDTv2, and $\Delta$Tree's actual energy
efficiency is comparable to that of the static vEB-based VTMtree and leads over
the other trees by up to 33\%. plDTv2 energy
efficiency is 80\% better than CBTree's and other trees in 50\% update using 10-threads. 
As expected, the VTMtree concurrent update results are very poor for the same reason
discussed on section \ref{sec:vtmversus}.

$\Delta$Trees are locality-aware trees, and their memory access pattern are more
efficient to the other trees. Profiling data showed that branching and cache-misses
are kept low. Assessments on memory transfers (Figure \ref{fig:mem-rw} and \ref{fig:mem-w})
suggests that all $\Delta$Tree
versions are transferring less data between RAM and CPU compared to other trees.
Thus the energy consumed by the $\Delta$Trees on DRAM operation
is also efficient  and comparable to the static-vEB's VTMTree in the searching only case (cf. Figure \ref{fig:mem-energy}). The reader is referred to \cite{UmarAH13} for more details about $\Delta$Trees.

\subsection{Conclusion}
We have introduced a new {\em relaxed cache oblivious} model that enables high
parallelism while maintaining the key feature of the original cache oblivious
(CO) model \cite{Frigo:1999:CA:795665.796479}: an algorithm that is optimal in terms of data movement 
for a simple two-level memory is asymptotically optimal for an unknown multilevel memory. 
This feature facilitates the development of fine-grained locality-aware algorithms for 
deep memory hierarchy in modern architectures as desired by energy efficient computing \cite{Dally11}. 
Unlike the original CO model, the
relaxed CO model assumes a known upper bound on unknown memory block sizes $B$
of multilevel memory systems.

Based on the relaxed CO model, we have developed a novel {\em dynamic} van Emde
Boas (dynamic vEB) layout that makes the vEB layout suitable for
highly-concurrent data structures with update operations. The dynamic vEB
supports dynamic node allocation via pointers while maintaining the optimal
search cost of  $O(\log_B N)$ memory transfers for vEB-based trees of size
$N$ without knowing memory block size $B$. 

Using the dynamic van Emde Boas layout, we have developed a pointer-based $\Delta$Tree, 
a balanced $\Delta$Tree, and a heterogenous version of the latter that
support both high concurrency and fine-grained data locality. Both pointer-based $\Delta$Tree
as well as balanced DT and heterogenous BDT {\em Search} operations are wait-free. 
Only $\Delta$Tree has non-blocking {\em Insert} and {\em Delete}
operations. All 3 versions of $\Delta$Tree are relaxed cache oblivious: the expected memory
transfer costs of its {\em Search, Delete} and {\em Insert} operations are  
$O(\log_B N)$, where $N$ is the tree size and $B$ is unknown memory block size in
the ideal cache model \cite{Frigo:1999:CA:795665.796479}. Our experimental
evaluation comparing the $\Delta$Trees with non-blocking binary search tree of \cite{EllenFRB10}, concurrent AVL tree \cite{BronsonCCO10}, concurrent red-black tree \cite{DiceSS2006}, and speculation-friendly trees \cite{Crain:2012:SBS:2145816.2145837} from the the Synchrobench benchmark \cite{synchrobench}, and the 
highly-concurrent B-tree of \cite{Lehman:1981:ELC:319628.319663} has shown that the best version of $\Delta$Tree achieves the best performance and 
the highest energy efficiency.




\leaveout{ 

\remind{Remove if empty.}
} 
\section{Conclusion} \label{sec:Conclusion}
In this work, we have presented our current results on the investigation and modeling of the trade-off between energy and performance as follows.
\begin{itemize}

\item A new power model for the Movidius Myriad platform has been
  proposed. The model can predict power consumption of our
  micro-benchmarks with $\pm 4\%$ accuracy compared to the measured
  values on the real platform. The new power model confirms that the
  dynamic power consumption is proportional to the number of SHAVE
  (Streaming Hybrid Architecture Vector Engine) processors used.

\item Inspired by EXCESS D1.1~\cite{D1.1}, a new version of the energy
  model for the CPU-based platform has been developed. This new model
  decomposes the power into static, active and dynamic power, and
  considers power as the sum of power from CPU, main memory and
  uncore. The model parameters have been derived and evaluated;
  dynamic powers roughly depend on nature of operation, amount
  of memory accessed per unit of time and remote accesses,
  respectively for CPU, memory and uncore.

\item We have made a first step towards realistic application by
  modeling the performance and power dissipation of synthetic
  applications, whose design is based on concurrent queues. We have
  exhibited a small set of parameters that rules both performance and
  power consumption of the application and discriminates the different
  queue implementations. It leads to a good prediction of those two
  key metrics on the whole space of study, while needing only a few
  measurements from this space.

\item Several concurrent queue designs have been transferred to
  Myriad1 platform. Mutex with two locks implementation is the fastest
  and most scalable since it provides maximum concurrency until 8
  SHAVEs.  With 4 SHAVEs, FIFO-based implementation performs well and
  is 28.3\% faster than the mutex with two locks. Shared variable
  based implementation has the worst performance.  In terms of power,
  SHAVE FIFO-based communication method is the most energy
  efficient. When contention is high and stall is common, it is power
  efficient to avoid spinning and set SHAVEs to stall. Communication
  via shared variables consumes more power because spinning on a
  memory location is energy inefficient, even when the spinning
  happens in a local CMX slice.

\item Another data structure that has been investigated and analyzed
  in this work is concurrent search tree. Our new concurrent search
  trees show the improvement on performance and energy
  consumption. Based on experimental evaluations, our new concurrent
  search trees called $\Delta$Trees that are up to 140\% faster and
  80\% more energy efficient than the traditional search trees.
\end{itemize}

In the next steps of this work, WP2 aims to identify other essential
concurrent data structures and algorithms for inter-process
communication in HPC and embedded computing and focus on customizing
them. We will exploit common data-flow patterns to create a
generalized communication abstraction with which application designers
can easily create and exploit the customization for the data-flow
patterns. The results will also constitute a white-box methodology for
tuning energy efficiency and performance of concurrent data structures
and algorithms, and programming abstractions with which application
designers can easily create and exploit the customization for common
data-flow patterns. The novel concurrent data structures and
algorithms will constitute libraries for inter-process communication
and data sharing on EXCESS platforms.


\newpage

\bibliographystyle{plain}
\bibliography{./D2.1_related_papers,../WP6-bibtex/excess,../WP6-bibtex/related-papers,../WP6-bibtex/peppher-related}

\section*{Glossary}

\begin{flushleft}
\begin{tabular}{lp{12cm}}
\textbf{BRU}    &  Branch Repeat Unit (on SHAVE processor) \\
\textbf{CAS}    &  Compare-and-Swap instruction \\
\textbf{CMX}    &  Connection MatriX on-chip (shared) memory unit, 128KB (Movidius Myriad) \\
\textbf{CMU}    &  Compare-Move Unit (on SHAVE processor) \\
\textbf{Component} & 1. [hardware component] part of a chip's or motherboard's 
  circuitry; \ 2. [software component] encapsulated and annotated reusable
  software entity with contractually specified interface and
  explicit context dependences only, subject to third-party (software) composition.\\
\textbf{Composition}    & 1. [software composition] Binding a call to a 
  specific callee (e.g., implementation variant of a component) and allocating
  resources for its execution; \ 2. [task composition] Defining a macrotask and
  its use of execution resources 
  by internally scheduling its constituent tasks in serial,
  in parallel or a combination thereof. \\
  
\textbf{CPU}    &  Central (general-purpose) Processing Unit\\

\textbf{uncore}    &  including the ring interconnect, shared cache, integrated memory controller, home agent, power control unit, integrated I/O module, config Agent, caching agent and Intel QPI link interface \\ 
\textbf{CTH}    &  Chalmers University of Technology \\
\textbf{DAQ}    &  Data Acquisition Unit \\
\textbf{DCU}    &  Debug Control Unit (on SHAVE processor) \\
\textbf{DDR}    &  Double Data Rate Random Access Memory \\
\textbf{DMA}    &  Direct (remote) Memory Access \\
\textbf{DRAM}   &  Dynamic Random Access Memory \\
\textbf{DSP}    &  Digital Signal Processor \\
\textbf{DVFS}   &  Dynamic Voltage and Frequency Scaling \\
\textbf{ECC}    &  Error-Correcting Coding \\
\textbf{EXCESS} &  Execution Models for Energy-Efficient Computing Systems\\
\textbf{GPU}    &  Graphics Processing Unit\\
\textbf{HPC}    &  High Performance Computing\\
\textbf{IAU}    &  Integer Arithmetic Unit (on SHAVE processor) \\
\textbf{IDC}    &  Instruction Decoding Unit (on SHAVE processor) \\
\textbf{IRF}    &  Integer Register File (on SHAVE processor) \\
\textbf{LEON}    &  SPARCv8 RISC processor in the Myriad1 chip\\
\textbf{LIU}    &  Link\"oping University \\
\textbf{LLC}    &  Last-level cache\\
\textbf{LSU}    &  Load-Store Unit (on SHAVE processor) \\
\textbf{Microbenchmark} & Simple loop or kernel developed to measure one or few properties of the underlying architecture or system software\\
\textbf{PAPI}   &  Performance Application Programming Interface\\
\end{tabular}
\end{flushleft}

\newpage 

\begin{flushleft}
\begin{tabular}{lp{12cm}}
\textbf{PEPPHER} &  Performance Portability and Programmability for Heterogeneous Many-core Architectures. FP7 ICT project, 2010-2012, www.peppher.eu \\
\textbf{PEU}    &  Predicated Execution Unit (on SHAVE processor) \\
\textbf{Pinning} &  [thread pinning] Restricting the operating system's CPU scheduler in order to map a thread to a fixed CPU core \\
\textbf{QPI}    &  Quick Path Interconnect\\
\textbf{RAPL}   &  Running Average Power Limit energy consumption counters (Intel)\\
\textbf{RCL}   &  Remote Core Locking (synchronization algorithm)\\
\textbf{SAU}    &  Scalar Arithmetic Unit (on SHAVE processor) \\
\textbf{SHAVE}  &  Streaming Hybrid Architecture Vector Engine (Movidius) \\
\textbf{SoC}    &  System on Chip \\
\textbf{SRF}    &  Scalar Register File (on SHAVE processor) \\
\textbf{SRAM}   &  Static Random Access Memory \\
\textbf{TAS}    &  Test-and-Set instruction\\
\textbf{TMU}    &  Texture Management Unit (on SHAVE processor) \\
\textbf{USB}    &  Universal Serial Bus \\
\textbf{VAU}    &  Vector Arithmetic Unit (on SHAVE processor) \\
\textbf{Vdram}  &  DRAM Supply Voltage \\
\textbf{Vin}    &  Input voltage level  \\
\textbf{Vio}    &  Input/Output voltage level  \\
\textbf{VLIW}   &  Very Long Instruction Word (processor) \\
\textbf{VLLIW}  &  Variable Length VLIW (processor) \\
\textbf{VRF}    &  Vector Register File (on SHAVE processor) \\
\textbf{Wattsup}&  Watts Up .NET power meter \\
\textbf{WP1}   &  Work Package 1 (here: of EXCESS) \\
\textbf{WP2}   &  Work Package 2 (here: of EXCESS) \\
\end{tabular}
\end{flushleft}

\end{document}